\documentclass{theoretics}

\title{Unifying the Three Algebraic Approaches to the CSP via Minimal Taylor Algebras}

\ThCSauthor[Prague]{Libor Barto}{libor.barto@gmail.com}[https://orcid.org/0000-0002-8481-6458]
\ThCSauthor[none]{Zarathustra Brady}{notzeb@gmail.com} 
\ThCSauthor[SimonFraser]{Andrei Bulatov}{abulatov@sfu.ca}[https://orcid.org/0000-0002-5516-1704]
\ThCSauthor[Krakow]{Marcin Kozik}{marcin.kozik@uj.edu.pl}[https://orcid.org/0000-0002-1839-4824]
\ThCSauthor[Prague]{Dmitriy Zhuk}{zhuk@karlin.mff.cuni.cz}[https://orcid.org/0000-0003-0047-5076]

\ThCSaffil[Prague]{Department of Algebra, Faculty of Mathematics and Physics, Charles University, Prague, Czechia}
\ThCSaffil[none]{No affiliation}
\ThCSaffil[SimonFraser]{School of Computing Science, Simon Fraser University, Burnaby BC, Canada}
\ThCSaffil[Krakow]{Theoretical Computer Science Department,
Faculty of Mathematics and Computer Science,
Jagiellonian University, Krakow, Poland}

\ThCSthanks{
Part of this work appeared in LICS'21 paper
\emph{Minimal Taylor Algebras as a Common Framework for the Three Algebraic Approaches to the CSP}.

 Libor Barto and Dmitriy Zhuk were funded by the European Union (ERC, CoCoSym, 771005) and (ERC, POCOCOP,
101071674). Views and opinions expressed are however those of the authors only and do not necessarily reflect those of the
European Union or the European Research Council Executive Agency. Neither the European Union nor the granting authority
can be held responsible for them.
Dmitriy Zhuk was also supported by HSE University Basic Research Program.
Zarathustra Brady has received funding from the NSF Mathematical Sciences Postdoctoral Research Fellowship under Grant No. (DMS-1705177).
Andrei Bulatov has received funding from an NSERC Discovery grant.
Marcin Kozik was partially funded by the National Science Centre, Poland, under the  Weave-UNISONO  call  in  the  Weave  programme  2021/03/Y/ST6/00171.
}

\ThCSshortnames{L.\ Barto, Z.\ Brady, A.\ Bulatov, M.\ Kozik and S.\ Zhuk }
\ThCSshorttitle{Unifying the Three Algebraic Approaches to the CSP via Minimal Taylor Algebras}

\ThCSyear{2024}
\ThCSarticlenum{14}
\ThCSreceived{May 23, 2023}
\ThCSrevised{Feb 21, 2024}
\ThCSaccepted{Mar 25, 2024}
\ThCSpublished{May 15, 2024}
\ThCSdoicreatedtrue
\ThCSkeywords{Constraint Satisfaction Problem, Universal Algebra, Minimal Taylor Algebra, Polymorphism}

\usepackage{graphicx}

\usepackage{comment}

\usepackage{mathtools}
\usepackage{mathrsfs}
\usepackage{thm-restate}

\newcommand{\zee}[1]{\mathbb{Z}/{#1}}
\newcommand{\relstr}[1]{\mathbb #1}

\newcommand{\algA}{\mathbf A}
\newcommand{\algB}{\mathbf B}
\newcommand{\alg}[1]{\mathbf #1}
\newcommand{\tuple}[1]{{\mathbf #1}}
\newcommand{\abs}{\trianglelefteq}
\newcommand{\st}{\mathrel{:}}
\newcommand{\sd}{\leq_{sd}}
\newcommand{\setsd}{\subseteq_{sd}}
\newcommand{\types}[1]{\mathrm{#1}}

\newcommand{\clone}[1]{\mathscr{#1}}
\newcommand{\relclone}[1]{\mathscr{#1}}

\DeclareMathOperator{\Sg}{Sg}
\DeclareMathOperator{\Inv}{Inv}
\DeclareMathOperator{\maj}{maj}

\DeclareMathOperator{\winner}{winner}
\DeclareMathOperator{\paper}{paper}
\DeclareMathOperator{\rock}{rock}
\DeclareMathOperator{\scissors}{scissors}
\DeclareMathOperator{\proj}{proj}

\DeclareMathOperator{\Clo}{Clo}

\newcommand{\ooo}[1]{#1}

\newcommand{\dnote}[1]{}
\newcommand{\lnote}[1]{}
\newcommand{\mnote}[1]{}
\newcommand{\anote}[1]{}
\newcommand{\znote}[1]{}

\newcommand{\sss}[1]{\rightsquigarrow_{#1}}
\newcommand{\sssB}{\sss{B}}
\newcommand{\sssC}{\sss{C}}
\newcommand{\sssD}{\sss{D}}

\newcommand{\rrr}{}
 \newenvironment{ThCSrestatable}[4]
 {\restatable{#1}{#2}\ooo{#3} \rrr\label{#4}}
 {\endrestatable}

\newenvironment{ThCSrestatablen}[4] 
{\restatable{#1}{#2}(#3) \rrr\label{#4}}
{\endrestatable}

\begin{document}

\maketitle

\begin{abstract}
This paper focuses on the algebraic theory underlying the study of the complexity and the algorithms for the Constraint Satisfaction Problem (CSP). We unify, simplify, and extend parts of the three approaches that have been developed to study the CSP over finite templates --  absorption theory that was used to characterize CSPs solvable by local consistency methods (JACM'14), and Bulatov's and Zhuk's theories that were used for two independent proofs of the CSP Dichotomy Theorem  (FOCS'17, JACM'20). 

As the first contribution we present an elementary theorem about primitive positive definability and use it to obtain the starting points of Bulatov's and Zhuk's proofs as  corollaries. 
As the second contribution we propose and initiate a systematic study of minimal Taylor algebras. 
This class of algebras is broad enough that it suffices to verify the CSP Dichotomy Theorem on this class only,  but still is unusually well behaved. In particular, many concepts from the three approaches coincide in this class, which is in striking contrast with the general setting.

We believe that the theory initiated in this paper will eventually result in a simple and more natural proof of the Dichotomy Theorem that employs a simpler and more efficient algorithm, and will help in attacking complexity questions in other CSP-related problems.

\end{abstract}

\section{Introduction}%
\label{sec:introduction}

The Constraint Satisfaction Problem (CSP) has attracted much attention from researchers 
in various disciplines. One direction of the CSP research has been greatly motivated by the so-called Dichotomy Conjecture of Feder and Vardi \cite{Feder93:monotone,Feder98:monotone} that concerns the computational complexity of CSPs over finite relational structures. 
The \emph{Constraint Satisfaction Problem over a finite relational structure} $\relstr{A}$ of finite signature (also called a \emph{template}), in its logical formulation, is the problem of deciding the validity of a given \emph{primitive positive sentence} (\emph{pp-sentence}), i.e., a sentence that is an existentially quantified conjunction of atomic formulas over $\relstr{A}$ -- the \emph{constraints}. Examples of problems in this class include satisfiability problems, graph coloring problems, and solving systems of equations over finite algebraic structures (see \cite{Hell04:homomorphism,Jeavons97:closure,Barto17:polymorphisms,Larose06:taylor}). The CSP is also ubiquitous in artificial intelligence \cite{Dechter03:processing}. 

A classic result in the field is a theorem by Schaefer~\cite{Schaefer} that completely classifies the complexity of CSPs over relational structures with a two-element domain, so-called \emph{Boolean structures}, by providing a dichotomy theorem: each such a CSP is either  solvable in polynomial time or is NP-complete. The Dichotomy Conjecture of Feder and Vardi states that Schaefer's result extends to arbitrary finite domains.  This conjecture inspired a very active research program
in the last 20 years, culminating in a positive resolution independently obtained by Bulatov~\cite{Bulatov17:dichotomy} and Zhuk~\cite{Zhuk17:proof,Zhuk20:dichotomy}. The exact borderline between tractability and hardness can be formulated as follows~\cite{Bulatov05:classifying,Barto18:wonderland,Barto17:polymorphisms}. 

  \begin{theorem}[Dichotomy Theorem] \label{dichotomy}
    Let $\relstr{A}$ be a finite relational structure over a finite signature.
    \begin{itemize}
        \item If every finite structure is homomorphically equivalent to a finite structure pp-interpretable in $\relstr{A}$, then the CSP over $\relstr{A}$ is NP-hard,
        \item otherwise it is solvable in polynomial time.
    \end{itemize}
  \end{theorem}

  It was already recognized in Schaefer's work (in fact, it was the basis of his approach) that the complexity of a CSP depends only on the set of relations that are \emph{pp-definable} (i.e., definable by a primitive positive formula) from the template. Such sets of relations are now usually referred to as \emph{relational clones}. 
  The impetus of rapid development in the area after Feder and Vardi's seminal work~\cite{Feder98:monotone} was a series of papers \cite{Jeavons97:closure,Jeavons98:algebraic} that brought to attention and applied a Galois connection between operations and relations studied in the sixties~\cite{Geiger68:pol-inv,Bodnarcuk69:pol-inv}, which gives a bijective correspondence between relational clones and \emph{clones} -- sets of term operations of algebras. 
  
  One way to phrase this core fact is as follows: for any finite algebra $\alg A$, its set of invariant relations (\emph{subuniverses of powers} or \emph{subpowers} in algebraic terminology) is always a relational clone; every relational clone is of this form; and two algebras have the same relational clone of subpowers if and only if they have the same set of term operations (see Subsection~\ref{subsec:alg_approach}). For instance, a Boolean CSP, say over the domain $\{0,1\}$,  is solvable in polynomial time if and only if the relations of the template are subpowers of one of four types of algebras -- an algebra with a single constant operation, a semilattice, the majority algebra, or the affine Mal'cev algebra of $\zee{2}$ (see Subsection~\ref{subsec:boolean}).

  This connection between relations and operations allowed researchers to apply techniques from Universal Algebra. Application of these techniques became known as the \emph{algebraic approach to the CSP}, although one may argue that the name misses the point a little -- the success of the approach lies mostly in combining and moving back and forth between the relational and algebraic side, and this is the case for this paper as well.
  The general theory of the CSP was further refined in subsequent papers~\cite{Bulatov05:classifying,Barto18:wonderland} 
  and turned out to be an efficient tool in other types of constraint problems including the
Quantified CSP \cite{Borner09:complexity,Carvalho17:complexity,Zhuk22:qcsp}, 
the Counting CSP \cite{Bulatov07:towards,Bulatov13:complexity}, some optimization problems, e.g.\ the Valued 
CSP~\cite{Krokhin17:complexity} and robust approximability \cite{Barto16:robustly},
infinite-domain CSPs \cite{Bodirsky18:discrete-temporal,bodirskyBook},
related promise problems such as 
``approximate coloring'' and the Promise CSP~\cite{Brakensiek18:promise,Bulin19:algebraic}, and many others.

  One useful technical finding of~\cite{Bulatov05:classifying} is that every CSP is equivalent to a CSP over an \emph{idempotent template}, i.e. a template that contains all the singleton unary relations. This allows us to use parameters in pp-definitions 
  and omit homomorphic equivalence in the first item of Theorem~\ref{dichotomy}. On the algebraic side, this allows us to concentrate on so-called idempotent algebras (see Subsection~\ref{subsec:prelim_algebra}). Another important contribution of that paper was a conjecture postulating, for idempotent structures, the exact borderline between polynomial solvability and NP-hardness, which coincides with the borderline stated in Theorem~\ref{dichotomy}. The hardness part was already dealt with in the same paper  and what was left was the tractability part. Within the realm of idempotent structures, the algebras corresponding to the second item of Theorem~\ref{dichotomy} are so-called Taylor algebras (see Subsection~\ref{subsec:prelim_taylor}). The following theorem is therefore the core of the two  proofs of the Dichotomy Conjecture.

  \begin{theorem}[\cite{Bulatov17:dichotomy,Zhuk17:proof,Zhuk20:dichotomy}] \label{dichotomy_real}
   Let $\relstr{A}$ be an idempotent structure. If there exists an idempotent Taylor algebra $\alg A$ such that all relations in $\relstr{A}$ are subpowers of $\alg A$, then the CSP over $\relstr{A}$ is solvable in polynomial time.
  \end{theorem}
  
  Partial results toward Theorem~\ref{dichotomy_real} include dichotomies for various classes of relational structures and algebras (e.g. the class of 3-element algebras~\cite{Bulatov06:three-element} and the class of structures containing all unary relations~\cite{Bulatov11:conservative}), understanding of the limits of algorithmic techniques (e.g. local consistency methods~\cite{Barto14:local} and finding generators for the set of solutions~\cite{Idziak10:few}), and finding potentially useful characterizations of Taylor algebras (e.g. by means of weak near-unanimity operations~\cite{Maroti08:weak-nu} and by means of cyclic operations~\cite{Barto12:cyclic}). The papers~\cite{Barto14:local} and~\cite{Barto12:cyclic} initiated a technique which is now referred to as the \emph{absorption theory}~\cite{Barto17:absorption}. Absorption theory is one of the fruits of CSP-motivated research which also impacted other CSP-related problems as well as universal algebra~(e.g. \cite{Barto18:CM})  and it is one of the three theories this paper is concerned with.

  Bulatov and Zhuk in their resolution of the Dichotomy Conjecture (and their prior and subsequent work) developed novel techniques, which we refer to as \emph{Bulatov's theory} and \emph{Zhuk's theory} in this paper.  
  These theories are (understandably) mostly focused on the task at hand, to prove Theorem~\ref{dichotomy_real}, and  as such have several shortcomings. First, some of the new concepts are still evolving as the need arises, and they do not yet feel elegant and settled. Moreover, the theories are technically complex which makes it difficult to master them and to apply them in different contexts. This is best witnessed by the absence of results which employ these theories from different authors (needless to say they have already witnessed their potential). 
  Second, they both employ the following trick. Instead of studying a general, possibly wild Taylor algebra, one can first tame it by taking a certain \emph{Taylor reduct} -- an algebra whose operations are only some of the term operations but which is still Taylor.
  Taking reducts does not result in any loss of generality in Theorem~\ref{dichotomy_real}, since reducts keep all the original invariant relations, so proving tractability for a reduct is sufficient for tractability for the original problem. However, taking reducts does result in loss of generality of the theory, and it is not yet clear to which  natural classes of algebras the theories apply. Moreover, these reducts are different in the two approaches.
  Third, connections between Bulatov's and Zhuk's theories were not understood at all. While Zhuk's theory and absorption theory at least had some concepts in common, Bulatov's theory seemed quite orthogonal to the rest.

  The contributions of this paper unify, simplify, and extend parts of these three theories, making them, we hope, more accessible and reducing the prerequisites for the dichotomy proofs.
  In particular, we initiate a systematic study of \emph{minimal Taylor algebras}, i.e., those algebras that are Taylor but such that none of their proper reducts is Taylor. Thus, we employ the above trick to the extreme and study, in a sense, the tamest algebras or, in other words,  ``hardest'' tractable CSPs. This restriction, on the one hand, limits the scope of the theory but, on the other hand, gives us a framework in which the three theories do not look separate at all anymore. Indeed, the authors find the extent to which the notions of the three theories simplify and unify in minimal Taylor algebras to be truly striking.
  Even though our results do not cover some advanced parts of the three theories, 
  we believe that they have the potential to evolve into one coherent theory of finite algebras that would make the CSP Dichotomy Theorem an exercise (albeit hard) and that would have applications well beyond constraint problems. 
  
  The contributions can be divided into two groups, results for (all finite) Taylor algebras stated in Section~\ref{sec:taylor} and results for minimal Taylor algebras in~Sections~\ref{sec:mtalgebras} and \ref{sec:omitting_types}. We now describe them in more detail together with more background.

\subsection{Taylor algebras}

The central concept in absorption theory is that of absorbing subuniverses introduced formally in Subsection~\ref{subsec:prelim_absorption}. These are invariant subsets of algebras with an additional property resembling ideals in rings. A fundamental theorem, the \emph{absorption theorem}~\cite{Barto12:cyclic}, shows that nontrivial absorbing subuniverses in Taylor algebras exist under rather mild conditions and this fact makes the theory applicable in many situations. For instance, the strategy in~\cite{Barto14:local} to provide a global solution to a locally consistent instance is to propagate local consistency into proper absorbing subuniverses. The abundance of absorption provided by the absorption theorem makes this propagation often possible, and if it is not, gives us sufficient structural and algebraic information about the instance which makes the propagation possible nevertheless, until the instance becomes trivially solvable.

Zhuk's starting point is a theorem stating that every Taylor algebra has a proper subuniverse of one of four special types (see Subsections~\ref{subsec:four_types}). Zhuk derives the \emph{four types theorem}~\cite{Zhuk20:dichotomy} from a complicated result in clone theory, Rosenberg's classification of maximal clones~\cite{Rosenberg70} (the dependence of this approach on Rosenberg's result is removed in \cite{zhuk21:strong-subalgs}).
Given the four types theorem, the overall strategy for the polynomial algorithm for Theorem~\ref{dichotomy_real} is natural and similar in spirit to the absorption technique -- to keep reducing to one of such subuniverses until the problem becomes trivial. 
Although Zhuk's theory has a nontrivial intersection with the absorption theory,  these connections were not properly explored and verbalized.

Bulatov's algorithm in his proof of Theorem~\ref{dichotomy_real} employs a similar general idea, he reduces the instance to certain subuniverses. However, these special subuniverses are defined, as opposed to absorption and Zhuk's theories, in a very local way. They are sets that are, in a sense, closed under edges (e.g. strongly connected components) of a labeled directed graph whose vertices are the elements of the algebra. Bulatov introduces three basic kinds of edges (see Subsection~\ref{subsec:prelim_edges}), whose presence indicates that the local structure around the adjacent vertices, namely the subuniverse generated by the two vertices, somewhat resembles the three interesting tractable cases in Schaefer's Boolean dichotomy. What makes this approach work is a fundamental theorem (Theorem~1  \cite{Bulatov16:graph}, see also \cite{Bulatov04:graph}), the \emph{connectivity theorem}, which says that the edges sufficiently approximate the algebra in the sense that the directed graph is connected.
The proof uses rather technically challenging constructions involving  operations in the algebra.

In Section~\ref{sec:taylor} we first describe some of the connections between absorption theory and Zhuk's theory, and explain simplifications and refinements that were scattered across literature, including a refinement of the absorption theorem that follows from~\cite{Zhuk17:proof,Zhuk20:dichotomy}. We also give two new results improving pieces of the two theories. 
The major novel contribution of Section~\ref{sec:taylor} is Theorem~\ref{ff:subdirect}, a purely relational fact which roughly states that each ``interesting'' relation that uses all the domain elements in every coordinate pp-defines a binary relation with the same properties or a ternary relation of a very particular shape. Although the proof is elementary and not very long, it enables us to derive both Zhuk's four types theorem and Bulatov's connectivity theorems as corollaries. It may be also of interest for some readers to note that  theorems in this section often even do not require the algebra to be Taylor -- they concern all finite idempotent algebras.

\subsection{Minimal Taylor algebras}

The advantage of studying minimal reducts within a class of interest was clearly demonstrated in the work of Brady~\cite{Brady20:bw}. He concentrated on so-called bounded width algebras -- these are algebras that play the same role in solvability of CSPs by local consistency methods~\cite{Barto14:local} as Taylor algebras do for polynomial time solvability. The theory he developed enabled him to 
classify all the minimal bounded width algebras on small domains.
 Our first contributions in Section~\ref{sec:mtalgebras} show that the basic facts for minimal bounded width algebras have their counterparts for minimal Taylor algebras. For instance, Proposition~\ref{prop:TMexist} shows that every Taylor algebra does have a minimal Taylor reduct, and so minimal Taylor algebras are indeed sufficiently general, e.g., in the CSP context. 

The main results in Section~\ref{sec:mtalgebras} show that in minimal Taylor algebras
 concepts of the three approaches are simpler and have stronger properties (Subsection~\ref{ssec:mtalgebras_absorption} for absorption and the four types, Subsection~\ref{ssec:mtalgebras_edges} for edges) and
there are  deep connections among them (Subsection~\ref{subsec:abs_vs_edges}). Additional connections are given in Section~\ref{sec:omitting_types}, where various  classes of algebras are characterized in terms of algorithmic properties and types of edges, types of operations, and types of absorption present in the algebras. We now  discuss a sample of the obtained results. 

Edges, as we already mentioned, are pairs of elements for which the local structure around the pair resembles one of the three interesting polynomially solvable cases in Schaefer's Boolean dichotomy \cite{Schaefer}. More precisely, and specializing to one kind of edges, we say that $(a,b)$ is a \emph{majority edge} if the subalgebra $\alg E$ generated by $a$ and $b$ has a proper congruence (i.e., invariant equivalence relation) $\theta$ and a term operation $t$ that acts as the majority operation on the equivalence classes $a/\theta$ and $b/\theta$. The resemblance of the two-element majority algebra is in general quite loose -- the equivalence $\theta$  can have many more classes and there may be many more operations in $\alg E$ other than $t$. However, in minimal Taylor algebras, $\alg E$ modulo $\theta$ is always term equivalent to the two element majority algebra (Theorem~\ref{thm:minimaledgesquotients}).

The second sample concerns the simplest absorbing subuniverses, the 2-absorbing ones, which also constitute one of the four types of Zhuk's fundamental theorem. The 2-absorption of a subuniverse $B$ is a relatively strong property that requires the existence of \emph{some} binary term operation $t$ whose result is always in $B$ provided at least one of the arguments is in $B$. An extreme further strengthening is as follows: the result of applying \emph{any} operation $f$ to an argument that contains an element in $B$ in \emph{any} essential coordinate is in $B$. It turns out that these notions actually coincide for minimal Taylor algebras (Theorem~\ref{thm:bin_abs}). What is perhaps even more surprising is the connection to Bulatov's theory: 2-absorbing sets are exactly subsets stable (in a certain sense) under all the three kinds of edges (Theorem~\ref{thm:bin_abs_stable}). 

Finally, we mention that the clone of any minimal Taylor algebra is generated by a single ternary operation (Theorem~\ref{thm:single_operation_generates}). This, together with other structural results in this paper, may help in enumerating Taylor algebras -- at the very least we know that  there are at most $n^{n^3}$ of them over a domain of size $n$. Such a catalogue could be a valuable source of examples for CSP-related problems as well as universal algebra. Additionally, having a complete catalogue of minimal Taylor algebras for a given domain allows us to write down an explicit, concrete generalization of Schaefer's Dichotomy Theorem 
for a domain of that size.

\subsection{Follow up work} \label{subsec:follow-up}

Brady has already initiated the project of enumerating minimal Taylor algebras~\cite{brady:notes}. In particular, he has proved that,  up to term-equivalence and permutations of the domain, there are exactly $24$ minimal Taylor algebras on a domain of size $3$. This gives us a concrete list of the hardest tractable CSPs on the 3-element domain, refining the main result of~\cite{Bulatov06:three-element}. In an unpublished work of a subset of the authors and Albert Vucaj, the number of cases required to state the 3-element dichotomy was further reduced by means of the pp-constructibility technique from~\cite{Barto18:wonderland}. 

Another novel contribution in Brady's CSP notes~\cite{brady:notes} is the theory of stable sets which provide a common generalization of three out of the four types in Zhuk's approach. This theory is applied to simplify and improve results on local consistency~\cite{Barto14:local,Kozik16:weak-consistency} and  robust approximation~\cite{Barto16:robustly} of CSPs.

\subsection{Organization of the paper}
  
The paper is split into two parts. In Part I we state and discuss the main concepts and results. Part II provides the technical details, that is, additional definitions, results, and proofs of the claims made in Part I.

\newpage

\part{Main results}

\section{Basics}

In this section we expand on the explanation of the algebraic approach given in the previous section. The aim is to provide a short but self-contained introduction to the subject which, however, cannot replace more comprehensive material such as~\cite{Barto17:polymorphisms}. Subsection~\ref{subsec:alg_approach} explains the translation of the CSP to the algebraic language, Subsection~\ref{subsec:boolean} details Schaefer's dichotomy theorem, and Subsection~\ref{subsec:prelim_taylor} introduces  the central concept of the algebraic theory of the CSP, the Taylor algebra.
Throughout the paper, notations and concepts are introduced as the need arises, occasionally in a somewhat informal way; please see Section~\ref{app:prelim} for a more systematic list of definitions.

\subsection{From structures to algebras}
\label{subsec:alg_approach}

To every CSP template $\relstr{A}$ we assign a relational clone $\relclone R$, a clone $\clone C$, and an algebra $\algA$.
The connections between these concepts can be depicted as follows.

$$
\mbox{relational structures} \rightarrow
\mbox{relational clones} \leftrightarrow
\mbox{clones} \leftarrow
\mbox{algebras}
$$

\noindent
We restrict our attention to idempotent structures, i.e., structures containing all the singleton unary relations. This is, as already discussed, not a severe restriction.
Let us fix an idempotent relational structure $\relstr{A}$ with finite universe $A$.

Let $\relclone R$ be the set of all relations that are \emph{pp-definable} from $\relstr{A}$, i.e., definable by a first order formula using only relations in $\relstr{A}$, the equality relation, conjunction, and existential quantification. 
This set is a \emph{relational clone}, that is, it is closed under pp-definitions; and it still captures the complexity of the CSP over $\relstr A$ in the following sense: if a structure $\relstr B$ with the same universe as $\relstr A$ contains only relations in $\relclone R$, then the CSP over $\relstr B$ can be reduced to the CSP over~$\relstr A$ by simply replacing each constraint by its pp-definition. Moreover, if $\relstr B$ is rich enough, e.g., contains the original relations in $\relstr A$, then the CSPs over $\relstr A$ and $\relstr B$ are polynomial-time (even log-space) equivalent.

Next, we define $\clone C$ as the set of all \emph{operations}, i.e., mappings $f:A^n \to A$ for some $n$, that preserve all relations in $\relstr A$ (equivalently, all relations in $\relclone R$), i.e., $f$ applied coordinate-wise to tuples in a relation of $\relstr A$ gives a tuple which is again in that relation. Such operations are also called \emph{compatible} with $\relstr A$  or \emph{polymorphisms} of $\relstr{A}$, and we also say that relations in $\relstr A$ are \emph{invariant} under $f$. The set $\clone C$ is a \emph{clone}, that is, it contains all the projections $\proj^n_i$ (the $n$-ary projection to the $i$-th coordinate) and is closed under composition. Since each $f \in \clone C$ in particular preserves the singleton unary relations in $\relstr A$, it is \emph{idempotent}, i.e., $f(a,\dots, a)=a$ for any $a \in A$. 
Importantly~\cite{Geiger68:pol-inv,Bodnarcuk69:pol-inv}, a  relation is in $\relclone R$ if and only if it is invariant under every $f \in \clone C$, so our clone still captures the complexity of the CSP over $\relstr A$. 

The final step from clones to algebras is not essential and is taken mostly for convenience: algebras are more classical objects than clones, with a better established terminology. In particular, we can apply the standard constructions of taking subalgebras or subuniverses (universes of subalgebras), products, and quotients over congruences (invariant equivalence relations). 
This final step is done by selecting some ``generating  operations'' in $\clone C$ and giving them names. Formally, we take an algebra $\alg{A}$ (with universe $A$) of some signature  such that the smallest clone containing the operations in $\alg{A}$, called the \emph{clone of $\alg A$} and denoted $\Clo(\algA)$, is equal to $\clone C$.  Note that $\Clo(\algA)$ is precisely the set of \emph{term operations} of $\alg A$, i.e., operations that can be defined from the operations in $\alg A$ by a term. The algebra $\alg A$ is not uniquely determined by $\relstr A$, any algebra $\alg A'$ which is \emph{term-equivalent} to $\alg A$, i.e., $\Clo(\alg A') = \Clo(\algA)$, can be taken instead. Also observe that $\alg A$ is idempotent (consists of idempotent operations) and that the connection to $\relclone{R}$ remains valid: a relation $R$ is in $\relclone R$ if and only if it is invariant under every operation of~$\algA$; in other words, if $R$ is a \emph{subpower} (subuniverse of a finite power) of~$\alg A$. We use $\leq$ to denote the subuniverse or subalgebra relation, e.g., $R \leq \alg{A}^n$ means that $R$ is a subuniverse of the $n$-th power of $\alg A$. The set of all subpowers of $\alg A$ is denoted $\Inv(\algA)$.

In summary, to every idempotent relational structure $\relstr{A}$ we assigned an idempotent algebra $\alg A$ so that the relational clone  of all pp-definable relations is equal to $\Inv(\alg A)$, the set of subpowers of $\alg A$; and  the clone of all compatible operations is equal to $\Clo(\alg A)$, the set of all term operations of $\alg A$.  The CSP over $\relstr A$ is polynomial-time equivalent to the CSP over any sufficiently rich structure whose relations are in $\Inv(\algA)$. It also follows  that the CSP over some structure $\relstr B$ is polynomial-time reducible to the CSP over $\relstr A$ whenever $\alg A$ is a \emph{reduct} of an algebra $\alg B$ assigned to $\relstr{B}$, i.e., $\Clo(\alg A) \subseteq \Clo(\alg B)$.

As we concentrate on algebras coming from finite idempotent relational structures, we make the following \textbf{running assumption in definitions and theorems: all the involved algebras are finite and idempotent. We do not usually explicitly mention this assumption in our theorem statements. }

\subsection{Boolean CSPs} \label{subsec:boolean}

We now discuss the dichotomy theorem for Boolean CSPs from the algebraic perspective. 
The following three types of two-element algebras play a special role.

\begin{itemize}
    \item \emph{Two-element semilattice}: a two-element set together with the binary maximum operation with respect to one of the two possible orderings of the domain, e.g., $\alg A = (\{0,1\}; \vee)$ where $\vee$ is the maximum operation. Here $\Clo(\algA)$ is the set of operations of the form $x_{i_1} \vee x_{i_2} \vee \dots \vee x_{i_k}$ and $\Inv(\algA)$ is the set of relations pp-definable from the singleton unary relations and the relation $x \vee y \vee \neg z$ (note that we abuse the notation and use $\vee$ both for the maximum operation and logical disjunction). Therefore, if $\alg A$ is a reduct of an algebra associated to a template $\relstr B$, then the CSP over $\relstr B$ can be reduced to Dual-Horn-3SAT (and thus to Horn-3SAT).  
    
    \item \emph{Two-element majority algebra}: a two-element set together with the unique ternary operation that returns the majority of its arguments, e.g., $\alg A = (\{0,1\}; \maj)$, $\maj(a_1, a_2, a_3) = 1$ iff there are at least two $i$s with $a_i=1$.
    Here $\Clo(\algA)$ consists of all idempotent, \emph{monotone} (compatible with the inequality relation $\leq$), and \emph{self-dual} (i.e., compatible with the disequality relation $\neq$) operations and $\Inv(\algA)$ is the set of relations pp-definable from the binary relations $x \vee y$, $\neg x \vee y$, and $\neg x \vee \neg y$. Therefore, if $\alg A$ is a reduct of an algebra associated to a template $\relstr B$, then the CSP over $\relstr B$ can be reduced to 2SAT.   
    \item \emph{Affine Mal'cev algebra of a two-element group}: a two-element set together with the ternary addition with respect to the unique group structure on the universe, e.g., $\alg A = (\{0,1\}; x+y+z \pmod 2)$. Here $\Clo(\alg A)$ is the set of operations of the form $x_{i_1} + x_{i_2} + \dots + x_{i_k} \pmod{2}$ with $k$ odd and $\Inv(\algA)$ is the set of relations pp-definable from the relations $x+y+z\equiv 1\pmod{2}$ and $x+y+z\equiv 0\pmod{2}$. Therefore, if $\alg A$ is a reduct of an algebra associated to a template~$\relstr B$, then the CSP over~$\relstr B$ can be solved by Gaussian elimination.
\end{itemize}

It follows from the full description of clones on a two-element set by Post~\cite{post41:clones} (and is not hard to verify) that every idempotent clone, except the clone consisting only of projections, contains one of the two semilattice operations or the majority operation or the ternary addition. Therefore, we have the following dichotomy for the CSP over $\relstr A$ and an assigned algebra $\alg A$. Either 
\begin{itemize}
    \item $\alg A$ contains only projections, in which case $\Inv(\alg A)$ contains all relations, and then the CSP over $\relstr{A}$ is NP-complete since any CSP, e.g. 3SAT, can be reduced to it; or
    \item $\alg A$ contains some nonprojection, in which case one of three types of algebras above is a reduct of $\alg A$, and then the CSP over $\relstr{A}$ is solvable in polynomial time.
\end{itemize}

\subsection{Taylor algebras} \label{subsec:prelim_taylor}

For two-element idempotent structures, the necessary and sufficient condition for polynomial-time solvability is that $\relstr{A}$ does not pp-define every finite structure or, in algebraic terms, $\alg A$ contains a nonprojection. For larger universes, this condition is no longer sufficient and pp-definability needs to be weakened to so called pp-interpretability. The algebraic counterpart goes as follows.

\begin{definition} \label{def:taylor}
An (idempotent, finite) algebra $\algA$ is a \emph{Taylor algebra} if no quotient of a subalgebra of a power  of $\algA$ is a two-element algebra whose every operation is a projection. 
\end{definition}

\noindent
The following proposition shows that powers can be dropped from this definition.  In particular, a two element algebra is Taylor iff it has a nonprojection operation.

\begin{proposition}[\cite{Bulatov01:varieties}] \label{prop:taylor-without-powers}
An algebra $\algA$ is Taylor if and only if  no quotient of a subalgebra of $\algA$ is a two-element algebra whose every operation is a projection. 
\end{proposition}

Birkhoff's HSP Theorem~\cite{Birkhoff} connects the three main algebraic constructions (subalgebras, products, quotients) with \emph{identities}, i.e., universally quantified equations. In particular, $\alg A$ is Taylor iff its operations satisfy some set of identities which are not satisfiable by projections. Several types of operations satisfying such identities generalize the three types in the previous subsections.

\begin{itemize}
    \item A \emph{semilattice operation} is a binary operation $\vee$ which is commutative, idempotent, and associative. An algebra $(A;\vee)$, where $\vee$ is a semilattice operation, is called a \emph{semilattice}.
    \item A \emph{majority operation} is a ternary operation $m$ which satisfies $m(x,x,y)=m(x,y,x)=m(y,x,x)=x$ (for any $x,y$ in the universe). More generally, an $n$-ary \emph{near unanimity operation} is an operation $f$ satisfying $f(x,x, \dots, x, y, x, x, \dots, x)=x$ for any position of $y$. An example for an odd $n$ is the $n$-ary majority operation $\maj_n$ on $\{0,1\}$, that is, $\maj_n(a_1, \dots, a_n)=1$ iff the majority of the $a_i$ is 1.  
    \item A \emph{Mal'cev operation} is a ternary operation $p$ satisfying $p(y,x,x) = p(x,x,y)=y$. An example is the operation $x-y+z$ on $A$, where $+$ and $-$ is computed with respect to an  abelian group structure  on $A$. In this case the  algebra $(A; x-y+z)$ is called the \emph{affine Mal'cev algebra} of that abelian group.
\end{itemize}

The two basic algorithmic ideas to efficiently solve a CSP are  local propagation algorithms~\cite{Bulatov08:dualities,Barto14:local} and finding a generating set of all solutions~\cite{Berman10:varieties,Idziak10:few}. 

We say that a CSP, or its assigned algebra $\alg A$, has \emph{bounded width} if a local propagation algorithm correctly decides it. Examples include algebras with a semilattice or a majority operation, and a nonexample is an affine Mal'cev algebra of a nontrivial abelian group. The bounded width theorem~\cite{Barto14:local} characterizes bounded width algebras as algebras $\alg A$ such that no quotient of a subalgebra of a power $\algA$ is a nontrivial abelian algebra.

\begin{definition} \label{def:abelian}
An algebra $\alg{A}$ is \emph{abelian} if the diagonal $\Delta_A = \{(a,a) \st a \in A\}$ is an equivalence class of a congruence of $\alg{A}^2$.
\end{definition}

\noindent
An example of an abelian algebra is an algebra whose every operation is a projection. A more interesting example is an affine Mal'cev algebra, where a  congruence satisfying the definition is the congruence $\alpha$ defined by
$
((x_1,x_2), (y_1,y_2)) \in \alpha \mbox{ iff }
x_1-x_2 = y_1-y_2.
$
More generally, each \emph{affine module}, i.e., an algebra whose term operations are exactly the idempotent term operations of a module over a unital ring, is abelian.
In fact, these are the only examples of abelian \emph{Taylor} algebras. 

\begin{theorem}[\cite{hobby88:tct}] \label{thm:FTA}
A Taylor algebra is abelian if and only if it is an affine module. 
\end{theorem}

\noindent
 The original proof of this result makes use of tame congruence theory, a well-developed theory of finite algebras that we have not mentioned yet. An alternative proof using absorption is available~\cite{Barto15:malcev}.

We now turn to the other algorithmic idea, finding a generating set of all solutions. It turns out~\cite{Berman10:varieties} that $\alg A$ has ``small'' generating sets of powers (polynomial in the exponent) if and only if it has \emph{few subpowers}, i.e., the number of subuniverses of $\alg A^n$ is $2^{O(n)}$. In such a case, the few subpowers algorithm~\cite{Idziak10:few} finds a generating set of solutions in polynomial time. Examples of algebras with few subpowers include algebras with a near unanimity or Mal'cev operation, and a nonexample is any nontrivial semilattice. 

Algebras with few subpowers were characterized in~\cite{markovic12:cube-terms} as those that do not have certain special subuniverses, called cube term blockers. It became clear that the concept is significant beyond this context~\cite{Barto15:malcev,zhuk21:strong-subalgs}. For this reason we prefer the terminology from the latter paper and call them projective subuniverses. In the  definition, $\Clo_n(\alg A)$ is the set of $n$-ary operations in $\Clo(\alg A)$, and we use $a_i$ to denote the $i$-th component of a tuple $\tuple{a}$. 

\begin{definition}
Let $\alg A$ be an algebra and $B \subseteq A$. We say that $B$ is a \emph{projective subuniverse} if 
for every $f \in \Clo_n(\alg A)$ there exists a coordinate $i$ of $f$ such that $f(\tuple{a}) \in B$ whenever $\tuple{a} \in A^n$ is such that $a_i \in B$.
\end{definition}

\noindent
Observe that a projective subuniverse of $\alg A$ is, indeed, a subuniverse. 

Many of the algebraic concepts that we introduce (such as absorbing subuniverses or strongly projective subuniverses from Section~\ref{sec:mtalgebras}) have useful equivalent characterizations in terms of relations. Such a characterization for projective subuniverses is especially elegant, and we state it here for comparison with characterizations of absorption in Subsection~\ref{subsec:mt_abs}.

\begin{proposition}[Lemma 3.2 in~\cite{markovic12:cube-terms}] \label{prop:cubetermblocker}
    Let $\alg A$ be an algebra and $B \subseteq A$. Then $B$ is a projective subuniverse of $\alg A$ if and only if, for every $n$, the relation $B(x_1) \vee B(x_2) \vee \dots \vee B(x_n)$ is a subpower of $\alg A$.
\end{proposition}

\section{Three approaches}

In this section we introduce the central concepts and results from the three algebraic approaches. 

\subsection{Edges} \label{subsec:prelim_edges}

We start by formally introducing the three types of edges used in Bulatov's approach to the CSP. These three types are inspired by the three types of two-element algebras from Subsection~\ref{subsec:boolean}. In the definition, $\Sg_{\alg A}(a,b)$ denotes the subalgebra of $\alg A$ generated by $\{a,b\}$. 

\begin{definition} \label{def:edges}
    Let $\alg A$ be an algebra. 
    A pair $(a,b) \in A^2$ is an {\em edge} if there exists a proper congruence $\theta$ on $\Sg_{\alg A}(a,b)$, called a \emph{witness} for the edge,  such that one of the following happens:
      \begin{itemize}
          \item (\emph{semilattice edge}) There is a term operation $f \in \Clo_2(\alg A)$ acting as a join semilattice operation on $\{a/\theta,b/\theta\}$ with top element $b/\theta$, i.e. $f(a/\theta,b/\theta)$, $f(b/\theta,a/\theta) \subseteq b/\theta$.
          \item (\emph{majority edge}) There is a term operation $m \in \Clo_3(\alg A)$ acting as a majority operation on $\{a/\theta,b/\theta\}$.
          \item (\emph{abelian edge}) The algebra $\Sg_{\alg A}(a,b)/\theta$ is abelian.
      \end{itemize}
    An edge $(a,b)$ is called \emph{minimal} if for some maximal congruence $\theta$ witnessing the  edge     and every $a',b' \in A$ such that $(a,a'), (b,b') \in \theta$, we have $\Sg_{\alg A}(a',b') = \Sg_{\alg A}(a,b)$. Such $\theta$ is called a witness for this minimal edge. 
\end{definition}
\noindent 
  A witnessing congruence $\theta$ for an edge $(a,b)$ necessarily separates $a$ and $b$, i.e., $(a,b) \not\in \theta$, since each congruence class of an idempotent algebra is a subuniverse. 
  Moreover, if $\theta$ is a witness for an edge $(a,b)$, then any proper congruence of $\Sg_{\alg A}(a,b)$ containing $\theta$ witnesses the same edge.
  \anote{It is possible that a larger congruence does not witness the same type of an edge. This applies to majority edges. For example, let $A=(\{0,1,2\}, \vee,m)$ where $\vee$ is the semilattice operation with $0\vee 1=2$, and $m$ is majority on $\{0,1\}$ and $m(x,y,z)=2$ whenever $2\in\{x,y,z\}$. Then $01$ majority witnessed by the equality relation, but is not a majority edge witnessed by, say $02|1$. I would add a caveat that type is not necessarily preserved.}
\lnote{I don't think this is possible - the same witnesses the same edge. Note that our edges are inclusive - it can be both semilattice and majority. In your example $02|1$ is not a congruence}

Note that if $(a,b)$ is an edge  of majority or abelian type, then so is $(b,a)$. 
If $(a,b)$ is a semilattice edge it can happen that $(b,a)$ is not an edge at all,
in fact this is always the case for minimal edges in a minimal Taylor algebra, see Subsection~\ref{ssec:mtalgebras_edges}.

In order to make the concepts in this paper elegant and theorems more general, we deviate from the definition given in e.g.~\cite{Bulatov04:graph,Bulatov20:local1}. 
There, majority edges have an additional requirement that the same congruence does not witness the semilattice type, and abelian edges (called affine) required the quotient to be an affine module. In Taylor algebras, abelian edges are the same as affine edges by Theorem~\ref{thm:FTA}. Outside Taylor algebras, it makes sense to separate abelian edges into two types: affine  and sets whose only term operations are projections, as is done in~\cite{Bulatov20:local1}.

Also note that the definition of abelian edges  is of a different type: it restricts the set of term operations from above, as opposed to semillatice and majority edges that restrict them from below. We shall see in Theorem~\ref{thm:minimaledgesquotients} that these differences disappear in minimal Taylor algebras.

Minimal edges do not appear in Bulatov's theory in this form. Somewhat related are thin edges, which at present have rather technical definitions with the exception of thin semilattice edges. We show in Proposition~\ref{prop:thin_semilattice} that minimal semilattice edges and thin semilattice edges coincide in minimal Taylor algebras.

The fundamental theorem of Bulatov's approach shows that edges sufficiently approximate any algebra in the following sense. 

\begin{ThCSrestatablen}{theorem}{BulatovFundamentalThm}{The Connectivity Theorem \cite{Bulatov16:graph,Bulatov04:graph}\ooo{, \hyperlink{ProofThm32}{link to proof}}}{thm:connected}
  The directed graph formed by the edges of any algebra is (weakly) connected.
\end{ThCSrestatablen}

Another important technical result for the edge approach is that the operations appearing in Definition~\ref{def:edges} can be significantly unified, see Theorem 7 in~\cite{Bulatov16:graph}. We give a much stronger unification result
in Theorem~\ref{thm:unifiededges}.

\subsection{Absorption} \label{subsec:prelim_absorption}

A central concept for the absorption theory, as well as for Zhuk's theory, is an absorbing subuniverse.

\begin{definition} \label{def:absorption}
Let $\alg A$ be an algebra and $B \subseteq A$. 
We call $B$ an $n$-absorbing set of $\alg A$ 
if there is a term operation $t \in \Clo_n(\alg A)$ 
  such that $t(\tuple{a}) \in B$ whenever $\tuple{a} \in A^n$ and  $|\{i\st a_i\in B\}|\geq n-1$.
  
  If, additionally, $B$ is a subuniverse of $\alg A$, we write $B\abs_n\alg A$, or $B\abs\alg A$ when the arity is not important.
\end{definition}

\noindent  
We also say ``$B$ absorbs $\alg A$ (by $t$)'' in the situation of Definition~\ref{def:absorption}. Of particular interest for us are $n$-absorbing subuniverses with $n=2$,  e.g., $\{1\}$ in the semilattice $(\{0,1\}; \vee)$, or $n=3$, e.g., $\{0\}$ and $\{1\}$ in the two-element majority algebra on $\{0,1\}$.

The fundamental theorem for the absorption theory states that every proper subdirect linked subpower produces a nontrivial absorbing subuniverse. The required definitions are as follows.
We say that $R$ is a \emph{subdirect product} of $A_1$, \ldots, $A_n$, and write $R \subseteq_{sd} A_1 \times \dots \times A_n$, if $R \subseteq A_1 \times \dots \times A_n$ and the projection of $R$ to any coordinate $i$ is full, i.e., equal to $A_i$. Similarly, we write $R \leq_{sd} \alg{A}_1 \times \dots \times \alg{A}_n$ if $R$ is additionally a subuniverse of that product.  %
If $R \subseteq_{sd} A \times C$, we call $R$ \emph{linked} if it is connected when viewed as a bipartite graph between disjoint copies of $A$ and $C$.

\begin{theorem}[The Absorption Theorem~\cite{Barto12:cyclic}] 
\label{thm:AToriginal} 
Suppose $R \leq_{sd} \alg A \times \alg C$ is proper and linked, and $\alg A$ and $\alg C$ are Taylor.
 Then $\alg A$ or $\alg C$ has a nontrivial absorbing subuniverse.
\end{theorem}

We will often take advantage of one of the results of the absorption theory, the characterization of Taylor algebras by means of cyclic operations.
An alternative proof is now available using Zhuk's approach~\cite{zhuk21:strong-subalgs}.

\begin{theorem}[\cite{Barto12:cyclic}] \label{thm:cyclic}
The following are equivalent for any algebra.
\begin{itemize}
    \item $\algA$ is Taylor.
    \item There exists $n>1$ such that  $\algA$ has a term operation $t$ of arity $n$ which is \emph{cyclic}, that is, for any $\tuple{x} \in A^n$,
$$
t(x_1, x_2, \dots, x_n) = t(x_2, \dots, x_n,x_1).
$$
\item For every prime $p > |A|$, $\alg A$ has a term operation $t$ of arity $p$ which is cyclic.
\end{itemize}
\end{theorem}

\subsection{Four types} \label{subsec:four_types}

The fundamental theorem for Zhuk's approach is that each Taylor algebra has a proper subuniverse of one of four types. We first introduce the required additional concepts, centers and polynomially complete algebras.

The concept of a center is still evolving, and it is not yet clear what the best version would be for general algebras. Our definition follows~\cite{Zhuk20:dichotomy}, although a more recent paper~\cite{zhuk21:strong-subalgs} made an adjustment motivated by this work. As we shall see in Theorem~\ref{thm:center_abs},
the situation is  much cleaner for minimal Taylor algebras.

Our definition of center of an algebra requires the concept of a (left/right) center of a relation.
The \emph{left center}  of $R \subseteq A \times B$ is the set $\{a \in A \st (\forall b \in B) (a,b) \in R\}$. 
If $R$ has a nonempty left center, it is called \emph{left central}. 
\emph{Right center} and \emph{right central} relations are defined analogically. 
A relation is \emph{central} if it is left central and right central.

  \begin{definition} \label{def:center}
  A subset $B \subseteq A$ is a \emph{center} of $\alg A$ if there exists an algebra $\alg C$ (of the same signature) with no nontrivial 2-absorbing subuniverse and $R \leq_{sd} \alg A \times \alg C$ such that $B$ is the left center of $R$. The relation $R$ is called a \emph{witnessing} relation. 
  If $\alg C$ can be chosen Taylor, we call $B$ a \emph{Taylor center} of $\alg A$.
 \end{definition}

An algebra $\alg A$ is  \emph{polynomially complete} if  every operation on $A$ is in  the clone generated by $\alg A$ together with the constant operations. This property can be equivalently phrased in terms of relations: $\alg A$ is polynomially complete if and only if it has no proper reflexive (that is, containing all the tuples $(a,a, \dots, a)$) irredundant (that is, no binary projection is the equality relation) subpowers.

The fundamental theorem can now be stated as follows.
\begin{ThCSrestatablen}{theorem}{ZhukFundamental}{The Four Types Theorem \cite{Zhuk20:dichotomy}}{thm:4types} 
Let $\alg A$ be an algebra, then
 \begin{enumerate}
  \item[(a)] $\alg A$ has a nontrivial 2-absorbing subuniverse, or 
  \item[(b)] $\alg A$ has a nontrivial center
    (which is a Taylor center in the case where $\alg A$ is a Taylor algebra), 
    or 
  \item[(c)] $\alg A/\alpha$ is abelian for some proper congruence $\alpha$ of $\alg A$, or 
  \item[(d)] $\alg A/\alpha$ is polynomially complete for some proper congruence $\alpha$ of $\alg A$.
 \end{enumerate}
\end{ThCSrestatablen}
  
  \noindent
  We referred to four types of \emph{subuniverses} whereas cases (c) and (d) talk about congruences -- the subuniverses used in~\cite{Zhuk20:dichotomy} are the equivalence classes of such congruences.

  Examples of \emph{simple} (with only trivial congruences) Taylor algebras, for which one of the cases takes place and no other, 
  are (a) a two-element semilattice, (b) a two-element majority algebra, (c) an affine Mal'cev algebra, and (d) the three element rock-paper-scissors algebra  
  $$
  (\{\paper,\rock,\scissors\}; \winner(x,y)).
  $$
  Note, however, that Theorem~\ref{thm:4types} does not require that $\alg A$ is Taylor. 
  If it is, then we get additional properties: centers  are 3-absorbing by Proposition~\ref{prop:center_absorbs} and abelian algebras are term equivalent to affine modules by Theorem~\ref{thm:FTA}. For non-Taylor idempotent algebras, \cite{zhuk21:strong-subalgs} suggests a similar five type  theorem, which also follows immediately from our results.

\section{Taylor algebras} 
\label{sec:taylor}

This section presents unifications, simplifications, and refinements of the three algebraic theories in the setting of  Taylor algebras  (still finite and idempotent) that are not necessarily minimal. In Subsection~\ref{subsec:AT} we discuss the already existing refinements to the proof of the absorption theorem, and provide two additional new refinements in Proposition~\ref{prop:center} and Proposition~\ref{prop:center_absors_in_T}. 
This gives tight links 
to centers and projective subuniverses. Subsection~\ref{subsec:SIS} contains the main contribution of this section, Theorem~\ref{ff:subdirect}. This theorem together with additional technical contributions, Theorems~\ref{thm:maj} and \ref{thm:semilat}, directly imply the fundamental facts in the two proofs of the CSP Dichotomy Theorem  -- the four types theorem and the connectivity theorem.
The proofs of claims in this section are in Section~\ref{app:taylor}.

\subsection{Absorption theorem} \label{subsec:AT}

We first sketch the original proof of the absorption theorem (Theorem~\ref{thm:AToriginal}) and then comment on subsequent improvements and simplifications. 

For simplicity, we sketch the proof for a slighty simplified version with $\alg A = \alg C$. So, we assume that $\alg A$ is Taylor and that $R \leq_{sd} \algA^2$ is linked. The original proof that $\alg A$ then necessarily has a proper absorbing subuniverse can be divided into 3 steps.
\begin{itemize}
    \item[(1)] From $\alg A$  being Taylor it is derived that $\alg A$  either has a nontrivial 2-absorbing subuniverse or a \emph{transitive} term operation $t$ of some arity $n$, i.e., for each $b,c \in A$ and every coordinate $i$ of $t$, there exists a tuple $\tuple{a} \in A^n$ with $a_i=b$ such that $t(\tuple{a})=c$.
    \item[(2)] Using the transitive operation, it is proved that if $\alg A$ has no nontrivial absorbing subuniverses, then $R$ is left or right central. 
    \item[(3)] It is shown that the transitive operation witnesses that the left (right) center absorbs $\alg A$. 
\end{itemize}
 
The first step was explored in more detail in~\cite{Barto15:malcev}.  Lemma 2.7 in~\cite{Barto15:malcev} shows that each  algebra has a nontrivial projective subuniverse or a transitive term operation. A simple argument
then shows that every projective subuniverse in a Taylor algebra is 2-absorbing, a witness is, e.g., any operation of the form  $t(x, \dots,x$, $y, \dots,y)$ where $t$ is cyclic.

As for the second step, it has turned out that left (or right) central relations can be very easily obtained from linked relations by means of pp-definitions, avoiding algebraic considerations altogether.

\begin{ThCSrestatable}{proposition}{LinkedImpliesCenterOld} {(\hyperlink{ProofProp41}{link to proof})}{prop:linked-implies-central}
Let $R \subseteq_{sd} A \times C$ be proper and linked. Then either $R$ is left central or pp-defines a proper subdirect symmetric central relation on $C$. 
\end{ThCSrestatable}

\noindent
We give a refined version for the case $A=C$ that derives central relations with further properties. 

\begin{ThCSrestatable}{proposition}{LinkedImpliesCenter}{(\hyperlink{ProofProp42}{link to proof})}{prop:center}
  Let $R\setsd A^2$ be linked and proper.
  Then $R$ pp-defines a subdirect proper central relation on $A$ which is symmetric or transitive.
\end{ThCSrestatable}

The third step, that a transitive operation witnesses absorption of left centers, is straightforward, see Proposition~\ref{prop:trans-and-center-implies-absorption}. A significant refinement, Corollary 7.10.2 in~\cite{Zhuk20:dichotomy},  shows that left centers are, in fact,  3-absorbing. An adjustment of the proof will also help us in proving Theorem~\ref{thm:center_abs}. 

\begin{ThCSrestatablen}{proposition}{CenterTernaryAbsorbs}{\cite{Zhuk20:dichotomy}\ooo{, \hyperlink{ProofProp43}{link to proof}}}{prop:center_absorbs}
If $B$ is a Taylor center of an algebra $\alg A$, then $B \abs_3 \alg A$. 
\end{ThCSrestatablen}
\noindent Note that in the previous proposition $\alg A$ need not be a Taylor algebra, but $\alg C$~(where the witnessing relation is $R\leq_{sd}\alg A\times \alg C$) must be.
The following proposition states that we can switch the condition:
\begin{ThCSrestatable}{proposition}{CenterTernaryAbsorbsSecond}{(\hyperlink{ProofProp44}{link to proof})}{prop:center_absors_in_T}
If $B$ is a center of a Taylor algebra $\alg A$, then $B\abs_3 \alg A$.
\end{ThCSrestatable}
\noindent
The assumptions of the latter proposition are easier to satisfy --- the algebra $\alg A$ is usually Taylor by default.

Altogether, either of the propositions above provides the following improvement of the absorption theorem, which does not seem to be explicitly stated in the literature.

\begin{ThCSrestatable}{corollary}{NewAT}{(\hyperlink{ProofCor45}{link to proof})}{cor:newAT}
Suppose $R \leq_{sd} \alg A \times \alg C$ is proper and linked, and $\alg C$ is  Taylor.
 Then $\alg A$ or $\alg C$ has a nontrivial 3-absorbing subuniverse.
\end{ThCSrestatable}

\subsection{Subdirect irredundant subpowers} \label{subsec:SIS}

We now present the unification result. It says that any ``interesting'' (subdirect irredundant proper) relation either pp-defines an interesting binary relation or pp-defines (and is even inter-pp-definable with) ternary relations of very particular shape (strongly functional). 
We have recently found out that Theorem~\ref{ff:subdirect} can be also deduced from Szendrei's results in~\cite{Szendrei:idempotent}. Our contribution is thus the presented formulation, its alternative proof, and the application to the fundamental theorems of the dichotomy proofs.

The definitions are as follows. A relation $R \subseteq_{sd} A^n$ is \emph{irredundant} if no projection to a pair of coordinates is the graph of a bijection; in other words, there is no bijective correspondnce between $a_i$ and $a_j$ in tuples $\tuple{a} \in R$.  We remark the definition of irredundancy is different than in some other papers. Also note that this definition agrees with the definition of irredundancy for reflexive relations given in Subsection~\ref{subsec:four_types}.

\emph{Strongly functional} relations can be defined as graphs of quasigroup operations or, more explicitly, as follows.

\begin{definition}
A relation $R \subseteq A^3$ is called \emph{strongly functional}   if 
        \begin{itemize}
        \item the binary projections of $R$ are equal to $A^2$, and
        \item every tuple in $R$ is determined by its values on any two coordinates.
        \end{itemize}
\end{definition}

We are ready to state the main result. In the theorem statement, pp-definability should be understood as pp-definability with parameters, i.e., one is allowed to use singleton unary relations.

\begin{ThCSrestatable}{theorem}{DefinabilityTheorem}{(\hyperlink{ProofThm47}{link to proof})}{ff:subdirect}%
    Let $R\subseteq_{sd} A^n$ be an irredundant proper relation. Then either
    \begin{itemize}
        \item 
        $R$ pp-defines an irredundant and proper $R'\subseteq_{sd} A^2$, or
        \item there exist strongly functional ternary relations $R_1,\dotsc,R_n\subseteq_{sd}A^3$ such that the set $\{R_1,\dots$, $R_m\}$ is inter-pp-definable with $R$ (i.e., the $R_i$s pp-define $R$ and, conversely, $R$ pp-defines all the $R_i$s). 
    \end{itemize}
\end{ThCSrestatable}

Theorem~\ref{ff:subdirect} implies that every algebra $\alg A$ has at least one of the following properties of its invariant relations.
\begin{enumerate}
    \item [(1)] $\alg{A}$ has no proper irredundant subdirect subpowers.
    \item[(2)] $\alg{A}$ has a proper irredundant binary subdirect subpower.
    \item[(3)] $\alg{A}$ has a ternary strongly functional subpower.
\end{enumerate}

\noindent
In the last case, it is easy to pp-define a congruence on $\alg A^2$ such that the diagonal is one of its classes, so $\alg A$ is abelian in that case.

\begin{ThCSrestatable}{proposition}{ImpliesAbelian}{(\hyperlink{ProofProp48}{link to proof})}{prop:ImpliesAbelian}
  If  $R \leq \alg A^3$ is a strongly functional relation, then $\alg A$ is abelian.
\end{ThCSrestatable}

\noindent
In case (1), subdirect relations have a very simple structure; for instance, any constraint $R(x_1, \dots$, $x_n)$ with subdirect $R$ is effectively a conjunction of bijective dependencies $x_i = f(x_j)$ (in particular, $\alg A$ is simple). It is also immediate that $\alg A$ is  polynomially complete. 
Less trivially, case (1) often leads to majority edges, as we show in Theorem~\ref{thm:maj} below. We require the following definition first.

\begin{definition}
Let $\algA$ be an algebra. By the \emph{connected-by-subuniverses} equivalence, denoted $\mu_{\algA}$, we mean the  
smallest equivalence containing all the pairs $(a,b)$ such that $\Sg_{\alg A}(a,b)\neq A$.
\end{definition}

\noindent We remark that the equivalence $\mu_{\alg A}$ is not, in general, a congruence of $\alg A$, so this concept may seem somewhat unnatural from the algebraic perspective.
This equivalence may be full. If we are in case (1) and $\mu_{\alg A}$ is not full, then the following theorem produces majority edges.

\begin{ThCSrestatable}{theorem}{GettingMajorityEdges}{(\hyperlink{ProofThm410}{link to proof})}{thm:maj}
Suppose that $\alg A$ has no subdirect proper irredundant subpowers. Then there exists a term operation $t \in \Clo_3(\alg A)$ such that for any $a,b \in A$ with $\Sg_{\alg A}(a,b) = \alg A$ (e.g., $(a,b)\notin\mu_{\algA}$), $t(a,a,b)=t(a,b,a)=t(b,a,a)=a$.
\end{ThCSrestatable}

\noindent 
In case (2) and when $\alg A$ is simple, a binary irredundant relation is necessarily linked. Then we get a central relation, e.g., by Proposition~\ref{prop:center}, and, by the following theorem, also semilattice edges in case that $\mu_{\alg A}$ is not full and the size of $A$ is at least three (in case that $\mu_{\alg A}$ is full the theorem is vacuously true). 

\begin{ThCSrestatable}{theorem}{GettingSemilatticeEdges}{(\hyperlink{ProofThm411}{link to proof})} {thm:semilat}
Suppose $\alg A$ with $|A|>2$ is simple and 
there exists a proper irredundant subdirect binary subpower. Then 
 there exists a $\mu_{\alg A}$-class $B$ such that, for every $b\in B,a\notin B$, the pair $(a,b) $ is a semilattice edge witnessed by the identity congruence.
\end{ThCSrestatable}

Zhuk's four types theorem (Theorem~\ref{thm:4types}) is now a consequence of Theorem~\ref{ff:subdirect}, Proposition~\ref{prop:ImpliesAbelian}, and Proposition~\ref{prop:center}. Indeed, one simply applies these facts to $\alg A$ factored by a maximal congruence, which is a simple algebra, and then lifts 2-absorbing subuniverses and centers back to $\alg A$. 
Moreover, we can improve item (d) in the four types theorem as follows.

\begin{ThCSrestatable}{corollary}{ZhukFundamentalRefined}{(\hyperlink{ProofCor412}{link to proof})}{thm:4typesRefined} 
Let $\alg A$ be an algebra, then
 \begin{enumerate}
  \item[(a)] $\alg A$ has a nontrivial 2-absorbing subuniverse, or 
  \item[(b)] $\alg A$ has a nontrivial center
    (which is a Taylor center in the case where $\alg A$ is a Taylor algebra), 
    or 
  \item[(c)] $\alg A/\alpha$ is abelian for some proper congruence $\alpha$ of $\alg A$, or 
  \item[(d)] $\alg A/\alpha$ has no proper irredundant subdirect subpowers  for some proper congruence $\alpha$ of $\alg A$.
 \end{enumerate}
\end{ThCSrestatable} 

The connectivity theorem (Theorem~\ref{thm:connected}) is also a straightforward consequence of the obtained results, Theorem~\ref{ff:subdirect}, Proposition~\ref{prop:ImpliesAbelian}, Theorem~\ref{thm:maj}, and Theorem~\ref{thm:semilat}. In fact, a little additional effort gives a somewhat stronger result.

\begin{ThCSrestatable}{corollary}{BulatovFundamentalMinimal}{(\hyperlink{ProofCor413}{link to proof})}{cor:connected}
  The directed graph formed by the minimal edges of any algebra is (weakly) connected.
\end{ThCSrestatable}

\section{Minimal Taylor algebras} \label{sec:mtalgebras}

We start this section by recalling the central definition and giving some examples. 

\begin{definition}
An algebra $\algA$ is called a \emph{minimal Taylor algebra} if it is Taylor but no proper reduct of $\alg A$ is. 
\end{definition}

There exist exactly four minimal Taylor algebras on a two-element set, counted up to term-equivalence: the two semilattices, the majority algebra, and the affine Mal'cev algebra. Indeed, from the description of their term operations given in Subsection~\ref{subsec:boolean} it follows that clones of these algebras are even minimal, in the sense that the only proper subclone is the clone of projections. 
A nice example of a minimal Taylor algebra on a three-element domain is the rock-paper-scissors algebra mentioned in Subsection~\ref{subsec:four_types}.
To see that this algebra is minimal Taylor
observe that any term operation 
behaves on any two-element set like the term operation of a two-element semilattice with the same set of essential coordinates.
Therefore, the original operation can be obtained by identifying variables in  any term operation having at least two essential coordinates. The same argument shows that any semilattice, not necessarily two-element, is minimal Taylor. Affine Mal'cev algebras are also minimal Taylor; this can be deduced e.g. from Theorem~\ref{thm:FTA}.

In Subsection~\ref{ssec:mtalgebras_general} we give the basic general theorems that were
proved in~\cite{Brady20:bw} in the context of minimal bounded width algebras. Subsection~\ref{subsec:mt_abs} concentrates on absorption and related concepts in Zhuk's theory. It turns out  that 2-absorbing sets are exactly  projective subuniverses (Theorem~\ref{thm:bin_abs}) and 3-absorbing sets are exactly centers (Theorem~\ref{thm:center_abs}). Subsection~\ref{ssec:mtalgebras_edges} shows that edges substantially simplify in minimal Taylor algebras (Theorem~\ref{thm:minimaledgesquotients}) and gives additional information for minimal edges; in particular, minimal semilattice edges coincide with thin semilattice edges as defined in~\cite{Bulatov16:graph,Bulatov20:local1} 
(Proposition~\ref{prop:thin_semilattice}). 
Finally, in Subsection~\ref{subsec:abs_vs_edges}, we demonstrate a strong interaction between absorption and edges. We show that 2-absorbing subuniverses are exactly subsets that are, in some sense, stable under all the edges (Theorem~\ref{thm:bin_abs_stable}), we provide somewhat weaker interaction between absorbing subuniverses and subsets stable under semilattice and abelian edges (Theorem~\ref{thm:abs_stable}),    we give a common witnessing operation for all of the edges as well as all of the 2- and 3-absorbing subuniverses (Theorem~\ref{thm:unifiededges}), and we show that each such a witnessing operation generates the whole clone of term operations (Theorem~\ref{thm:single_operation_generates}). 
Proofs for this section, including verification of examples, are in Section~\ref{app:mtalgebras}.

\subsection{General facts}\label{ssec:mtalgebras_general}

It is not immediate from the definitions that each Taylor algebra has a minimal Taylor reduct. Nevertheless, this fact easily follows from the characterization of Taylor algebras by means of cyclic operations.

\begin{ThCSrestatable}{proposition}{XXXpropTMexist}{(\hyperlink{ProofProp52}{link to proof})}{prop:TMexist}
Every Taylor algebra has a minimal Taylor reduct. 
\end{ThCSrestatable}

\noindent
Another simple, but important consequence of cyclic operations is the following proposition. The result is slightly more technical than most of the others, but it is in the core of many strong  properties of minimal Taylor algebras. The proposition can be further strengthened, see~\cite[Theorem 4.2.4.]{brady:notes}. 

\begin{ThCSrestatable}{proposition}{XXXtmsubs}{(\hyperlink{ProofProp53}{link to proof})}{ff:TMsubs}
    Let $\alg A$ be a minimal Taylor algebra and $B\subseteq A$ be closed under an operation $f \in \Clo(\alg A)$ such that $B$ together with the restriction of $f$ to $B$ forms a Taylor algebra.
    Then $B$ is a subuniverse of $\alg A$.
\end{ThCSrestatable}

A similar argument based on cyclic operations proves that the class of minimal Taylor algebras is closed under the standard constructions.    
    
\begin{ThCSrestatable}{proposition}{XXXhsp}{(\hyperlink{ProofProp54}{link to proof})}{prop:MinimalVariety}
Any subalgebra, finite power, or quotient of a minimal Taylor algebra is a minimal Taylor algebra.
\end{ThCSrestatable}

\subsection{Absorption} \label{subsec:mt_abs}\label{ssec:mtalgebras_absorption}

The goal of this section is to show that absorbing subsets, which are abundant in general Taylor algebras by Theorem~\ref{thm:AToriginal}, Theorem~\ref{thm:4types}, and Proposition~\ref{prop:center_absors_in_T}, have strong properties in minimal Taylor algebras.
We start with a surprising fact, which clearly fails in general Taylor algebras.

\begin{ThCSrestatable}{theorem}{XXXabssub}{(\hyperlink{ProofThm55}{link to proof})}{thm:AbsIsSubuniv}
      Let $\alg A$ be a minimal Taylor algebra and $B$ an absorbing set of $\alg{A}$. 
      Then $B$ is a subuniverse of~$\alg A$.
\end{ThCSrestatable}

Now we move on to $2$-absorption. 
We have already mentioned in Subsection~\ref{subsec:AT} that projectivity is a stronger form of absorption in Taylor algebras,
but we can go even further.

\begin{definition}
    Let $\alg A$ be an algebra and $B \subseteq A$.
    The set $B$ is a \emph{strongly projective subuniverse} of $\alg A$ if for every $f \in \Clo_n(\alg A)$ and every essential coordinate $i$ of $f$, we have $f(\tuple{a}) \in B$ whenever $\tuple{a} \in A^n$ is such that $a_i \in B$. 
\end{definition}

\noindent The property of being a strong projective subuniverse is indeed very strong.
For example, in any nontrivial clone, a strong projective subuniverse is $2$-absorbing and every binary operation of the clone, except for projections, witnesses the absorption. 
The next theorem states that  strong projectivity in minimal Taylor algebras is equivalent to $2$-absorption,  which in general is a much weaker concept. 

\begin{ThCSrestatable}{theorem}{XXXtwoabs}{(\hyperlink{ProofThm57}{link to proof})}{thm:bin_abs}
The following are equivalent for any minimal Taylor algebra $\algA$ and any $B \subseteq A$.
\begin{enumerate}
    \item[(a)] $B$ 2-absorbs $\algA$.
    \item[(b)] $R(x,y,z)= B(x)\vee B(y)\vee B(z)$  is a subuniverse of $\algA^3$.   
    \item[(c)] $B$ is a projective subuniverse of $\alg A$.
    \item[(d)] $B$ is a strongly projective subuniverse of $\algA$. 
\end{enumerate}
\end{ThCSrestatable}

\noindent
The main value of this theorem is the implication showing that,
in minimal Taylor algebras,
every $2$-absorption, i.e. (a), is as strong as possible (d).
Moreover, (b) provides a nice relational description of $2$-absorption, which collapses the general condition from Proposition~\ref{prop:cubetermblocker} for projectivity to arity $3$.
The implications (d) implies (c) implies (b) implies (a) hold for all Taylor algebras.
The following examples disprove the reverse implications. We leave the verification as an exercise.

\begin{example} 
    The algebra $(\{0,1\};x\vee(y\wedge z))$ is Taylor and $\{1\}$ is a projective universe which is not strongly projective.
    This shows that ``(c) implies (d)'' in Theorem~\ref{thm:bin_abs} fails in Taylor algebras.
    
    Let $t(x,y,z,w)$ be an operation on $\{0,1\}$
    which is equal to $x\vee y\vee z \vee w$ everywhere, 
    except that it is zero on the tuple $(1,0,0,0)$ and its cyclic shifts.
    Then $\alg A =(\{0,1\};t)$ is a Taylor algebra and the $R$ defined for $B= \{1\}$ in (b) of Theorem~\ref{thm:bin_abs} is compatible with $\alg A$. 
    On the other hand $B$ is not a projective subuniverse of $\alg A$ as witnessed by, for example, $t$ and thus (c) of Theorem~\ref{thm:bin_abs} is false.
    
    Finally, the two-element lattice $(\{0,1\}; \wedge, \vee)$ is an example of a Taylor algebra such that (a) of Theorem~\ref{thm:bin_abs} holds, but (b) does not.
\end{example}

The following proposition 
collects some strong and unusual properties of $2$-absorbing subuniverses in minimal Taylor algebras. 
Already the  first item might  be surprising  since the union of two subuniverses is  rarely a subuniverse. 

\begin{ThCSrestatable}{proposition}{PropertiesBinAbs}{(\hyperlink{ProofProp59}{link to proof})}{prop:bin_abs_intersect}
Let $\algA$ be a minimal Taylor algebra.
  \begin{enumerate}
      \item[(1)] If $B \abs_2 \alg A$ and $C\leq \alg A $, then $B \cup C \leq \algA$.
      \item[(2)] If $B \abs_2 \alg A$ and $C \abs \alg A$ by $f$, where $\emptyset \neq B,C \neq A$, then
      \begin{enumerate}
          \item $B\cup C\abs \alg A$ by $f$, and
          \item $B\cap C\neq \emptyset$ and $B\cap  C\abs \alg A$ by $f$.
      \end{enumerate}
      \item[(3)] If $\alg C \abs_2 \alg B \abs_2 \alg A$, then $\alg C \abs_2 \alg A$.
      \item[(4)] $\algA$ has a unique minimal 2-absorbing subalgebra $\alg B$. Moreover, this algebra $\alg B$ does not have any nontrivial 2-absorbing subalgebra.
  \end{enumerate}
\end{ThCSrestatable}

As for absorption of higher arity, we have already shown in Proposition~\ref{prop:center_absors_in_T} that centers are 3-absorbing. 
Next theorem says
that, in minimal Taylor algebras, the converse is true as well.

\begin{ThCSrestatable}{theorem}{thmcenterabs}{(\hyperlink{ProofThm510}{link to proof})}{thm:center_abs}
The following are equivalent for any minimal Taylor algebra $\algA$ and any $B \subseteq A$.
\begin{enumerate}
    \item[(a)] $B$ 3-absorbs $\algA$. 
    \item[(b)] $R(x,y) = B(x) \vee B(y)$ is a subuniverse of $\algA^2$.
    \item [(c)] $B$ is a (Taylor) center of $\alg A$.
    \item[(d)] there exists $\alg C$ with $\Clo(\alg C)\subseteq\Clo(\{0,1\};\maj)$ such that 
    $R(x,y)= B(x) \vee (y=0)$
    is a centrality witness.
\end{enumerate}
Moreover, if $B=\{b\}$, then these items are equivalent to 
\begin{enumerate}
  \item[(e)] $B$ absorbs $\algA$.
  \end{enumerate}
\end{ThCSrestatable}

\noindent
Just like in Theorem~\ref{thm:bin_abs} we have that a relatively weak notion of $3$-absorption
implies a very strong type of centrality which is (d).
Let us investigate  (d) in greater detail. The fact that $R$ is a subuniverse of $\alg A \times \alg C$ translates to the following fact.
To every operation of $\alg A$, say $f$, we can associate an operation $f'\in\Clo(\{0,1\},\maj)$ such that $f(\tuple a)\in B$ whenever $f'(\tuple x) = 1$ and $\tuple x$ is the characteristic tuple of $\tuple a$ with respect to $B$~%
(i.e. $x_i=1$ if and only if $a_i\in B$).
That is, from the viewpoint of ``being outside $B$'' vs. ``being inside $B$'' every operation outputs ``inside $B$'' every time the corresponding operation of $\Clo(\{0,1\};\maj)$ outputs $1$.

In fact, 
there exists a cyclic $t$ in $\alg A$ (say, $p$-ary) such that $R(x,y) = B(x)\vee (y=0)$ is a subuniverse of $(A;t)\times (\{0,1\};\maj_p)$ for every $3$-absorbing $B$.
This translates to a simpler statement:
for every $3$-absorbing $B$ we have $t(\tuple a) \in B$ whenever majority of the $a_i$ belong to $B$, 
and we cannot expect more, as witnessed by the $2$-element majority algebra.
Since $t$ is cyclic it generates the whole clone and, 
for example, item (2) in Proposition~\ref{prop:3_abs_intersect} below becomes obvious.

Item (b) provides a relational description of $3$-absorption, while item (c) provides a connection with the notion of a center~(whether it is Taylor or not).
We now give an example showing that (e) and (a) are not equivalent~%
even in minimal Taylor algebras
if $B$ has more than one element.

\begin{ThCSrestatable}{example}{ExAbsorptionDoesntImplyThreeAbsorption}{(\hyperlink{VerEx511}{link to verification})}{dummylink}
Consider the algebra $\algA = (\{0,1,2\}, m)$ where $m$ is the majority operation such that $m(a,b,c) = a$ whenever $|\{a,b,c\}|=3$. This algebra is minimal Taylor and the set $C = \{0,1\}$ is an absorbing subuniverse of $\alg A$. However, $C$ is not a center of~$\alg A$.
\end{ThCSrestatable}

Finally, we list some strong and unusual properties of $3$-absorbing subuniverses.
They are not as strong as in the case of $2$-absorbing subuniverses,
which is to be expected
since every $2$-absorbing subuniverse is $3$-absorbing but not vice versa.

\begin{ThCSrestatable}{proposition}{propThreeabsintersect}{(\hyperlink{ProofProp512}{link to proof})}{prop:3_abs_intersect}
Let $\algA$ be a minimal Taylor algebra.
  \begin{enumerate}
      \item[(1)] If $B,C\abs_3 \algA$, then $B\cup C\leq \alg A$ and $B \cap C \abs_3 \alg A$.
      \item[(2)] If $\emptyset \neq B,C\abs_3 \algA$ and $B \cap C = \emptyset$, then
          $B^2\cup C^2$ is a congruence on the subalgebra of $\alg A$ with universe $B\cup C$
          and the quotient is term-equivalent to a two-element majority algebra.
     \item[(3)] If $\alg C \abs_3 \alg B \abs_3 \alg A$, then $\alg C \abs_3 \alg A$.
  \end{enumerate}
\end{ThCSrestatable}

\subsection{Edges}\label{ssec:mtalgebras_edges}

The next theorem says that, in minimal Taylor algebras, every ``thick'' edge, in the terminology of~\cite{Bulatov16:graph, Bulatov16:restricted}, is automatically a subuniverse. This property is a simple consequence of the results we have already stated, whereas it was relatively  painful to achieve using the original approach. We additionally obtain that semilattice and majority edges have unique witnessing congruences. 

\begin{ThCSrestatable}{theorem}{ThmMinimalEdgesQuotients}{(\hyperlink{ProofThm513}{link to proof})}{thm:minimaledgesquotients}
Let $(a,b)$ be an edge (semilattice, majority, or abelian) of a minimal Taylor algebra $\alg A$ 
and $\theta$ a witnessing congruence of $\alg E = \Sg_{\algA}(a,b)$. 
\begin{enumerate}
    \item[(a)] If $(a,b)$ is a  semilattice edge, then $\alg E/\theta$ is term equivalent to a two-element semilattice with absorbing element $b/\theta$.
    \item[(b)] If $(a,b)$ is a  majority edge, then 
    $\alg E/\theta$ is term equivalent to a two-element majority algebra.
    \item[(c)] 
    if $(a,b)$ is an  abelian edge, then
    $\alg E/\theta$ is term equivalent to an affine Mal'cev algebra of an abelian group isomorphic to $\zee{m}$ for some positive integer $m$, where $\zee{m}$ denotes the group of integers modulo $m$.
\end{enumerate}
Moreover, a semilatice edge is witnessed by
exactly one congruence of $\alg E$, and that congruence is maximal.
The same holds for majority edges.
\end{ThCSrestatable}

\noindent
For minimal edges we can say a bit more.
If $(a,b)$ is a minimal edge witnessed by $\theta$, a congruence on $\alg E = \Sg_{\alg A}(a,b)$, then $\alg E/\theta$ is simple.
In particular, for abelian minimal edges, 
$\alg E/\theta$ is an affine Mal'cev algebra of a group isomorphic to $\zee{p}$. 
Moreover, such an $\alg E$ has a unique maximal congruence
as shown in the next proposition.
This implies that the type of a minimal edge is unique 
and so is the direction of a semilattice minimal edge and the prime $p$ associated to an abelian minimal edge.

\begin{ThCSrestatable}{proposition}{ProUniqueType}{(\hyperlink{ProofProp514}{link to proof})}{pro:unique-type}
  Let $(a,b)$ be a minimal edge in a minimal Taylor algebra.
  Then $\alg E = \Sg_{\alg A}(a,b)$ has a unique maximal congruence equal to $\mu_{\alg E}$.
  In particular, minimal edges have unique types.
\end{ThCSrestatable}
\noindent
The structure of minimal semilattice edges is especially simple.

\begin{ThCSrestatable}{proposition}{PropThinSemilattice}{(\hyperlink{ProofProp515}{link to proof})} 
{prop:thin_semilattice}
  Let $(a,b)$ be a minimal semilattice edge in a minimal Taylor algebra. 
  Then $\{a,b\}$ is a subuniverse of $\alg A$, so $\Sg_{\alg A}(a,b) = \{a,b\}$ and the witnessing congruence is the equality. 
\end{ThCSrestatable}

\noindent
Unfortunately, majority and abelian edges do not simplify in a similar way; 
see Example~\ref{example:noEZmajority} and Example~\ref{example:noEZabelian}. 
Weaker versions of Proposition~\ref{prop:thin_semilattice} have been developed by Bulatov~(comp. Lemma~12 and Corollary~13 in~\cite{Bulatov16:graph})
to deal with this problem.

\begin{ThCSrestatable}{example}{ExMajorityAintSimple}{(\hyperlink{VerEx516}{link to verification})}{example:noEZmajority} 
Let $A=\{0,1,2,3\}$ and $\alpha$ the equivalence 
relation on $A$ with classes $\{0,2\}$ and $\{1,3\}$.
Define a \emph{symmetric} ternary operation $g$ on $A$ as follows. When two of the inputs to $g$ are equal, $g$ is given by
$g(a,a,a+1)=a$, $g(a,a,a+2)=g(a,a,a+3)=a+2$ (all modulo  $4$)
and when all three inputs to $g$ are distinct, $g$ is given by $g(a,b,c) = d-1 \pmod{4}$
where $a,b,c,d$ are any permutation of $0,1,2,3$.
Then $\algA = (A; g)$ is a minimal Taylor algebra,  $\alpha$ is a congruence on $\algA$, and each of pair of elements in different $\alpha$-classes is a minimal majority edge 
with witnessing congruence $\alpha$.
\end{ThCSrestatable}

\begin{ThCSrestatable}{example}{ExAbelianAintEasy}{(\hyperlink{VerEx517}{link to verification})}{example:noEZabelian}  Let $\algA = (\{a,b,c,d\}, p)$, where $p$ is a Mal'cev operation with the following properties. The operation $p$ commutes with the permutations $\sigma = (a\ c)$ and $\tau = (b\ d)$. The polynomials $+_a = p(\cdot,a,\cdot), +_b = p(\cdot,b,\cdot)$ define abelian groups:
\begin{center}
\begin{tabular}{cc}
\begin{tabular}{c|cccc} $+_a$ & $a$ & $b$ & $c$ & $d$\\ \hline $a$ & $a$ & $b$ & $c$ & $d$\\ $b$ & $b$ & $c$ & $d$ & $a$\\ $c$ & $c$ & $d$ & $a$ & $b$\\ $d$ & $d$ & $a$ & $b$ & $c$ \end{tabular} &
\begin{tabular}{c|cccc} $+_b$ & $a$ & $b$ & $c$ & $d$\\ \hline $a$ & $b$ & $a$ & $d$ & $c$\\ $b$ & $a$ & $b$ & $c$ & $d$\\ $c$ & $d$ & $c$ & $b$ & $a$\\ $d$ & $c$ & $d$ & $a$ & $b$ \end{tabular}
\end{tabular}
\end{center}
Then $\algA$ is a minimal Taylor algebra, with a unique maximal congruence $\theta$ whose congruence classes are $\{a,c\}$ and $\{b,d\}$. Each pair of elements of $\algA$ in different congruence classes of $\theta$ is a minimal abelian edge of $\algA$ with witnessing congruence  $\theta$.
\end{ThCSrestatable}

We can also provide nontrivial information about $\Sg(a,b)$ in the case where  $(a,b)$ is not necessarily an edge, and this information helps in proving Theorem~\ref{thm:single_operation_generates} in the next subsection (and shows that case (d) in Theorem~\ref{thm:4types} is never necessary for two-generated algebras). However,  the following fundamental question remains open:
Is there  a minimal Taylor algebra such that, for some $a,b$, 
neither $(a,b)$ nor $(b,a)$ is an edge?

\subsection{Absorption and edges} \label{subsec:abs_vs_edges}

We start this subsection with a definition that will connect absorption with edges.

\begin{definition}
Let $\alg A$ be an algebra, let $B \subseteq A$ and let $(b,a)$ be an edge.
We say that $B$ is \emph{stable under $(b,a)$} if, for every  witnessing congruence $\theta$ of $\Sg_{\alg A}(b,a)$ such that  $b/\theta$ intersects $B$, each  $\theta$-class intersects $B$.
\end{definition}

\noindent
As the next theorem states, stability under every edge can be added as an additional equivalent condition in  Theorem~\ref{thm:bin_abs}. This direct connection between absorption, which is a global property, to the local concepts in Bulatov's theory is among the most surprising phenomena that the  authors have encountered in this work. 

\begin{ThCSrestatable}{theorem}{BinAbsIsStable}{(\hyperlink{ProofThm519}{link to proof})} 
 {thm:bin_abs_stable}
The following are equivalent for any minimal Taylor algebra $\algA$ and any $B \subseteq A$.
\begin{enumerate}
    \item[(a)] $B$ 2-absorbs $\alg A$. 
    \item[(b)] $B$ is stable under all the edges.
\end{enumerate}
\end{ThCSrestatable}

\noindent
The implication from (b) to (a) does not require the full strength of stability for semilattice and majority edges. It is enough to require that for a minimal semilattice or a majority edge $(b,a)$  it is never the case that $b/\theta \subseteq B$ and $a/\theta \cap B = \emptyset$, where $\theta$ is the edge-witnessing congruence of $\Sg(b,a)$ (which is the equality relation on $\{a,b\}$ in case of semilattice edges). The following example shows that stability under abelian edges cannot be significantly weakened.

\begin{ThCSrestatable}{example}{ExStabilityCantBeSimplified}{(\hyperlink{VerEx520}{link to verification})}{ex:abelian_closed_not_enough}   We consider the four-element algebra $\algA = (\{0,1,2,*\},\cdot)$ with binary operation $\cdot$ given by
\begin{center}
\begin{tabular}{c|cccc}
$\cdot$ & $0$ & $1$ & $2$ & $*$\\
\hline
$0$ & $0$ & $2$ & $1$ & $*$\\
$1$ & $2$ & $1$ & $0$ & $2$\\
$2$ & $1$ & $0$ & $2$ & $1$\\
$*$ & $*$ & $2$ & $1$ & $*$
\end{tabular}
\end{center}
Then $\algA$ is a minimal Taylor algebra, with a semilattice edge $(0,*)$, with $\{0,1,2\}$ an affine subalgebra, and with a congruence $\theta$ corresponding to the partition $\{0,*\},\{1\},\{2\}$ such that $\algA/\theta$ is affine. The set $\{*\}$ is stable under semilattice and majority edges and there is no minimal abelian edge $(*,a)$ with $a \neq *$. But $\{*\}$ is not an absorbing subalgebra of $\algA$.
\end{ThCSrestatable}

For absorption of higher arity the connection to edges is not as tight as for 2-absorption. Nevertheless, one direction still works, and both directions work for singletons.

\begin{ThCSrestatable}{theorem}{AbsIsStable}{(\hyperlink{ProofThm521}{link to proof})} {thm:abs_stable}
  Any absorbing set of a minimal Taylor algebra $\alg A$ is stable under semilattice and abelian edges.
  Moreover, for any $b \in A$ the following are equivalent.
  \begin{itemize}
      \item[(a)] $\{b\}$ absorbs $\algA$.
      \item[(b)] $\{b\}$ is stable under semilattice and abelian edges.
  \end{itemize}
\end{ThCSrestatable}

Stability under semilattice edges for the implication from (b) to (a) can be again replaced by the requirement that there is no minimal semilattice edge $(b,a)$ with $b \in B$ and $a \not\in B$. Example~\ref{ex:abelian_closed_not_enough} shows that this is not the case for abelian edges.

The following example shows that  the implication from (b) to (a) does not hold for non\-singleton subuniverses.

\begin{example} 
Consider the algebra $\algA = (\{0,1,2\},m)$ where $m$ is the majority operation such that $m(a,b,c)=2$ whenever $|\{a,b,c\}|=3$. This algebra is minimal Taylor, 
every pair of distinct elements forms a subuniverse, and every pair is a minimal majority edge. So there are no semilattice or abelian edges. However, the subuniverse $\{0,1\}$ is not absorbing.
\end{example}

An important fact for the edge approach is that semilattice, majority, and Mal'cev operations coming from edges can be unified.
In minimal Taylor algebras, a simple consequence of the already stated results is that we not only have a common ternary witness for all the edges but also for all the binary and ternary absorptions.  

\begin{ThCSrestatable}{theorem}{UnifiedEdges}{(\hyperlink{ProofThm523}{link to proof})} {thm:unifiededges}
  Every minimal Taylor algebra $\alg A$ has a ternary term operation $f$ such that if $(a,b)$ is an edge witnessed by $\theta$ on $\alg E =\Sg_{\alg A}(a,b)$,
  then
  \begin{itemize}
    \item if $(a,b)$ is a semilattice edge, then $f(x,y,z)=x\vee y\vee z$ on $\alg E/\theta$~(where $b/\theta$ is the top);
    \item if $(a,b)$ is a majority edge, then $f$ is the majority operation on $\alg E/\theta$~(which has two elements);
    \item if $(a,b)$ is an abelian edge, then $f(x,y,z) = x-y+z$  on $\alg E/\theta$;
    \item $f$ witnesses all the 3-absorptions $B \abs_3 \alg A$;
    \item any binary operation obtained from $f$ by identifying two arguments  witnesses all the  2-absorptions $B \abs_2 \alg A$.
  \end{itemize}
\end{ThCSrestatable}

\noindent
In fact, any ternary operation $f$ defined from a cyclic term operation $t$ of odd arity $p \geq 3$ by
$$
f(x,y,z) = t(\underbrace{x,x, \dots, x}_{k \times },
         \underbrace{y,y, \dots, y}_{l \times },
         \underbrace{z,z, \dots, z}_{m \times }),
$$         
where $k+l, l+m, k+m > p/2$, satisfies all the items in Theorem~\ref{thm:unifiededges} except possibly the third one (which can be obtained by picking $k$, $l$, and $m$ a bit more carefully). 

We finish this section with a theorem stating that any ternary witness of edges generates the whole clone of the algebra. In particular, the number of minimal Taylor clones on a domain of size $n$ is at most $n^{n^3}$.

\begin{ThCSrestatable}{theorem}{SingleOperationGenerates}{(\hyperlink{ProofThm524}{link to proof})}{thm:single_operation_generates}
If $\algA$ is a minimal Taylor algebra, then $\Clo(A;f) = \Clo(\alg A)$ for any operation $f$ satisfying the first three items in Theorem~\ref{thm:unifiededges}.
\end{ThCSrestatable}

\section{Omitting types}
\label{sec:omitting_types}

In this section we consider classes of algebras whose graph only contains edges of certain types. 
We say that an algebra is \emph{$\types{a}$-free} if it has no abelian edges. More generally, an algebra is \emph{$\types{x}$-free} or is \emph{$\types{xy}$-free}, where $\types x,\types y\in\{\text{($\types{a}$)belian, ($\types{m}$)ajority}$, $(\text{$\types s$)emilattice}\}$ if it has no edges of type $\types x$ (of types $\types{x,y}$).

It turns out that within minimal Taylor algebras these ``omitting types'' conditions are often equivalent to important properties of algebras. In the theorems below we prove the equivalence of the following four types of conditions: (i) the absence of edges of a certain type (equivalently, minimal edges of the same type); (ii) properties of absorption and the four types in Zhuk's approach;
(iii) the existence of a certain special term operations; (iv) algorithmic properties of the CSP. 
Theorems in this section are consequences of the theory we have already built in the previous section and known results (see \cite{Barto17:polymorphisms}). 
The proofs are in Section~\ref{sec:proofsomittingtypes}.

The first theorem concerns the class of algebras omitting abelian edges. Numerous characterizations of this class are known for general algebras and we do not add a new one, but we state the characterization for comparison with the other classes. 
In order to state a characterization in terms of identities we recall that an operation $f$ is a \emph{weak near unanimity} operation (or \emph{wnu} for short) if it satisfies $
    f(y,x,\dots,x)=f(x,y,x,\dots,x) = \dots
    =f(x,\dots,x,y)$ for every $x,y$ in the algebra.

\begin{ThCSrestatable}{theorem}{ThmAFree}{(\hyperlink{ProofThm61}{link to proof})}{thm:a-free}
The following are equivalent for any algebra $\algA$.
\begin{enumerate}
    \item[(i)] $\algA$ is $\types{a}$-free.
    \item[(ii)] 
      No subalgebra of $\algA$ has a nontrivial abelian quotient, i.e., 
      no subalgebra of $\algA$ falls into case (c) in Theorem \ref{thm:4types}. 
    \item[(iii)] $\alg A$ has a wnu term operation of every arity $n \geq 3$.
    \item[(iv)] $\algA$ has bounded width. 
\end{enumerate}
\end{ThCSrestatable}

\noindent
Minimal Taylor algebras omitting other types of edges \emph{do} have significantly stronger properties than general Taylor algebras omitting those edges.
Minimal  $\types{s}$-free algebras are exactly those for which option (a) in Theorem~\ref{thm:4types} does not hold, and that have  few subpowers~\cite{Bulatov16:restricted}. 
Few subpowers algebras  can be characterized by the existence of an edge term operation~\cite{Berman10:varieties} in general. In minimal Taylor algebras, the second strongest edge operation always exists -- the 3-edge operation defined by the identities $e(y,y,x,x)=e(y,x,y,x)=e(x,x,x,y)=x$.
This is significant, because the exponent in the running time of the few subpowers algorithm depends on the least $k$ such that the algebra has a $k$-edge term operation. The number $3$ here is best possible: a $2$-edge operation is the same as a Mal'cev operation appearing in Theorem~\ref{thm:as-free}.

\begin{ThCSrestatable}{theorem}{ThmSFree}{(\hyperlink{ProofThm62}{link to proof})}{thm:s-free}
The following are equivalent for any minimal Taylor algebra $\algA$.
\begin{enumerate}
    \item[(i)] $\algA$ is $\types{s}$-free.
    \item[(ii)] No subalgebra of $\algA$ has a nontrivial 2-absorbing subuniverse, i.e.,
      no subalgebra of $\algA$ falls into case (a) in Theorem~\ref{thm:4types}.
    \item[(iii)] $\algA$ has a $3$-edge term operation.
    \item[(iv)] $\algA$ has few subpowers.
\end{enumerate}
\end{ThCSrestatable}

\noindent
For the remaining omitting-single-type condition, $\types{m}$-freeness, we do not provide a natural condition in terms of identities, and we are not aware of algorithmic implications of this condition. Nevertheless, it can be characterized by means of absorption.

\begin{ThCSrestatable}{theorem}{ThmMFree}{(\hyperlink{ProofThm63}{link to proof})}{thm:m-free}
The following are equivalent for any minimal Taylor algebra $\algA$.
\begin{enumerate}
    \item[(i)] $\algA$ is $\types{m}$-free.  
    \item[(ii)] Every center (3-absorbing subuniverse) of any $\alg B \leq \alg A$ 2-absorbs $\alg B$,
      i.e., case (b) implies case (a) in Theorem~\ref{thm:4types} in all the subalgebras of $\alg A$.
    \item[(ii')] Every subalgebra of  $\algA$ has a unique minimal 3-absorbing subuniverse.
    \end{enumerate}
\end{ThCSrestatable}

\noindent
Surprisingly, 
if along with $\types{m}$-freeness we also limit the type of abelian edges allowed in an algebra, the resulting condition is equivalent to the existence of a binary commutative term operation. 
This is interesting to us (the authors), since we did not regard the existence of a commutative term operation to be a natural requirement for the CSP (cf.~\cite{Barto17:polymorphisms}).
We call an abelian edge $(a,b)$ a \emph{$\zee{2}$-edge} if the corresponding affine Mal'cev algebra $\Sg(a,b)/\theta$ is isomorphic to the affine Mal'cev algebra of $\zee{2}$. 

\begin{ThCSrestatable}{theorem}{ThmMTwoFree}{(\hyperlink{ProofThm64}{link to proof})}{thm:m2-free}
The following are equivalent for any minimal Taylor algebra $\algA$.
\begin{enumerate}
    \item[(i)] $\algA$ is $\types{m}$-free and has no $\zee{2}$-edges.
    \item[(iii)] $\algA$ has a binary commutative term operation.
    \item[(iii')] $\Clo(\alg A)$ can be generated by a collection of binary operations.
\end{enumerate}
\end{ThCSrestatable}

Properties of minimal Taylor algebras having edges of only one type  can be derived as conjunctions of the properties stated above. For two of these cases, $\types{sm}$-free and $\types{as}$-free, we provide additional information.

Minimal Taylor $\types{am}$-free algebras are exactly those which have wnu operations of every arity $n \geq 2$.
These are exactly the minimal \emph{spirals}  in the terminology of~\cite{Brady20:bw} and a significant property is that for every $(a,b)$ such that neither $(a,b)$ nor $(b,a)$ is a minimal semilattice edge, there is a surjective homomorphism from $\Sg\{a,b\}$ onto the (three-element) free semilattice on two generators.

The $\types{sm}$-free minimal Taylor algebras are those where cases (a) and (b) in Theorem~\ref{thm:4types} do not occur. Additionally, these are exactly the hereditarily absorption free algebras studied in~\cite{Barto15:malcev}  and, also, the algebras with a Mal'cev term operation -- a type of operation that played a significant role in the CSP~\cite{Bulatov06:maltsev}.

\begin{ThCSrestatable}{theorem}{ThmSMFree}{(\hyperlink{ProofThm65}{link to proof})}{thm:sm-free}
The following are equivalent for any minimal Taylor algebra $\algA$.
\begin{enumerate}
    \item[(i)] $\algA$ is $\types{sm}$-free.
    \item[(ii)] No subalgebra of  $\algA$ has a nontrivial absorbing subuniverse.    
    \item[(iii)] $\algA$ has a Mal'cev term operation.
\end{enumerate}
\end{ThCSrestatable}

\noindent
Finally, the $\types{as}$-free algebras are those where cases (a) and (c) in Theorem~\ref{thm:4types} do not occur and those that have bounded width and few subpowers. It is known~\cite{hobby88:tct,Berman10:varieties} that the latter property in general implies having a near-unanimity term operation of some arity. Surprisingly, in minimal Taylor algebras, the arity goes down directly to three. In the algorithmic language, these algebras have strict width two~\cite{Feder98:monotone,Barto17:polymorphisms}.

\begin{ThCSrestatable}{theorem}{ThmASFree}{(\hyperlink{ProofThm66}{link to proof})}{thm:as-free}
The following are equivalent for any minimal Taylor algebra $\algA$.
\begin{enumerate}
    \item[(i)]  $\algA$ is $\types{as}$-free.
    \item[(iii)] $\alg A$ has a near unanimity term operation.
    \item[(iii')] $\alg A$ has a majority term operation.    
    \end{enumerate}
\end{ThCSrestatable}

\newpage
\part{Technical details}

\renewcommand{\ooo}[1]{}
\renewcommand{\rrr}{{\normalfont\bfseries (Restated)}}

\section{Preliminaries} \label{app:prelim}

In this section  we list definitions and facts that will be, often implicitly, used in the proofs.

\subsection{Relations} \label{subsec:prelim_relations}

A relation on $A$ is a subset of $A^n$, but we often work with more general ``multisorted'' relations $R \subseteq A_1 \times A_2 \times \dots \times A_n$.
 We call such an $R$  {\em proper} if $R \neq A_1 \times \dots \times A_n$ and
 \emph{nontrivial} if it is nonempty and proper. Tuples are written in boldface
 and components of $\tuple{x} \in A_1 \times \dots \times A_n$ are denoted
$x_1, x_2, \dots$. 
Both $\tuple{x} \in R$ and $R(\tuple{x})$ are used to denote the fact that $\tuple{x}$ is in $R$.
 The projection of $R$ onto the coordinates $i_1, \dots, i_k$ is denoted $\proj_{i_1, \dots, i_k}(R)$. The relation $R$ is \emph{subdirect}, denoted $R \subseteq_{sd} A_1 \times \dots \times A_n$, if  $\proj_{i}(R) = A_i$ for each $i$.
  We call  $R$ {\em redundant}, if there exist coordinates $i \neq j$ 
  such that $\proj_{ij}(R)$ is a graph of bijection from $A_i$ to $A_j$; otherwise $R$ is \emph{irredundant}. A relation $R \subseteq A^n$ is \emph{reflexive} if $(a,a, \dots, a)\in R$ for each $a \in A$. Note that for a redundant reflexive relation, some $\proj_{ij}(R)$, $i \neq j$, is the equality relation.

For a subset $R \subseteq A^X$, the projection of $R$ onto a set of coordinates $I \subseteq X$ is also denoted $\proj_I(R)$; it is a subset of $A^I$.

  We say that a set of relations $\mathcal{R}$ {\em pp-defines} $S$ if $S$ can be defined from $\mathcal{R}$ by a primitive positive formula with parameters, that is, using the existential quantifier, relations from $\mathcal{R}$, the equality relation, \textbf{and the singleton unary relations}.

For binary relations we write $-R$ instead of $R^{-1}$ and $R+S$ for the relational composition of $R$ and $S$, that is $R+S = \{(a,c) \st (\exists b) \, R(a,b) \wedge R(b,c)\}$. 
For a unary relation $B$ we write $B+S$ to denote the set $\{c\st (\exists b)\, B(b) \wedge S(b,c)\}$ and 
if $B$ is a singleton we often write $b+S$ instead of $b +S$.
Also, we set $R-S = R + (-S)= R\circ S^{-1}$. 
A  relation $R \subseteq A \times B$ is \emph{linked} if $(R-R)+(R-R)+ \dots + (R-R)$ is equal to $(\proj_1(R))^2$ for some number of summands. 
In other words, $R$ is connected when viewed as a bipartite graph between disjoint copies of $A$ and $B$, with possible isolated vertices if $R$ is not subdirect.
The \emph{left center}  of $R \subseteq A \times B$ is the set $\{a \in A \st a + R = B\}$. 
If $R$ has a nonempty left center, it is called \emph{left central}. 
\emph{Right center} and \emph{right central} relations are defined analogically. 
A relation is \emph{central} if it is left central and right central.
Note that $R+S$, $-R$, and the left (right) center of $R$ are pp-definable from $\{R,S\}$, e.g.,
the left center of $R$ is defined by the formula $R(x,b_1) \wedge R(x,b_2) \wedge \dots \wedge R(x,b_n)$, where the $b_i$ are selected so that $\{b_1, \dots, b_n\} = B$.

\subsection{Algebras} \label{subsec:prelim_algebra}

Algebras, i.e. structures with a purely functional signature, will be denoted by boldface capital letters (e.g., $\alg A$) and their universes (also called domains) typically by the same letter in the plain font (e.g., $A$). The basic general algebraic concepts, such as subuniverses, subalgebras, products, and quotients modulo congruences  are used in the standard way (see, e.g.~\cite{Bergman12:UAbook}). An algebra is \emph{nontrivial} if it has at least two elements, otherwise it is \emph{trivial}. We use $B \leq \alg A$ to mean that $B$ is a subuniverse of $\alg A$.  We often abuse notation and write $t$ both for a term and the induced term operation $t^{\alg A}$ of an algebra $\alg A$.  
By a \emph{subpower} we mean a subuniverse (or a subalgebra) of a finite power. Subpowers are the same as \emph{invariant relations} and we may also call them \emph{compatible relations}. The set of all subpowers  is denoted $\Inv(\alg A)$. The subuniverse (or the subalgebra) of $\algA$ generated by a set $X \subseteq A$ is denoted 
$\Sg_{\algA}(X)$ or $\Sg_{\algA}(x_1, \dots, x_n)$ when $X = \{x_1, \dots, x_n\}$.
An algebra is \emph{simple} if it has only the \emph{trivial congruences} -- the equality relation and the full relation.

An operation $f$ on $A$ is \emph{idempotent} if $f(a,a, \dots, a)=a$ for every $a \in A$. An algebra is idempotent if all of its operations are. Recall that
\textbf{all theorems in this paper concern algebras that are finite and idempotent}.

An operation $f$ on $A$ is a \emph{semilattice operation} if it is binary, commutative, idempotent, and associative; \emph{cyclic operation} if $f(a_1, \dots, a_n) = f(a_2, \dots, a_n,a_1)$ for all $\tuple{a} \in A^n$; \emph{weak near unanimity (wnu)  operation} if $f(b,a, \dots, a) = f(a,b,a, \dots, a) = \dots = f(a, \dots, a,b)$ for all $a,b \in A$;
\emph{near unanimity operation} if it is a wnu and $f(b,a, \dots, a)=a$; a \emph{majority operation} if it is a ternary near unanimity operation; a \emph{Mal'cev operation} if $f(b,a,a)=f(a,a,b)=b$ for all $a,b \in A$; and \emph{affine Mal'cev operation} if $f(a,b,c)=a-b+c$ for all $a,b,c \in A$, where $+,-$ are computed with respect to an abelian group structure on $A$. For an odd $n$, the $n$-ary majority operation on a two-element domain is denoted $\maj_n$, i.e., $\maj_n(a_1, \dots, a_n)=a$ iff the majority of the $a_i$ is $a$.
A \emph{semilattice} is an algebra with a semilattice operation (and no other operations); an \emph{affine Mal'cev} algebra is an algebra with an affine Mal'cev operation; and an \emph{affine algebra} is an algebra 
whose term operations are exactly the idempotent term operations of a module over a unital ring.

A \emph{(function) clone} is a set of operations $\clone{C}$ on a set $A$ which contains all the projections $\proj^n_i$ (the $n$-ary projection to the $i$-th coordinate, i.e., $\proj^n_i(\tuple{a})=a_i$) and is closed under composition, i.e., $f(g_1, \dots, g_n) \in \clone{C}$ whenever $f \in \clone{C}$ is $n$-ary and $g_1$, \dots, $g_n \in \clone{C}$ are all $m$-ary, where
  $f(g_1, \dots, g_n)$ denotes the operation defined by 
  $f(g_1(x_1, \dots, x_m), \dots,  g_n(x_1, \dots, x_m))$.
By $\Clo(\algA)$ ($\Clo_{n}(\algA)$, respectively), we denote the clone of all term operations (the set of all $n$-ary term operations, respectively) of $\algA$. 
An algebra $\alg B$ is a \emph{reduct} of $\alg A$ if they have the same universe $A=B$ and $\Clo(\alg B) \subseteq \Clo(\algA)$. 
Algebras $\alg A$ and $\alg B$ are \emph{term-equivalent} if each of them is a reduct of the other, i.e., $\Clo(\alg A) = \Clo(\alg B)$.

An algebra $\alg A$ is  \emph{polynomially complete} if  every operation on~$A$ is in  the clone generated by~$\alg A$ together with the constant operations. Note that $\alg A$ is polynomially complete if and only if it has no proper reflexive irredundant subpowers.

A coordinate $i$ of an operation $f: A^n \to A$ is \emph{essential} if $f$ depends on the $i$th coordinate, i.e., $f(\tuple{a}) \neq f(\tuple{b})$ of some tuples $\tuple{a},\tuple{b} \in A^n$ that differ only at the $i$th coordinate.

\subsection{Star and cyclic composition}

If $t,s$ are  operations on the same set, of arities $p$ and $q$, then the {\em star composition} of $t$ and $s$ is defined by
\begin{equation*}
  t\big(s(x_1,x_{p+1},\dotsc,x_{qp-p+1}),\dotsc,s(x_p,x_{2p},\dotsc,x_{qp})\big).
\end{equation*}
The star composition of a $p$-ary cyclic operation and a $q$-ary cyclic operation is  a cyclic operation of arity $pq$. 
 
If $t$ is a cyclic operation and $s$ is any operation of the same arity $p$, then the {\em cyclic composition} of $t$ and $s$ is defined by
\begin{equation*}
  t\big(s(x_1,\dotsc,x_p),s(x_2,\dotsc,x_p,x_1),\dotsc,s(x_p,x_1,\dotsc,x_{p-1})\big)
\end{equation*}
and it is a cyclic operation of arity $p$.

\subsection{Subpowers and pp-definitions} \label{subsec:app_prelim_pp}

Our proofs heavily exploit the fact that subpowers of algebras are closed under pp-definitions. 

We have already mentioned in Subsection~\ref{subsec:prelim_relations} that the relational composition $R+S$, the inverse relation $-R$, and  left and right centers of a binary relation $R$ are pp-definable from $\{R,S\}$. Therefore if $R$ and $S$ are subuniverses of $\alg A^2$, then $R+S$ is a subuniverse of $\alg A^2$, and  $-R$ as well as the centers of $R$ are subuniverses of $\alg A$ (and this also applies to the multisorted setting when $R$ and $S$ are subuniverses of $\alg A \times \alg B$). We now discuss further such observations that we use often in this paper.

If $R$ is a subuniverse of $\alg A^n$, then the projection onto a set of coordinates $I \subseteq \{1,2, \dots, n\}$ is a subpower of $\alg{A}$ as well. We sometimes \emph{fix} some coordinate $i$ to a subuniverse $B \leq \alg A$ before projecting, i.e., we consider the relation $R'(x_1, \dots, x_n) \equiv R(x_1, \dots, x_n) \wedge B(x_i)$. 

Let $R$ now be a subdirect subuniverse of $\alg A^2$. Recall that $R$ is linked if $S = (R-R) + \dots + (R-R)$ is equal to $A^2$ when we take a sufficiently large number of summands. In general, $S$ is an equivalence relation on $A$ which is a subuniverse of $\alg A^2$ -- a congruence of $\alg A$. In particular, if $\alg A$ is simple, then $S$ is either $A^2$ (so $R$ is linked) or the equality relation, in which case $R$ is a graph of a permutation $A \to A$ -- an automorphism $\alg A \to \alg A$. So, for a simple algebra $\alg A$, an irredundant binary relation $R \leq_{sd} \alg A^2$ is necessarily linked.

One type of subpower is of particular importance. Fix $n$ and 
consider the subalgebra $\alg F$ of $\alg A^{A^n}$ with universe $F = \Clo_n(\alg A)$, the set of $n$-ary operations on $A$. Thus tuples in $F$ are $n$-ary term operations of $\alg A$, a coordinate of a tuple in $F$ is an element $\tuple{a}$ of $A^n$, and the $\tuple{a}$-th component of a tuple $f \in F$ is the value $f(\tuple{a})$. The algebra $\alg F$ is isomorphic to the free algebra for $\alg A$ over an $n$-element set of generators.

\subsection{Absorption and pp-definitions} \label{subsec:app_prelim_absorption}

Pairs of subpowers of $\alg A$ with $\alg B \abs \alg C$ are also closed under pp-definitions in the following sense (this folklore fact was recorded as Lemma 2.9 in ~\cite{Barto16:deciding_abs}).

\begin{lemma}  \label{lem:absorption_and_pp}
  Assume that a subpower $R$ of $\algA$ is defined by 
  \[
    R(x_1, \dots, x_n) \equiv  \exists y_1,\dots,y_m:\, R_1(\sigma_1)\wedge
    \dots\wedge R_k(\sigma_k),
  \]
  where $R_1,\dots,R_k$ are subpowers of $\algA$ regarded as predicates
   and $\sigma_1$,\dots,$\sigma_k$ stand for sequences of (free or bound) variables. 
  Let $S_1,\dots,S_k$ be subpowers of $\algA$ such that $S_i\abs \alg{R}_i$  for all $i$. Then the subpower
  \[
    S(x_1, \dots, x_n) \equiv  \exists y_1,\dots,y_m:\, S_1(\sigma_1) \wedge
    \dots \wedge S_k(\sigma_k),
  \]
  absorbs $\alg{R}$. Moreover, if all the absorptions $S_i \abs \alg{R}_i$ are witnessed by $t$, then so is $S \abs \alg{R}$.
\end{lemma}

In particular, if $R \leq_{sd} \alg A^2$ and $B$ absorbs  $\alg A$, then $B+R$ and $B-R$ absorb $\alg A$ as well.

It is also useful to observe that absorption is transitive: if $\alg C \abs \alg B$ by $s$ and $\alg B \abs \alg A$ by $t$, then $\alg C \abs \alg A$ by the star composition of $s$ and $t$.

\subsection{Relational descriptions}

Many algebraic notions we deal with in this paper have their relational counterparts. We have already stated such a characterization for projectivity in Proposition~\ref{prop:cubetermblocker} and we have used a relational description of abelianess as a definition for this concept in Definition~\ref{def:abelian}. Now we state two other helpful facts.

\noindent 

The first one is a characterization of absorption by means of so called $B$-essential relations~(see e.g. Proposition~2.14 in \cite{Barto16:deciding_abs}). A relation $R \subseteq A^n$ is \emph{$B$-essential} if $R$ does not intersect $B^n$ but  every projection of $R$ onto all but one of the coordinates intersects the corresponding power of $B$. The characterization is that $B \abs_n \alg A$ if, and only if, there are no $B$-essential subuniverses of~$\alg A^n$. We state this fact as follows.

\begin{proposition} \label{prop:abs_blockers}
    Let $\alg A$ be an algebra and $B\leq \alg A$. 
    Then $B \abs_n \alg A$ if and only if 
    for every $\tuple a^1,\dotsc, \tuple a^n\in A^n$ such that $\tuple a^i_j\in B$ for $i\neq j$ 
    we have $\Sg_{\alg A^n}(\tuple a^1,\dotsc, \tuple a^n)\cap B^n\neq\emptyset$.
\end{proposition}

The second proposition characterizes strongly projective subuniverses.

\begin{proposition} \label{prop:rel_descr_strong_blocking}
    Let $\alg A$ be an algebra and $B \subseteq A$.
    Then $B$ is a strongly projective subuniverse of $\alg A$ if and only if
    the relation $R(x,y,z)= B(x)\vee (y = z)$ is a subuniverse of $\algA^3$.
\end{proposition}

\begin{proof}
    For the backward implication let  $f$ be an $n$-ary term operation of $\alg A$ and say, without loss of generality, that the first coordinate is essential
    as witnessed by tuples $(c,c_2,\dotsc,c_n)$ and $(c',c_2,\dotsc,c_n)$. 
    Take $(b,a_2,\dotsc,a_n) \in B\times A^{n-1}$ and note that  
    $R(b,c,c')$, $R(a_2,c_2,c_2), \dotsc,R(a_n,c_n$, $c_n)$.
    Therefore 
    \begin{equation*}
    R(f(b,a_2,\dotsc,a_n), f(c,c_2,\dotsc,c_n),f(c',c_2,\dotsc,c_n))\end{equation*}
     and, by the choice of $c,c',c_2,\dotsc,c_n$, we get $f(b,a_2,\dotsc,a_n)\in B$, as required.
    
    For the forward implication we proceed by way of contradiction and suppose that an application of an operation $f$ to triples from $R$ produces a triple outside.
    The resulting triple does not have an element of $B$ at the first position, therefore, by the assumption, all the input triples that \emph{have} an element of $B$ on the first position appear on  inessential coordinates of $f$. 
    The remaining triples have the same element on the second and third positions, therefore so does the resulting triple, a contradiction.
\end{proof}

\section{Proofs for Section~\ref{sec:taylor}: Taylor algebras} \label{app:taylor}

This section  contains proofs of the claims made in Section~\ref{sec:taylor}.

\subsection{Absorption theorem}

We start by formally stating the already known improvements and refinements of the proof of the absorption theorem (Theorem~\ref{thm:AToriginal}) discussed in Subsection~\ref{subsec:AT}. Some of these improvements do not seem to be recorded in the literature or are hidden inside proofs of different results. As a corollary we obtain an improved version of the absorption theorem, Corollary~\ref{cor:newAT}, which seems to be new. Afterwards, we prove the two refinements of the theory which were promised in Subsection~\ref{subsec:AT}.

The first step toward Corollary~\ref{cor:newAT} can be divided into two sub-steps. The first sub-step is  Lemma 2.7 in~\cite{Barto15:malcev}.

\begin{proposition} \label{prop:transitive-or-projective}
 Every algebra has a transitive operation or a nontrivial projective subuniverse.    
\end{proposition}

\noindent
The other sub-step has a short proof via Taylor operations, see the proof of Lemma 3.4 in~\cite{Barto15:malcev}. An elementary proof via clone homomorphisms is also available, but unpublished. It will appear in another paper. 

\begin{proposition} \label{prop:projective-implies-binabs}
 Every projective subuniverse of a Taylor algebra is 2-absorbing.
 \end{proposition}

\noindent
We remark that the argument for Proposition~\ref{prop:projective-implies-binabs} via cyclic operations sketched in Subsection~\ref{subsec:AT} is not ``fair'' since Theorem~\ref{thm:cyclic} heavily uses the absorption theorem, which uses this proposition in its proof. 

The second step is a purely relational fact and we provide a short proof.

\LinkedImpliesCenterOld*

\begin{proof} \hypertarget{ProofProp41}{}
   Suppose that $R$ is not left central.
   Since $R$ (and then $-R$ as well) is linked and subdirect, we have $(-R+R)+(-R+R) + \dots = C^2$ for a sufficiently large number of summands.
   
   If $-R+R \neq C^2$, then we first preprocess $R$ by taking $(-R+R)+(-R+R)+ \dots$ sufficiently many times so that $-R+R = C^2$ and $R$ is still proper. If it became left central, we are done since $(-R+R)+\dots$ is symmetric.
   
   For any set $D=\{c_1, \dots, c_k\} \subseteq C$ consider the binary pp-definable relation $S_D(x,y)$ expressing ``$x$ , $y$, and all the $c_i$ have a common neighbor'', that is $S_D(x,y) \equiv (\exists a) R(a,x) \wedge R(a,y) \wedge R(a,c_1) \wedge \dots \wedge R(a,c_k)$. Since $-R+R=C^2$, the set $S_{\emptyset}$ is equal to $C^2$, and since $R$ is not left central, the set $S_C$ is empty.
    Take a maximal $D$ so that $S_D=C^2$ and take any $E \subseteq C$ with $D \subseteq E$ and $|E \setminus D|=1$. Now observe that $S_E$ is a proper symmetric subdirect relation on $C$ whose  center is nonempty since it contains $E$.
\end{proof}

The third step shows how \emph{left} central relations together with transitive operations on the \emph{right} side produce absorption. 
The argument is in the final part of the proof of Theorem 2.11 in~\cite{Barto12:cyclic}.

\begin{proposition} \label{prop:trans-and-center-implies-absorption}
Let $\alg A$ and $\alg C$ be algebras and suppose $R \leq_{sd} \alg A \times \alg C$ is left-central. 
If $\alg C$ has a transitive term operation $t^{\alg C}$, then the left center of $R$ absorbs $\alg{A}$ by $t^{\alg A}$. 
\end{proposition}

\begin{proof}
 Let $t^{\alg C}$ be a transitive term operation of $\alg C$, say of arity $n$, and let $B$ be the left center of~$R$. We need to show that, for any $\tuple{a} \in A^n$ such that  $|\{i: a_i \in B\}| \geq n-1$, the result $t^{\alg A}(\tuple{a}) = b$ is in $B$, i.e., we need to show that for every $c \in C$ we have $(b,c) \in R$. Let $i$ be the only coordinate for which $a_i \not\in B$ (or take $i$ arbitrary if there is none such) and take $d$ such that $(a_i,d) \in R$. By transitivity of $t^{\alg C}$, there exists $\tuple{d} \in C^n$ with $d_i=d$ and $t^{\alg C}(\tuple{d})=c$. 
 All the component pairs $(a_j,d_j)$ are in $R$ since $B$ is the left center of $R$, so $(t^{\alg A}(\tuple{a}),t^{\alg C}(\tuple{d})) = (b,c)$ is in $R$, as required. 
\end{proof}

Note that the three steps combined
already give us the absorption theorem (Theorem~\ref{thm:AToriginal}). We prove a slightly stronger version where only $\alg C$ is assumed to be Taylor.

\begin{corollary} \label{cor:AT}
Suppose $R \leq_{sd} \alg A \times \alg C$ is proper and linked, and $\alg C$ is  Taylor.
 Then $\alg A$ or $\alg C$ has a nontrivial absorbing subuniverse.
\end{corollary}

\begin{proof}
 Proposition~\ref{prop:linked-implies-central} implies that there is a left central subuniverse $S$ of $\alg A \times \alg C$ or $\alg C^2$. The left center $B$ of $S$ (which is a subuniverse of $\alg A$ or $\alg C$ depending on the case) is a Taylor center unless $\alg C$ contains a nontrivial proper 2-absorbing subuniverse (in which case we are done). Then $\alg C$ has no nontrivial projective subuniverses by Proposition~\ref{prop:projective-implies-binabs}, and therefore it has a transitive term operation by Proposition~\ref{prop:transitive-or-projective}. Proposition~\ref{prop:trans-and-center-implies-absorption} finishes the proof by showing that $B$ is an absorbing subuniverse (of $\alg A$ or $\alg C$). 
\end{proof}

For future reference we also record the following variation.

\begin{corollary} \label{cor:ATvariation}
Suppose $R \leq_{sd} \alg A \times \alg C$ is left-central and $\alg C$ is  Taylor. Then $\alg C$ has a nontrivial projective and 2-absorbing subuniverse, or the left center of $R$ absorbs $\alg A$.
\end{corollary}

\begin{proof}
 Either $\alg C$ has a nontrivial  projective subuniverse (which is 2-absorbing by Proposition~\ref{prop:projective-implies-binabs}) or $\alg C$ has a transitive term operation by Proposition~\ref{prop:transitive-or-projective} and then the left center absorbs~$\alg A$ by Proposition~\ref{prop:trans-and-center-implies-absorption}.
\end{proof}

Now we improve absorption to 3-absorption. 
We further divide this task into two sub-steps. The first one isolates a property of centers that was, inspired by this paper, exploited in~\cite{zhuk21:strong-subalgs}.

\begin{lemma}\label{lem:alt_def_center}
    Let $B$ be a center of $\alg A$ and let $a\in A\setminus B$. 
    Then
    $(a,a)$ is {\em not} in the subuniverse of $\alg A^2$ generated by $(\{a\}\times B)\cup (B\times B) \cup (B \times \{a\})$.
  \end{lemma}

  \begin{proof}
    Let $R\sd\alg A\times\alg C$ be a witness of centrality. 
    Suppose, for a contradiction that $(a,a)$ is generated by a term operation $f^{\alg A}$, so
    \begin{equation*}
      f^{\alg A}(a,\dotsc,a,b_1,\dotsc,b_i) = a
      = f^{\alg A}(b'_1,\dotsc,b'_j,a,\dotsc,a)
    \end{equation*}
    for some $b_1,\dots,b_{i},b_{1}',\dots,b_{j}'\in B$ where $i+j$ is not less than the arity of $f$. 
    Therefore $f^{\alg C}(a+R, \dots, a+R, b_1+R, \dots, b_i+R) \subseteq a+R$ and, denoting $D = a+R$, we have $f^{\alg C}(D, \dots, D,C, \dots, C) \subseteq D$ (with $i$ occurrences of $C$ on the left).
    Similarly, we obtain $f^{\alg C}(C, \dots, C,D, \dots, D) \subseteq D$ with $j$ occurrences of $C$.
    It follows that the binary operation on $C$ obtained from $f^{\alg C}$ by identifying the first $j$ variables to $x$ and the rest to $y$ witnesses the nontrivial 2-absorption $D \abs_2 \alg C$, a contradiction with the definition of a center.
  \end{proof}

The second sub-step derives 3-absorption from absorption and a weakening of the property of centers isolated in the first sub-step. 

\begin{proposition}\label{prop:getting3abs}
    Suppose $B\abs \alg A$ and that for every $a\in A\setminus B$ we have $(a,a)\notin\Sg_{\alg A^2}(\{a\}\times B\cup B\times\{a\})$.
    Then $B$ $3$-absorbs $\alg A$.
\end{proposition}
\begin{proof}
Let $n$ be the minimal number such that 
$B \abs_{n+1} \alg A$ and assume, striving for a contradiction, 
that $n>2$. 
By Proposition~\ref{prop:abs_blockers} 
    there exist $\tuple a^1,\dotsc, \tuple a^n\in A^n$ such that 
    $a^i_j\in B$ for $i\neq j$ and  
    $\Sg_{\alg A^n}(\tuple a^1,\dotsc, \tuple a^n)\cap B^n=\emptyset$.
    Put 
    $R = \Sg_{\alg A^n}(\tuple a^1,\dotsc, \tuple a^n)$ and  
    assume that $R$ is an inclusion minimal relation 
    among all  choices of $\tuple a^1,\dotsc, \tuple a^n\in A^n$.
    
By $S$ we denote the binary relation defined by
 $$\Sg_{\alg A^{2}}\left((\{a^n_n\}\times B)\cup (B \times \{a^n_n\})\right).$$
Note that $a^n_n$ is not in $B$. We will now work towards showing that $(a^n_n,a^n_n)$ is in $S$, which will contradict the assumption on $B$.

We pp-define 
$R'\le \alg A^{2n-2}$ by the formula
\begin{align*}
R'(x_{1}, \dots,&x_{n-1},x_{1}', \dots,x_{n-1}') = \exists x_{n}, x_{n}' :\\
&R(x_{1}, \dots,x_{n})\wedge 
R(x_{1}', \dots,x_{n}')\wedge 
S(x_{n},x_{n}').
\end{align*}

For $i\in\{1,\dots,n\}$
by $\tuple c^{i}$ we denote $\tuple a^{i}$ with the last coordinate removed.
By the definition of $R'$ and $S$, we have 
$(\tuple c^{i},\tuple c^{n}),
(\tuple c^{n},\tuple c^{i})\in R'$ for every $i\in\{1,\dots,n-1\}$.
Moreover, these $2n-2$ tuples satisfy 
the condition in the second part of Proposition~\ref{prop:abs_blockers}. Therefore, if $R'\cap B^{2n-2}=\emptyset$, then $B \not\abs_{2n-2} \alg A$. But $B \abs_{n+1} \alg A$ and $2n-2\ge n+1$, a contradiction.

Since 
$R'\cap B^{2n-2} \neq \emptyset$, there exist
$\tuple d, \tuple d'\in R\cap (B^{n-1}\times A)$ such that
$(d_{n},d'_{n})\in S$.
Let $E$ be the projection of $R$ onto the last coordinate after fixing all the other coordinates to $B$, that is,  $E = \proj_{n}(R\cap (B^{n-1}\times A))$.
Since $R$ was chosen inclusion minimal, we get
$\proj_{n}(R) \subseteq \Sg_{\alg A}(B\cup\{e\})$ for every $e\in E$,
otherwise we could replace 
$\tuple a^{n}$ by a tuple $\tuple b\in R\cap (B^{n-1}\times \{e\})$ and get a  $B$-essential relation properly contained in $R$.

Let $E' = E + S$. 
Since $B\cup\{d'_{n}\}\subseteq E'$ (as $a^n_n, d_n \in E$, $(d_n,d'_n) \in S$, and $(a^n_n,b) \in S$ for all $b \in B$), 
we have $E' \supseteq \proj_n(R)$, in particular $a_{n}^{n}\in E'$. 
Therefore $(e,a_n^n) \in S$ for some $e \in E$. 
It follows that $E'' = a_n^n - S$ contains $B \cup \{e\}$ and we get $E'' \supseteq \proj_n(R)$, in particular $a_n^n \in E''$. 
So $(a^n_n,a^n_n) \in S$, which contradicts our assumption on $B$.
\end{proof}

The argument in Corollary~\ref{cor:AT} together with Lemma~\ref{lem:alt_def_center} and Proposition~\ref{prop:getting3abs} now proves 
Proposition~\ref{prop:center_absorbs} (which is  Corollary~7.10.2 in \cite{Zhuk20:dichotomy}) and also the improved absorption theorem.

\CenterTernaryAbsorbs*

\begin{proof} \hypertarget{ProofProp43}{}
 The algebra $\alg C$ has no nontrivial projective subuniverses by Proposition~\ref{prop:projective-implies-binabs}, and therefore it has a transitive term operation by Proposition~\ref{prop:transitive-or-projective}. 
 By
 Proposition~\ref{prop:trans-and-center-implies-absorption}, $B$ is an absorbing subuniverse of $\alg A$.  
 Lemma~\ref{lem:alt_def_center} and Proposition~\ref{prop:getting3abs} now imply $B \abs_3 \alg A$.
\end{proof}

\NewAT*

\begin{proof} \hypertarget{ProofCor45}{}
  Proposition~\ref{prop:linked-implies-central} implies that there is a left central subuniverse of $\alg A \times \alg C$ or $\alg C^2$. In both cases, Proposition~\ref{prop:center_absorbs} finishes the proof. 
   \end{proof}

Now we move on to the improvements of the theory. The first one is that centers of Taylor algebras absorb. Note the difference to Proposition~\ref{prop:center_absorbs} which concerns Taylor centers in arbitrary algebras.

\CenterTernaryAbsorbsSecond*

\begin{proof} \hypertarget{ProofProp44}{}
  Let $R\leq_{sd} \alg A\times\alg C$ be the witnessing relation from the definition of center. 
  We define a directed graph on $A$, by putting
  $c\rightarrow d$ if there exists a cyclic term $t$ and a choice of elements $b_2,\dotsc,b_n \in B$ such that $t^{\alg A}(c,b_2,\dotsc,b_n)= d$.
  
  Obviously the graph has no sinks, 
  but it is also easy to see that it is transitively closed.
  Indeed, let $c\rightarrow c'$ be witnessed by $t$ and $b_2,\dotsc,b_n$, and $c'\rightarrow c''$ be witnessed by $t'$
  and $b_2',\dotsc,b_{n'}'$.
  Then the star composition of $t'$ and $t$,
  together with the tuple
  \begin{equation*}
      b_2,\dotsc,b_n,\underbrace{b_2',\dotsc,b_2'}_n,
      \dotsc,\underbrace{b'_{n'},\dotsc,b'_{n'}}_n,
  \end{equation*}
  reordered as in the definition of star composition,
  witnesses the edge $c\rightarrow c''$.
 
  Next we claim that there is no $a\in A\setminus B$ with a self-loop.
  For a contradiction, say $a$ is such and
  let $t^{\alg A}(a,b_2,\dotsc,b_n)=a$ with $a\in A\setminus B$ and $b_2,\dotsc,b_n\in B$.
  The set $D = a+R$
  is a subuniverse of $\alg C$, and moreover 
  $$
  t^{\alg C}(a+R,b_2+R,\dotsc,b_n+R)=
  t^{\alg C}(D,C,\dotsc,C)\subseteq D
  $$
  and the same holds for all the cyclic shifts of $t$.
  Therefore, $D$ is a $2$-absorbing subuniverse of $\alg C$ (by the operation $t^{\alg C}(x,y, \dots, y)$), a contradiction.
  
  We conclude that all the loops and directed cycles in this graph must be  inside $B$.
  Therefore there is a number $m$ so that starting with any $a\in A$, and following any directed walk, after $m$ steps we necessarily arrive in $B$.
  It suffices to take any cyclic term of $\alg A$ and star-compose it with itself $m$ times to obtain an operation which witnesses $\alg B\abs\alg A$. 
  Now it is enough to apply Lemma~\ref{lem:alt_def_center} and Proposition~\ref{prop:getting3abs}
\end{proof}

The second improvement, Proposition~\ref{prop:center}, 
shows that every linked relation $R \subseteq A \times B$ with $A=B$ pp-defines a central relation which is additionally symmetric or transitive. 
We need an auxiliary, folklore lemma about iterated composition of a binary relation with itself. We use the following notation for a positive integer $n$.
$$
nR  = \underbrace{R + R + \cdots + R}_{n \times}
$$

\begin{lemma} \label{lem:transitive}
    For any proper $R \subseteq A^2$ there exists $n \in \mathbb{N}$ such that $nR$ is proper and $2nR$ is either $A^2$ or $nR$.  
\end{lemma}
\begin{proof}
    First consider the case that $kR = A^2$ for some $k$. Then $A^2 = (k-1)R + R$ implies that the projection of $R$ to the second coordinate is full and therefore $(k+1)R = A^2+R$ is full. By induction, $k'R$ is full for each $k' \geq k$ so we can define $n$ as the largest number such that $nR$ is proper and get $2nR = A^2$.

    Assume now that $kR$ is proper for all $k$. Denote $l=|A|!$ and observe that for any $k \geq |A|$ we have $kR \subseteq (k+l)R$. Indeed, if $(a,b) 
    \in kR$, then there exists a directed walk from $a$ to $b$ of length $k$ in the digraph with edge set $R$. Since $k \geq |A|$, some segment of this walk is a directed cycle. Its length is divisible by $l$, so we can make the directed walk longer by $l$ going along the cycle multiple times, implying $(a,b) \in (k+l)R$.

    It follows that  $lR \subseteq 2lR \subseteq 3lR \subseteq  \ldots$, therefore this sequence of proper relations eventually stabilizes. We define $n = kl$ for a sufficiently large $k$ and obtain $2nR = nR$.
\end{proof}

\LinkedImpliesCenter*

\begin{proof} \hypertarget{ProofProp42}{}
  If the left or right center of $R$ is empty, we apply Proposition~\ref{prop:linked-implies-central} 
  to $R$ itself, or to $-R$,  
  and the result follows. So, let
  $R$ be central.
  We also assume, without loss of generality, that the left center of $R$ contains the maximal number of elements among central, proper and subdirect relations pp-definable from $R$.
  
  As is easily seen by induction, every $nR$ is subdirect and central; moreover, the left center of $nR$ contains the left center of $R$. We give an argument only for the last claim. If $a$ is in the left center of $nR$, then for every $b \in A$ we have $(c,b) \in R$ for some $c$ (as $R$ is subdirect) and $(a,c) \in nR$ (as $a$ is in the left center of $nR$), therefore $(a,b) \in nR+R = (n+1)R$, i.e. $a$ is in the left center of $(n+1)R$.

  By Lemma~\ref{lem:transitive}, there exists $n$ such that $nR$ is proper and $2nR$ is either $A^2$ or $nR$. In the latter case, $nR$ is the desired subdirect, proper, central, and transitive relation (as $nR + nR = nR$). 
  Assume now that $2nR = A^2$ and define $S = nR$. The relation $S$ is subdirect, proper, and central, it satisfies $S+S=A^2$, and its left center is the same as the left center of $R$ (since $R$ was chosen to have the largest possible left center).   
  
  Next, let $B$ be the right center of $S$, we consider two cases: either $B+S=A$ or $B+S\neq A$. (Note that $B$ is the \emph{right} center, which implies $B-S=A$, but not necessarily $B+S=A$. So the latter case cannot be excluded immediately.)

\medskip
\noindent
{\sc Case 1.}
$B+S=A$.

\smallskip 

Consider $S' = S\cap -S$. This relation is proper (because $S$ is proper) and is symmetric by construction. It is also subdirect, as $S+S=A^2$ implies that for every $a$ there is $b$ such that $S(a,b)$ and $S(b,a)$. Finally, $S'$ is also linked. Indeed, note that, since $B$ is the right center, $B^2\subseteq S$, and so $B^2\subseteq S'$. Also, the assumption $B+S=A$ implies that for any $a\in A$ there is $b\in B$ such that $(b,a)\in S$. On the other hand, $(a,b)\in S$, because $b$ belongs to the right center. Therefore $(a,b)\in S'$, implying together with $B^2\subseteq S'$ that $S'$ is linked. 
 
 If $S'$ is central, then we are done. Otherwise, since $S'$ is symmetric, its right center is empty, and we use Proposition~\ref{prop:linked-implies-central} to obtain a symmetric central relation.
 
 \medskip
 \noindent
{\sc Case 2.}
$B+S\neq A$.

\smallskip

We will derive a contradiction in order to show that this case is impossible.
Let $A=\{a_1,\dotsc,a_n\}$, and for $j\ge0$, let the relation $T_j$ be given by 
\[
T_j(x,y)=(\exists z) S(x,z)\wedge S(z,y)\wedge\bigwedge_{i=1}^jS(a_i,z).
\]
  Clearly, $T_0 = A^2$ and $T_n$ is not even subdirect. Indeed, for the latter claim if $(a,b)\in T_n$, then the value of $z$ in the pp-definition above belongs to $B$, the right center. As $B+S\neq A$, there is $c\in A$ such that $(z,c)\notin S$ for any feasible choice of $z$, witnessing that $c\notin\proj_2(T_n)$.
  Therefore there is $j$ such that $T_{j-1} = A^2$, and $T_j \neq A^2$.
  We will show that $T_j$ is central and has strictly larger left center than $S$, which contradicts the choice of $R$. 
  
  By the definition of $T_j$, for every $b \in A$, we have $(a_j,b)\in T_{j-1}$ if and only if $T_j(a_j,b)$, therefore $a_j+T_j = A$. 
  This implies that $\proj_2(T_j)=A$ and that $a_j$ is in the left center of $T_j$.
  Note that every element in the left center of $S$ is in the left center of $T_j$, since for any $a$ in the left center of $S$ we have $a + T_j \supseteq a_j + T_j = A$.
  Note also that $a_j$ does not belong to the left  center of $S$, because this would imply that $T_{j-1}=T_j$.
  Finally, we claim that $\proj_1(T_j)=A$, so $T_j$ is subdirect.
  Indeed, since $S$ has nonempty right center, for any $a\in A$, we can choose the value of $z$ in the definition of $T_j$ to be from the right center. Then $(a,z),(a_1,z),\dotsc,(a_j,z)\in S$, and a value for $y$ can be chosen with $(z,y)\in S$.

  Thus $T_j$ is proper, subdirect, central and pp-definable from $R$. However, its left center is a proper superset of the left center of $R$ since it contains the element $a_j$ which is not in the left center of $R$. This is a contradiction with the choice of $R$. 
\end{proof}

\subsection{Subdirect irredundant subpowers}

In this subsection we prove our main pp-definability result, Theorem~\ref{ff:subdirect}, and we also prove Theorems~\ref{thm:maj} and \ref{thm:semilat} which will give us enough edges for the connectivity theorem. We start with a technical lemma. 

\begin{lemma}\label{fl:long}
    Let $R\subseteq_{sd} A^n$ be a relation and let $I\subseteq [n]$ be an inclusion maximal set of 
    coordinates such that $\proj_I (R)$ is the full product $A^{I}$. 
    Suppose that $R$ does not pp-define $R'\subseteq_{sd} A^2$ which is irredundant and proper. Then every tuple in $R$ is determined by its projection to $I$. 
\end{lemma}
\begin{proof}
    Let $R \subseteq_{sd} A^n$ and $I \subseteq [n]$ form a counterexample minimal with respect to $n$. 
    In particular, $n \geq 2$, $R$ is not the full relation, and no relation pp-definable from $R$, such as the projection to a subset of variables, pp-defines a
    subdirect, irredundant, and proper binary relation.

    First, $|I|$ has to be $n-1$. Indeed, otherwise $|I|<n-1$ and, for any $j \not\in I$, the projection $S$ of $R$ onto $I \cup \{j\}$ does not pp-define a subdirect, irredundant, and proper binary relation. Also, the set of coordinates $I$ is maximal such that $\proj_{I} (S)$ is full. So, by the minimality of $n$, for every tuple $\tuple b \in S$ the value $b_{j}$  is determined by the remaining coordinates. It follows that every tuple $\tuple a \in R$ is determined by the values $a_i$, $i \in I$, which we know is not the case since $R$ and $I$ form a counterexample.

    Without loss of generality assume $I=\{1,\dotsc,n-1\}$.
    Next, we claim that $\proj_{1,\dotsc,n-2,n} (R)$ is full. Indeed, otherwise this projection together with $\{1,2, \dots, n-2\}$ is also a counterexample, contradicting the minimality of $n$. 
    
    Since $R$ is a counterexample, there are elements $a\neq a'$ and
    tuples $(a_1,\dotsc,a_{n-1},a), (a_1,\dotsc$, $a_{n-1},a')\in R$. Also, as $R$ is not full, $(c_1,\dotsc,c_n)\notin R$ for some tuple in $A^n$. Set $T = \{(b_1,\dotsc,b_{n-1})\st (b_1,\dotsc,b_{n-1},c_n)\in R\}$. We show that $T$ together with $I' = \{1,2, \dots, n-2\}$ is a smaller counterexample, thus obtaining a contradiction. 
    Indeed, the relation $T$ is proper, as it does not contain $(c_1,\dotsc,c_{n-1})$. Also, $\proj_{I'}(T)$ is the full relation since $\proj_{1,\dotsc,n-2,n}(R)$ is. It remains to show that tuples in $T$ are not determined by their projection to $I'$.  To this end we consider an auxiliary binary relation $S$ given by $S = \{(a,b)\st (a_1,\dotsc,a_{n-2},a,b)\in R\}$. This relation is subdirect, as both $\proj_I(R)$ and $\proj_{1,\dotsc,n-2,n}(R)$ are full relations. The relation $S$ is also irredundant, because $(a_{n-1},a),(a_{n-1},a')\in S$. By the assumptions about $R$, the relation $S$ cannot be proper. Therefore, the tuples $(a_1,\dotsc,a_{n-2},a,c_n),(a_1,\dotsc,a_{n-2},a',c_n)$ are in $R$, implying that  $(a_1,\dotsc,a_{n-2},a),(a_1,\dotsc,a_{n-2},a')$ are in $T$. 
\end{proof}

We are now in a position to prove Theorem~\ref{ff:subdirect}.

\DefinabilityTheorem*

\begin{proof} \hypertarget{ProofThm47}{}
    First, we argue that $R$ pp-defines some binary or ternary subdirect, proper, and irredundant relation. Let $R'$ be a subdirect proper irredundant relation of minimal arity pp-definable from $R$. Clearly $R'$ cannot be unary and if $R'$ is binary or ternary, we are done. Otherwise observe that the projection of $R'$ onto any proper set of coordinates is the full relation. Let $(a_1,\dotsc,a_k)\notin R'$. Consider the relation $S$ given by 
    \[
    S=\{(x_1,\dotsc,x_{k-1})\st (x_1,\dotsc,x_{k-1},a_k)\in R'\}.
    \] 
    This relation is proper, as $(a_1,\dotsc,a_k)\notin R'$. It is also subdirect, because every binary projection of $R'$ is the full relation. If $S$ is redundant, say, for $i,j\in\{1,\dotsc,k-1\}$ it holds that $\proj_{ij}(S)$ is the graph of bijection, then $\proj_{ijk}(R')$ is a proper relation, a contradiction with the choice of $R'$.
    
    If $R$ pp-defines a binary, subdirect, proper, and irredundant relation, the first item from the conclusion of the theorem holds. So, suppose that such a binary relation cannot be defined. 
    Let $R_1,\dotsc,R_m\subseteq_{sd} A^3$ be all the proper, ternary, irredundant and subdirect relations pp-definable from $R$.
    Any binary projection of each $R_i$ is the full relation, and
    by Lemma~\ref{fl:long}, any tuple from~$R_i$ is determined by any pair of its entries, therefore each $R_i$ is strongly functional.
    It remains to prove that the set $\{R_1,\dotsc,R_m\}$ pp-defines $R$.
    
    We show, by induction on the arity $l$, that any irredundant subdirect relation $S \subseteq_{sd} A^l$ pp-definable from $R_1,\dotsc,R_m,R$ is also pp-definable from $R_1,\dotsc, R_m$. For $l < 3$ the claim follows trivially, so assume $l \geq 3$.
    By Lemma~\ref{fl:long} there is some $I\subseteq\{1,\dotsc,l\}$ such that any tuple $(a_1,\dotsc,a_l)\in S$ is determined by its projection on $I$ and $\proj_I(S)$ is the full relation. Assume $I=\{1,\dotsc,k\}$. If $k+1\neq l$, then 
    \[
    S(x_1,\dotsc,x_k)=\bigwedge_{j=k+1}^l \proj_{I\cup\{j\}}(S(x_1,\dotsc,x_k,x_j)).
    \]
    By the induction hypothesis every $\proj_{I\cup\{j\}}(S)$ is pp-definable from $R_1,\dotsc,R_m$, hence, so is $S$ and we are done. Thus we may assume that $k+1=l$.
    
    \medskip
    \noindent
    {\sc Claim.}
    Each projection $S'=\proj_J(S)$ on $J\subseteq\{1,\dotsc,l\}$ with $|J|<l$ is full or $S$ is pp-definable from $R_1, \dotsc,R_m$. 
    
    \smallskip
    
    Indeed, it is the case if $J\subseteq I$ by the choice of $I$. Suppose for a contradiction that $\proj_J(S')$ is not full. Then we must have $l\in J$, and, as $\proj_{J\setminus\{l\}}(S')$ is the full relation, by Lemma~\ref{fl:long} every tuple from $S'$ is determined by its projection to $J\setminus\{l\}$. This means that $S$ can be obtained from $S'$ by extending the tuples from $S'$ in an arbitrary way, and this is a pp-definition of $S$ from $S'$. By the induction hypothesis $S'$ is pp-definable from $R_1,\dotsc,R_m$, therefore so is $S$.

    \smallskip

    We can therefore assume that the first case in the claim takes place.
    Since $S$ is proper, there is some $(c_1,\dotsc,c_l)\notin S$.
    Set $T=\{(a,b,c)\st (c_1,\dotsc,c_{l-3},a,b,c)\in S\}$. This relation is proper, since $(c_{l-2},c_{l-1},c_l)\notin T$. By the claim we get $\proj_{1,2}(T) = \proj_{1,3} (T) = \proj_{2,3}(T) = A^2$. Thus, $T$ is one of the $R_i$. Consider the relations $U$ and $U'$ given by
    \begin{equation*}
      U(x_1,\dots,x_l,y) = S(x_1,\dotsc, x_l) \wedge T(y,x_{l-1},x_l).
    \end{equation*}
    and 
    \begin{eqnarray*}
      U'(x_1,\dotsc,x_l,y) &=& \proj_{1,\dotsc,l-2, l+1} (U(x_1,\dotsc,x_{l-2},y))\\ 
      && \wedge T(y,x_{l-1},x_l).
    \end{eqnarray*}
    We will show that they are identical. This will imply the result, because, as is easily seen, $S=\proj_{1,\dotsc,l}(U)$, and $U'$ is pp-definable from $R_1,\dotsc,R_m$, as $\proj_{1,\dotsc,l-2, l+1}(U)$ is by the induction hypothesis.
    
    It is not hard to see that $U\subseteq U'$. Next, note that $U''=\proj_{1,\dotsc, l-2,l+1}(U)$ is not full, since $U(c_1,\dotsc,c_{l-3},a,b,c,d)$ imply that $d=a$. On the other hand, $\proj_{1,\dotsc,l-2}(U)$ is full. Therefore, by Lemma~\ref{fl:long}   for any $(a_1,\dotsc,a_l,a)\in U$ the value $a$ is determined by $a_1,\dotsc,a_{l-2}$. 
    
    Take a tuple $(a_1,\dotsc,a_l,a)\in U'$; since $\proj_{1,\dotsc,l-1}(S)$ is the full relation, we have that $(a_1, \dotsc, a_{l-1}, c$, $ d)\in U$ for some $c,d\in A$. Since $(a_1,\dotsc,a_{l-2},a)\in\proj_{1,\dotsc,l-2, l+1}(U)$, and in this relation the last value is determined by the first $l-2$ ones, we have $d=a$. Again by Lemma~\ref{fl:long} the third coordinate of the relation $T$ is determined by the first two coordinates. Therefore, as we have $T(d,a_{l-1},c)$ from the definition of $U$, $T(a,a_{l-1},a_l)$ from the definition of $U'$, and $d=a$, we also obtain $c=a_l$. Thus $(a_1,\dotsc,a_l,a)\in U$, completing the proof.
\end{proof}

As already discussed in Section~\ref{sec:taylor},
Theorem~\ref{ff:subdirect} implies that every algebra $\alg A$ has at least one of the following properties of its invariant relations.
\begin{enumerate}
    \item [(1)] $\alg{A}$ has no proper irredundant subdirect subpowers.
    \item[(2)] $\alg{A}$ has a proper irredundant binary subdirect subpower.
    \item[(3)] $\alg{A}$ has a ternary strongly functional subpower.
\end{enumerate}
Proposition~\ref{prop:ImpliesAbelian} shows that $\alg A$ is abelian in case (3). 

\ImpliesAbelian*

\begin{proof} \hypertarget{ProofProp48}{}
  We define a binary relation $S$ on $A^2$ by
  \begin{align*}
  S((&x_1,x_2),(y_1,y_2))= \exists u,v,u',v': \\
  &R(u, v, x_1) \wedge 
  R(u, v', x_2) \wedge 
  R(u',v, y_1) \wedge
  R(u',v', y_2).
  \end{align*}
  Since a tuple of $R$ is determined by any two coordinates, 
  $x_{1}=x_2$ implies $v=v'$, and this implies $y_1=y_2$, and 
  vice versa.
  Since any binary projection of $R$ is full, 
  for any $(x_{1},x_{2})$ we can choose 
  $u$ and then $v$, $v'$ such that 
  $(u,v,x_1),(u,v',x_2)\in R$. 
  Setting $u'=u$, $y_1=x_1$, and $y_2=x_2$, 
  we see that 
  $S$ is a reflexive relation on $A^{2}$, in particular, the projection of $S$ to any of the two coordinates is $A^2$. Finally, any pair of the form $((x,x),(y,y))$ is in $S$ as witnessed by picking $u$ arbitrarily and then choosing $v=v'$ and $u'$ appropriately.
  It follows that the ``linkedness'' congruence $(S-S)+ \dots + (S-S)$ is a congruence on $\alg A^2$ such that one of its classes is the diagonal $\Delta_A$, so $\alg A$ is abelian.
\end{proof}

The next theorem produces majority edges in case (1) when $\mu_{\alg A}$ is not full.

\GettingMajorityEdges*

\begin{proof} \hypertarget{ProofThm410}{}
We consider the subalgebra $\alg F$ of $\alg A^{A^3}$ formed by all the ternary term operations of~$\alg A$, i.e., $F = \Clo_3(\alg A)$. 
We set $P = \{(a,b) : \Sg_{\alg A}(a,b)=\alg A\}$ and define an equivalence relation $\sim$ on $P$ by $(a,b) \sim (a',b')$ if $t(a,a,b)$ determines $t(a',a',b')$ and vice versa; in other words, if
$\{(t(a,a,b),t(a',a',b')): t \in F\}$ is the graph of a bijection (from $A$ to $A$).
Note that $\sim$ is indeed an equivalence relation on $P$ and that $(a,b) \sim (a',b')$ if, and only if, $t(a,b,a)$ and $t(a',b',a')$ determine each other  (since coordinates of operations in $F$ can be permuted) and this happens if, and only if, $t(b,a,a)$ and $t(b',a',a')$ determine each other. 
Let $P'$ be a subset of $P$ that contains exactly one representative from each $\sim$-class, let
$I = \bigcup_{(a,b) \in P'} \{(a,a,b),(a,b,a),(b,a,a)\}$, and finally let $\alg R \leq \alg A^I$ be the projection of $\alg F$ onto the set of coordinates $I$.

By the definition of $P$, the subpower $R$ is subdirect. We claim that $R$ is irredundant, i.e., that $S=\proj_{(a_1,a_2,a_3),(a'_1,a'_2,a'_3)}(R) = \{(t(a_1,a_2,a_3),t(a'_1,a'_2,a'_3)): t \in F\}$ is not a graph of a bijection for any distinct triples $(a_1,a_2,a_3), (a'_1,a'_2,a'_3)$ in $I$. Indeed, if the position of the nonrepeating element in these two triples is different, then already the pairs $(a_1,a'_1),(a_2,a'_2),(a_3,a'_3)$ (which are in $S$ since $F$ contains the three ternary projections) witness this. On the other hand, if the position of the nonrepeating element is the same, then $S$ is not a graph of a bijection by the definition of $\sim$ and $P'$, and the mentioned equivalent definitions. 

Since $\alg A$ has no subdirect proper irredundant subpowers, we get $R = A^I$. 
In particular, there exists $t \in F$ such that $t(a',a',b')=t(a',b',a')=t(b',a',a') = a'$ for every $(a',b') \in P'$.
It remains to observe that these equalities hold for every $(a,b) \in P$. Indeed, by the definition of $P'$ there exists $(a',b') \in P'$ such that $(a,b) \sim (a',b')$. We have $\proj^3_1(a',a',b')=a'$ and $\proj^3_1(a,a,b)=a$, therefore $t(a',a',b')=a'$ implies $t(a,a,b)=a$. By a similar argument we get $t(a,b,a)=a=t(b,a,a)$ and the proof is concluded. 
 \end{proof}

\noindent 
In case (2) and when $\alg A$ is simple and $\mu_{\alg A}$ is not full we get semilattice edges by the following theorem.

\GettingSemilatticeEdges*

\begin{proof} \hypertarget{ProofThm411}{}
  We assume that $\mu_{\alg A}$ is not full, otherwise the theorem is vacuously true.
  We prove the conclusion via a sequence of claims. 
  We start with a basic observation: if $R\sd \alg A^2$ and 
  $(a,b),(a',b) \in R$ for some $(a,a')\notin\mu_{\alg A}$ then $b$ is in the right center of $R$ (since $\Sg_{\alg A}(a,a')=A$).
  \begin{claim}\label{claim:one}
    Every irredundant, proper $R\sd\alg A^2$ is central.
  \end{claim}
  \begin{proof}
    Since $\alg A$ is simple, $R$ needs to be linked.~
    Therefore, since $\mu_{\alg A}$ is not the full relation, we have $(a,b),(a',b)\in R$
    for some $(a,a')\notin\mu_{\alg A}$ which implies that $b$ is in the right center of $R$.
    The proof for left center is symmetric.
  \end{proof}
  \begin{claim}
    Let $R\sd\alg A^2$ be irredundant and proper. 
    If $a$ and $a'$ are in the left~(right) center of $R$ then $(a,a')\in\mu_{\alg A}$.
  \end{claim}
  \begin{proof}
    Suppose not.
    Then for every $b\in A$ we have $(a,b),(a',b)\in R$ which implies that $b$ is in the right center.
    This cannot happen in a proper $R$.
  \end{proof}
  \begin{claim}\label{claim:samecenter}
    Let $R\sd\alg A^2$ be irredundant and proper. 
    If $a$ is in the left center of $R$ and $a'$ is in the right center then $(a,a')\in\mu_{\alg A}$.
  \end{claim}
  \begin{proof}
    Suppose, for a contradiction, that $a$ is in the left center of $R$, $a'$ is in the right center, and $(a,a')\notin\mu_{\alg A}$.
    Since $(a,a),(a',a')\in R$ then, for every $b\in A$, we have $(b,b)\in R$.
    
    Consider any $a''$ such that $(a'',a')\notin\mu_{\alg A}$. Then
    $a''$ is in the left center of $R$, as both $(a'',a'')$ and $(a'',a')$ are in $R$.
    This implies, by the previous claim, that $(a'',a)\in\mu_{\alg A}$. 
    
    In particular we conclude that $\mu_{\alg A}$ has two equivalence classes,
    and that the left center is equal to $a/\mu_{\alg A}$. By a symmetric argument, the right center is equal to  $a'/\mu_{\alg A}$. Altogether, $\mu_{\alg A}$ is the equivalence with classes $a/\mu_{\alg A}$ and $a'/\mu_{\alg A}$ which are the left and right centers, respectively.

    Also note that no pair $(b,b')$ with $b \not\in a/\mu_{\alg A}$ and $b' \not\in a'/\mu_{\alg A}$ is in $R$: otherwise $(b,b), (b,b') \in R$ implies that $b$ is in the left center, which is not the case. 
    Therefore  $R = (a/\mu_{\alg A}\times A) \cup (A\times a'/\mu_{\alg A})$,
    but then $R\cap -R = \mu_{\alg A}$ is a proper congruence (recall $|A|>2$)  on a simple algebra $\alg A$,  a contradiction.
  \end{proof}

  So far we have shown that each subdirect, irredundant, and proper $R$ is central and both centers are contained in a single $\mu_{\alg A}$-class $B$. We refer to $B$ as the \emph{central} $\mu_{\alg A}$-class of $R$.

  \begin{claim}\label{claim:moving}
    Let $R\sd\alg A^2$ be irredundant and proper, $B$ its central $\mu_{\alg A}$-class, and $S\sd\alg A^2$ redundant. Then $B+S = B$.
  \end{claim}
  \begin{proof}
     The relation $R+S$ is clearly irredundant and proper, its left center is equal to the left center of $R$, and its right center is $C+S$, where $C$ is the right center of $R$. Therefore the central $\mu_{\alg A}$-class of $R+S$ is $B$ and $C+S \subseteq B$.
     
     The relation $S$ is the graph of an isomorphism from $\alg A$ to $\alg A$. Since isomorphism map subuniverses to subuniverses, it maps each $\mu_{\alg A}$-class to a $\mu_{\alg A}$-class. In particular $B+S$ is a $\mu_{\alg A}$-class and since it contains $C+S \subseteq B$, we get $B+S=B$.
  \end{proof}
  
  \begin{claim}\label{claim:last}
    Let $R, S\sd\alg A^2$ be irredundant and proper and $B,C$ their central $\mu_{\alg A}$-classes, respectively. Then $B=C$. 
  \end{claim}
  \begin{proof}
    Suppose, for a contradiction, that $B\neq C$.
    Let $b,b'$ be some elements of the left and right centers of $R$, respectively,
    and similarly $c,c'$ for $S$.
    Let $T = R \cap S$,
    and note that $(b,c'),(c,b')\in T$.
    As $(b,c)\notin\mu_{\alg A}$ and $(c',b')\notin\mu_{\alg A}$, the relation $T$ is subdirect in $A^2$.

    Since $B\cap C = \emptyset$ the relation $T$ has no center, and thus by Claim~\ref{claim:one} needs to be redundant.
    The previous claim  implies  $c' \in (B+T)\cap C = B \cap C = \emptyset$, a contradiction. 
  \end{proof}
    
  To finish the proof we take a proper irredundant subdirect binary relation $R$ provided by the assumption
  and set $B$ to be its central $\mu_{\alg A}$-class.
  Take any $b\in B,a\notin B$ and let $S= \Sg_{\alg A^2}((a,b),(b,a))$.
  The relation $S$ is subdirect (as $(a,b), (b,a) \in S$) and cannot be redundant as $a\in B+S$ would contradict Claim~\ref{claim:moving}.
  Thus, by Claim~\ref{claim:one} and Claim~\ref{claim:last}, there is $b'\in B$ in the left center of $S$.
  Since $(b',b),(a,b)\in S$ we conclude that $b$ is in the right center of $S$, in particular $(b,b)\in S$ and we are done as witnessed by the operation generating $(b,b)$ from the generators $(a,b)$, $(b,a)$.
\end{proof}

We are ready to derive the promised refinements of the fundamental theorems for Zhuk's and Bulatov's approach as corollaries.

\ZhukFundamentalRefined*

\begin{proof} \hypertarget{ProofCor412}{}\hypertarget{ProofThm37}{}
  Let us first assume that $\alg A$ is simple. 
  If case (d) does not apply, then there exists a proper irredundant subdirect subpower of $\alg A$. By Theorem~\ref{ff:subdirect}, such a relation pp-defines a strongly functional ternary relation or a proper irredundant subdirect binary relation $R$. The former situation leads to case (c) via Proposition~\ref{prop:ImpliesAbelian}. In the latter case, the binary relation $R \leq_{sd} \alg{A}^2$ is linked (as $\alg A$ is simple) and thus it defines a proper left central relation $S \leq_{sd} \alg A^2$ by Proposition~\ref{prop:linked-implies-central}. 
  Now either $\alg A$ has a  nontrivial 2-absorbing subuniverse and we are in case (a) or $S$ witnesses that we are in case (b).
  
  In the general case, take a maximal congruence $\alpha$ of $\alg A$, apply what we have already proved to the simple algebra $\alg A/\alpha$, and
  observe that 2-absorbing subuniverse and centers can be lifted: if $R$ is a 2-absorbing (central) subuniverse of $\alg A/\alpha$, then the preimage of $R$ under the quotient map $A \to A/\alpha$ is a 2-absorbing (central) subuniverse of $\alg A$.
  \end{proof}

\BulatovFundamentalThm*

\begin{proof} \hypertarget{ProofThm32}{}
  We prove the claim by induction on the size of $\alg A$.   If $\alg A$ has two elements the result follows from the classification of Boolean clones by Post~\cite{post41:clones}, so we assume $|A| \geq 3$. 
  
  Suppose that $\mu_{\alg A}$ is full, that is, each pair of elements in $A$ is connected by proper subuniverses. Since all the proper subalgebras of $\alg A$ have connected directed graphs of edges by the induction hypothesis, and these edges are trivially also edges in $\alg A$, it follows that the digraph of edges of $\alg A$ is connected. 
  Suppose further that $\mu_{\alg A}$ is not full.
  
  If $\alg A$ is simple, then 
   each equivalence class of $\mu_{\alg A}$ is connected by edges by the induction hypothesis. We apply Theorem~\ref{ff:subdirect} together with Proposition~\ref{prop:ImpliesAbelian}, Theorem~\ref{thm:maj}, or Theorem~\ref{thm:semilat} to conclude
  that either $\alg A$ is abelian and every pair is an abelian edge,
  or every pair $(a,a')\notin\mu_{\alg A}$ is a  majority edge, 
  or that there exists a $\mu_{\alg A}$-class $B$ such that every pair $a\notin B, b \in B$
  forms a semilattice edge $(a,b)$~%
  (and in every case the witnessing congruence is the identity congruence). 
  
  If $\alg A$ is not simple we consider any
  maximal congruence $\alpha$ on $\alg A$. 
    By the induction hypothesis  the congruence classes of $\alpha$ 
  as well as $\alg A/\alpha$ have connected directed graphs of edges.
  But since any edge $(a/\alpha,b/\alpha)$ in $\alg A/\alpha$ witnessed by
  $\theta$ on $\Sg_{\alg A/\alpha} (a/\alpha,b/\alpha)$ gives rise to edge $(a,b)$ in $\alg A$ witnessed by the lifted congruence $\theta'$ on $\Sg_{\alg A}(a,b)$, the proof is concluded in the nonsimple case as well.
\end{proof}

\BulatovFundamentalMinimal*

\begin{proof} \hypertarget{ProofCor413}{}
  Let $\alg A$ be a minimal counterexample to the theorem.
  By Theorem~\ref{thm:connected}, $A$ is connected by edges
  and it suffices to show that each pair of elements connected by an edge is connected by minimal edges.
  Let $(a,b)$ be an edge in $\alg A$ and let $\theta$ be a maximal
  congruence on $\Sg_{\alg A}(a,b)$
  witnessing the edge.
  The equivalence classes of $\theta$ are connected by minimal edges by the minimality of $\alg A$.
  Choose $a',b'$ such that $(a,a'),(b,b')\in\theta$ and $\alg B =\Sg_{\alg A}(a',b')$ is inclusion minimal.
  If $B=A$, then $(a,b)$ is a minimal edge and we are done.
  Otherwise $a'$ and $b'$ are connected by minimal edges by the minimality of $\alg A$, and then so are $a$ and $b$.
\end{proof}

\section{Proofs for Section~\ref{sec:mtalgebras}: Minimal Taylor algebras}
\label{app:mtalgebras}
This section  contains proofs of the theorems from Section~\ref{sec:mtalgebras}
together with a number of lemmas  mentioned there, or needed for the proofs. The final subsection gives proofs of the claims in examples in Section~\ref{sec:mtalgebras}.

\subsection{General facts}

\XXXpropTMexist*

\begin{proof} \hypertarget{ProofProp52}{}
    Let $\alg A$ be a Taylor algebra and $p>|A|$ be a prime number.
    By Theorem~\ref{thm:cyclic} the algebra~$\alg A$ has a cyclic term operation of a prime arity $p$.
    Consider a family of clones generated by such cyclic operations~%
    (the family is finite),
    and choose a term $t$ which defines a minimal (under inclusion)
    clone in this family.
    The algebra $(A;t)$ is clearly a Taylor reduct of $\alg A$.
    If $(A;t)$ had a proper Taylor reduct then, by Theorem~\ref{thm:cyclic},
    it would have a $p$-ary cyclic term operation. 
    This contradicts the choice of $t$ and shows that $(A;t)$ is a minimal Taylor algebra.
\end{proof}

\XXXtmsubs*

\begin{proof} \hypertarget{ProofProp53}{}
    Choose a prime number $p>|A|$.
    The set $B$ together with the restriction of $f$ to $B$, call it $f'$, forms a Taylor algebra.
    By Theorem~\ref{thm:cyclic} 
    $(B;f')$ has a cyclic operation of arity $p$.
    This operation is defined by a term in $f'$, and we let $h$ to be the operation of $\alg A$ defined by the same term after replacing $f'$ by $f$.
    The set $B$ is closed under $h$ and $h$ is a cyclic operation on $B$.
    Let~$s$ be a cyclic composition of $h$ and a cyclic operation of $\alg A$ of arity $p$ denoted by $t$~%
    ($t$ exists by Theorem~\ref{thm:cyclic}).
    The operation $s$ is cyclic and, more importantly,
    preserves $B$ as $t$ is idempotent and $s$ is cyclic on~$B$. 
    Since $\alg A$ is minimal Taylor, 
    $\Clo(\alg A)= \Clo(A;s)$ and $B$ is a subuniverse of~$\alg A$.
\end{proof}

\XXXhsp*

\begin{proof} \hypertarget{ProofProp54}{}
    For finite powers the claim follows from the definition of a power.
    Let $\alg A$ be a minimal Taylor algebra
    and let $\alg B$ be a subalgebra or quotient of $\alg A$. 
    We choose a prime number $p>|A|$ and a $p$-ary cyclic term operation $t$ of $\alg A$.
    Using Theorem~\ref{thm:cyclic} and Proposition~\ref{prop:TMexist} we find $s \in \Clo(\alg A)$ such that $B$ 
     together with the corresponding term operation $s^{\alg B}$ of $\alg B$ is minimal Taylor.
    Then the cyclic composition $h$ of $t$ and $s$ is a cyclic operation on $\alg A$ and the corresponding $h^{\alg B}$ coincides with $s^{\alg B}$. 
    Since $\alg A$ is minimal Taylor, we have 
    $\Clo(\alg A)= \Clo(A;h)$ 
    and therefore 
    $\Clo(\alg B) = \Clo(B;h^{\alg B}) = \Clo(B; s^{\alg B})$, 
    which completes the proof.
\end{proof}

The following, additional, proposition is proved in a similar way.

\begin{proposition}\label{prop:cyclic_minors}
    Any term operation of a minimal Taylor algebra $\alg A$ can be obtained by identifying and permuting coordinates (and adding dummy coordinates) of a cyclic term operation of $\alg A$.
\end{proposition}

\begin{proof}
    Since $\alg A$ is minimal Taylor, $\Clo(\alg A) = \Clo(A;t)$ for any cyclic operation $t \in \Clo(\alg A)$ (which exists by Theorem~\ref{thm:cyclic}). 
    The claim now follows by noting that the star composition of  cyclic operations is a cyclic operation and that, since $t$ is idempotent, any term operation defined by a term in the symbol $t$ can be defined by star composing $t$ multiple times and then permuting and identifying coordinates.
\end{proof}

\subsection{Absorption}

We begin by proving an auxiliary lemma.

\begin{lemma}\label{almostBinaryAbsorption}
    Suppose $\alg A$ is a minimal Taylor algebra, 
    $\emptyset \neq B\subsetneq C\subseteq A$, and
    $\Sg_{\algA}(C^{n}\setminus B^{n})\cap 
    B^{n} = \varnothing$ for every $n$.
    Then 
    for every $f\in \Clo_{n}(\alg A)$ and every essential coordinate $i$ of $f$ we have 
    $f(\tuple a)\notin B$ whenever 
    $\tuple a\in C^{n}$ is such that $a_i\in C\setminus B$.
\end{lemma}

\begin{proof}
    Any cyclic term operation satisfies the required property (by using the compatibility with $\Sg_{\algA}(C^{n}\setminus B^{n})$ on cyclic permutations of $\tuple{a}$) and the property is stable under identifying and permuting coordinates (and introducing dummy ones).
    The claim now follows from Proposition~\ref{prop:cyclic_minors}.
\end{proof}

\XXXabssub*

\begin{proof} \hypertarget{ProofThm55}{}
Let $f$ be a witness for $B$ absorbing $\alg A$ and assume, for a contradiction, that 
$B$ is not a subuniverse. 

Let $\alg{A}'$ be a reduct of $\alg A$ 
with all the operations from  $\Clo(\alg A)$ that preserve $B$. 
Since $\alg{A}'$ is a proper reduct, it is not a Taylor algebra, and so some quotient of a subalgebra of $\alg{A}'$ is a two-element algebra such that every operation restricts to a projection (see Proposition~\ref{prop:taylor-without-powers}).
Let $B_{0}$ 
and $B_{1}$ be the congruence classes in this quotient,
clearly
every operation $t$ from $\Clo(\alg A')$ 
acts like a projection on $\{B_{0},B_1\}$.
Note that $f$ preserves $B$, therefore it has this property and we assume, without loss of generality, that $f$ acts like the first projection on $\{B_0,B_1\}$.
 
It follows from the previous paragraph that, for every $n$, the relation $S_{n} = (B_{0}\cup B_{1})^{n}\setminus 
B_{0}^{n}$ (just like any other relation ``built'' out of the blocks $B_0$ and $B_1$) is compatible with every operation  from $\Clo(\alg A')$ and is thus pp-definable from 
$\Inv(\alg A)$ and $B$.
Let $T_{n}$ be the relation 
defined by the same pp-definition 
with each conjunct $B(x)$ replaced by the void $A(x)$. 
Since $B$ absorbs $\alg A'$ by $f$, we also know that  
$S_{n}$ absorbs $T_{n}$ by the same $f$ (see Lemma~\ref{lem:absorption_and_pp}) 

Suppose that $T_{n}\cap B_{0}^{n}\neq\emptyset$ and 
choose $\tuple a\in T_{n}\cap B_{0}^{n}$
and $\tuple b\in B_{1}^{n}$.
Then $f(\tuple a,\tuple b,\dots,\tuple b)$ belongs to $B_{0}^{n}$ because $f$ 
acts like the first projection on $\{B_0,B_1\}$, and it also belongs to $S_{n}$
because $S_{n}$ absorbs $T_{n}$ by $f$.
This contradiction shows that 
$T_{n}\cap B_{0}^{n}=\emptyset$ for every $n$.
Note that $S_{n}\subseteq T_{n}$ and $T_{n}$ is a subpower of $\algA$, 
hence 
$\Sg_{\algA^n}((B_{0}\cup B_{1})^{n}\setminus B_{0}^{n})\cap 
B_{0}^{n}=\emptyset$ for every $n$.
We apply Lemma~\ref{almostBinaryAbsorption} 
and consider the behaviour of $f$:
it preserves $B_0\cup B_1$ and has at least two essential coordinates~%
(since it is a witness for a nontrivial absorption)
and therefore  cannot act like a projection on $\{B_0,B_1\}$ --- this is a contradiction.
\end{proof}

\XXXtwoabs*

\begin{proof} \hypertarget{ProofThm57}{}
  We begin with the implications that were discussed in the main part of the paper.
  The implications from  (d) to (c)
  and from (c) to (b) hold for all algebras.
  The first one is trivial, and the second one follows immediately from the definitions.
  The implication from (b) to (a) fails in a trivial clone, but holds in Taylor algebras.
  To see this we assume (b) and let $t$ be a cyclic term operation of $\alg A$ of arity $2k+1$.
  Define $f(x,y)=t(x,\dotsc,x,y,\dotsc y)$ where $x$ appears exactly $k+1$ times. 
  It is easy to see that $f$ witnesses the 2-absorption of $B$.
  Indeed, take any $a\in A, b\in B$ and let $\tuple c_1 = (a,\dotsc,a,b,\dotsc,b)$ with exactly $k+1$ $a$'s.
  Let $\tuple c_2$ be a cyclic shift of $\tuple c_1$ by $k$ positions and $\tuple c_3$ by $k+1$ positions. 
  Clearly $(t(\tuple c_1), t(\tuple c_2), t(\tuple c_3))\in R$ and, as it is a constant tuple, we conclude that $t(\tuple c_1) = f(a,b)\in B$.
  The case of $f(b,a)$ is similar, so we are done showing that (b) implies (a).
  
  For the implication from (a) to (d) 
    let $g$ be a binary operation witnessing the $2$-absorption of $B$ into $A$
    and let $t \in \Clo(\alg A)$ be a  cyclic operation of arity $p$. 
    Define $h(x_1,\dotsc,x_p)$ 
    as
    \begin{equation*}
    g(\dotsb g(g(x_1,x_2),x_3),\dotsc x_p)
    \end{equation*}
    and note that $h(a_1,\dotsc, a_p)\in B$ whenever at least one of the $a_i$ is in $B$.
    The same property holds for the cyclic composition, call it $s$, of $t$ and $h$.
    To see this, take a $p$-tuple $\tuple a$ which, in some position, has an element from $B$.
    In the subterms of $s$ the operation $h$ is applied to cyclic shifts of $\tuple a$, and
    every time the result is in $B$.
    Then $t$ is applied to elements from $B$ and,
    by Theorem~\ref{thm:AbsIsSubuniv},
    the result is in $B$ as well.
    The operation $s$ is cyclic and therefore generates the whole clone of $\alg A$. 
    Thus if any variable appearing in a term $f$ build from $s$ (in particular, a variable corresponding to an essential coordinate of $f^{\alg A}$) is an element of $B$, then the whole term evaluates to an element of $B$. 
\end{proof}

\PropertiesBinAbs*

\begin{proof} \hypertarget{ProofProp59}{}
    For (1) consider the result of applying an operation $f \in \Clo(\alg A)$ to a tuple $\tuple{a}$. If $a_i \in C$ for all of the essential coordinates $i$ of $f$, then the result is in $C$ (as $C \leq \alg A$), and if $a_i \in B$ for an essential $i$, then the result is in $B$ by item (d) in Theorem~\ref{thm:bin_abs}. 
    
    The argument for (2.a) is similar. 
    From (1) we know that $B\cup C$ is a subuniverse. For the absorption note that if a tuple $\tuple{a}$ has all but one entry in $B \cup C$, then $f(\tuple{a}) \in B$ in the case where $a_i \in B$ for some essential $i$, or $f(\tuple{a}) \in C$ in the other case (as $C \abs \alg A$ by $f$ and the inessential coordinates can be substituted with elements of $C$).
    For (2.b) first note that $f$ has at least two essential coordinates. Then $B \cap C \abs \alg{A}$ follows again from item (d) in Theorem~\ref{thm:bin_abs}, and $B \cap C$ is nonempty since it contains $f(c, \dotsc,c, b, c, \dotsc, c)$ for any $b \in B$ and $c \in C$ (where $b$ is at an essential coordinate).

To prove (3) note that, by strong projectivity,  both absorptions $\alg C \abs_2 \alg B \abs_2 \alg A$ can be witnessed by the same operation $f$ (and, in fact, any binary term operation of $\alg A$ such that both coordinates are essential can be taken as a witness) and then $f(f(f(x,y),x),f(f(y,x),y))$ witnesses $\alg C\abs_2\alg A$

The first part of item (4) follows from (2.b) as, again, all 2-absorptions have a common witness. The second part follows from (3). 
\end{proof}

We are moving on to the proof of Theorem~\ref{thm:center_abs}. 
The argument is similar in spirit to Theorem~\ref{thm:bin_abs} but the reasoning is a bit more technical. The following notion, which can be thought of as a common generalization of projectivity and absorption, will be useful.

\begin{definition}\label{def:gabs}
Let $f: A^n \to A$, let $B \subseteq A$, and let $g: \{0,1\}^n \to \{0,1\}$ be monotone. 
We write $f \sssB g$  if $f(\tuple{a}) \in B$ whenever $g(\tuple{x}) = 1$, where $\tuple{x}$ is the characteristic tuple of $\tuple{a}$ with respect to $B$~(i.e. $x_i = 1$ iff $a_i \in B$). 
\end{definition}

\noindent
For instance, for $g$ defined by $g(\tuple{x})=1$ iff $\sum x_i \geq n-1$, we have $f \sssB g$ iff $B$ absorbs $A$ by $f$; if $g$ is constant 0, then $f \sssB g$ for any $f$ and $B$.

We begin with a basic property of the relation $\sssB$.

\begin{proposition}\label{prop:monotonesim}
    Fix $A$ and $B\subseteq A$. If $f\sssB g$ are $n$-ary and $f_1\sssB g_1,\dotsc,f_n\sssB g_n$ are all $k$-ary, then $f(f_1,\dotsc,f_n)\sssB g(g_1,\dotsc,g_n)$.
\end{proposition}
\begin{proof}
    Let $f,f_1,\dotsc,f_n$ and $g,g_1,\dotsc,g_n$ be as in the statement of the proposition. 
    Take any $\tuple a\in A^k$ and
    let $\tuple x$ be its characteristic $k$-tuple~(with respect to $B$).
    Assume that $g(g_1,\dotsc,g_n)(\tuple x) = 1$ and note that if $g_i(\tuple x) =1$ then $f_i(\tuple a)\in B$.
    Therefore the characteristic $n$-tuple of $(f_1(\tuple a),\dotsc,f_n(\tuple a))$
    is above the tuple $(g_1(\tuple x),\dots,g_n(\tuple x))$ and since $g$ applied to the latter tuple outputs~$1$, it outputs~$1$ on the former as well.
    By the definition of $\sssB$, this implies that $f(f_1(\tuple a),\dotsc,f_n(\tuple a))\in B$, and the proposition is proved.
\end{proof}

Note that the previous proposition is especially useful
when the algebra $\alg A$ is generated by a single operation.
More formally if $B\subseteq A$ and $f\sssB g$, then any $f'\in\Clo(A;f)$ is $\sssB$ related to some $g'\in\Clo(\{0,1\};g)$.
Indeed, for any fixed $k$, it suffices to start with $\proj^k_i \mbox{(on $A$)} \sssB \proj^k_i$ \mbox{(on $\{0,1\}$)}, for all $i\leq k$,
and apply the previous proposition as many times as needed.
This is in fact exactly how we prove ``(a) implies (d)'' in Theorem~\ref{thm:center_abs} -- the proof of this implication is extracted to the following theorem.

\begin{theorem}\label{thm:tomajclone}
    Let $\alg A$ be a minimal Taylor algebra. There exists an arity preserving  map $f\mapsto f'$ from $\Clo(\alg A)$ to $\Clo(\{0,1\};\maj)$ such that:
    \begin{itemize}
        \item if a set $B$ $3$-absorbs $A$ then $f\sssB f'$ for every $f$ in $\Clo(\alg A)$, and
        \item for every prime $p>|A|$ there is a cyclic term of $\alg A$ of arity $p$ which is mapped to $\maj_p$.
    \end{itemize}
\end{theorem}
\begin{proof}
    For any prime $p > |A|$, we have a $p$-ary cyclic  operation, say $t$, in $\alg A$.
    If $t$ satisfies $t\sssB\maj_p$ then one can generate, as in discussion after  Proposition~\ref{prop:monotonesim}, a subdirect relation $\sssB$ between $\Clo(\alg A)$ and $\Clo(\{0,1\};\maj)$.
    Once this is done, we see that for any $f\in\Clo(\alg A)$ we can choose any operation $\sssB$-related to $f$ and fix it to be $f'$. 
    Such a map satisfies the conclusion of the theorem, for the set $B$ in question.
    If the initial condition, i.e. $t\sssB\maj_p$, held for a number of $B$'s then the conclusion will hold for all of them. 
    We will prove the theorem by induction on the size of the set of $B$'s.

Given a set $\mathcal B$ of 3-absorbing subsets of $A$ and a cyclic operation $t \in \Clo_p(\alg A)$, such that $t\sssB\maj_p$  for every $B \in \mathcal B$, we will find another cyclic term operation $s$ which will still work for any $B \in \mathcal B$ and, additionally, for a new 3-absorbing subset $C$ in $\alg A$. 
The claim will then follow by induction since we can start with the empty $\mathcal B$ and any $p$-ary cyclic term operation of $\alg A$~%
(in such a case we can choose an arbitrary map $f\mapsto f'$  sending $t$ to $\maj_p$).

Let $f\mapsto f'$ be a map compatible with $\sssB$ for all $B\in\mathcal B$
and such that $t\mapsto \maj_p$.
Let~$f$ be a witness for $C$ 3-absorbing $\alg A$, and let $f'$ be the image of $f$. 
Finally, let $h$ be the $p$-ary term operation of $\alg A$ defined from $f$ by the same term as a term defining $\maj_p$ from $\maj$ (the latter term exists since $\Clo(\{0,1\},\maj)$ is the clone of monotone selfdual operations -- see the discussion in Section~\ref{sec:mtalgebras}) and let $s$ be the cyclic composition of $t$ and $h$. Our aim is to verify that  $s\sssB\maj_p$ for every $B \in \mathcal{B} \cup \{C\}$.

Observe that $f \sssC \maj$ by the definition of absorption, therefore $h\sssC\maj_p$ by the definition of $h$ and Proposition~\ref{prop:monotonesim}. 
Further $s\sssC\maj_p$ by the definition of $s$ 
and the fact that $C \leq \alg A$~%
(which follows from Theorem~\ref{thm:AbsIsSubuniv}).
We thus further concentrate on the case $B \in \mathcal B$.
If $f'$~(the image of $f$ from the previous paragraph) is $\maj$, then  $h\sssB\maj_p$ 
and further $s\sssB\maj_p$ for all $B\in \mathcal B$
by the same reasoning, so this case is done.

Otherwise $f'$ is a projection (again by the structure of the majority clone) and Proposition~\ref{prop:monotonesim} implies that  $h\sssB \proj_i$ for some~$i$.
But then, again using Proposition~\ref{prop:monotonesim},
a $k$-shift of~$h$ (i.e. a term obtained by cyclically shifting the arguments of $h$ by $k$ positions) is $\sssB$ related to $\proj_{i+k\bmod p}$.
Finally $s\sssB\maj_p$ by the definition of $s$~(and, again, Proposition~\ref{prop:monotonesim}).
This finishes the proof of the theorem.
\end{proof}

We now prove a version of Theorem~\ref{thm:center_abs} formulated in terms of the relation $\sssB$. Note that in this formulation we reordered the items to better match the strength of the concepts.

\begin{theorem}
The following are equivalent for any minimal Taylor algebra $\algA$ and any $B \subseteq A$.
\begin{enumerate}
    \item[(a)] $B$ 3-absorbs $\algA$. 
    \item[(b)] $B$ is a center of $\alg A$.
    \item[(c)] The relation $R(x,y) = B(x) \vee B(y)$ is a subuniverse of $\algA^2$.    
    \item[(d)] for every $f\in\Clo(\alg A)$
    there exists $g\in\Clo(\{0,1\};\maj)$ such that $f \sssB g$.
\end{enumerate}
Moreover, if $B=\{b\}$, then these items are equivalent to 
\begin{enumerate}
    \item[(e)] $B$ absorbs $\alg A$.
\end{enumerate}
\end{theorem}
\begin{proof}
    We will begin by showing that
    (c) is equivalent to (d) in every algebra.
  For the implication from (d) to (c) take any $n$-ary operation $f$ of algebra $\alg A$,
  and let $g$ be such that $f\sssB g$.
  Take two tuples $\tuple a, \tuple a'$ such that $(a_1,a_1'),\dotsc,(a_n,a_n')\in R$.
  If $f(\tuple a)\in B$ then $(f(\tuple a),f(\tuple a'))\in R$ and we are done,
  so suppose $f(\tuple a)\notin B$.
  Let $\tuple x$ be the characteristic tuple of $\tuple a$; by the definition of $\sssB$ we have $g(\tuple x) = 0$.
  On the other hand, the definition of $R$ and the choice of $\tuple{a}, \tuple{a}'$ implies that $\tuple x'$ 
  -- the characteristic tuple of $\tuple a'$ -- is greater than or equal to  $\tuple 1-\tuple x = (1-x_1,\dotsc,1-x_n)$.
  Since $g$ is in $\Clo(\{0,1\};\maj)$, the clone of monotone selfdual operations, we get $g(\tuple{x}') \geq g(\tuple 1 - \tuple x) = \tuple 1 - g(\tuple x) = 1$, therefore $f(\tuple{a}') \in B$. 
  This proves the first implication.
  
  For the implication from (c) to (d), note that to every $n$-ary operation $f\in\Clo(\alg A)$ we can associate $g':\{0,1\}^n\rightarrow \{0,1\}$ by putting $g'(\tuple x) = 1$ if and only if $f(\tuple a)\in B$ for all $\tuple a$ with characteristic tuple greater than or equal to $\tuple x$.
  The operation $g'$ is monotone by definition,
  and clearly $f\sssB g'$.
  From $g'$ we define $g$ in the following way:
  for every $\tuple x$ we put $g(\tuple x) = g'(\tuple x)$ unless $g'(\tuple x) = g'(\tuple 1-\tuple x) = 1$ and $x_1 = 0$ -- in this case we put $g(\tuple x) = 0$.
  The operation $g$ is clearly monotone, 
  and we obviously have $f\sssB g$.
  We will show that $g$ is self-dual, which will conclude the proof of (d). 
  Indeed, if $g$ is not self-dual then we have $g'(\tuple x) = g'(\tuple 1-\tuple x) = 0$.
  From the definition of $g'$ we have two tuples $\tuple a$ and $\tuple a'$ with characteristic tuples greater than or equal to $\tuple x$ and $\tuple 1-\tuple x$ respectively and such that $f(\tuple a), f(\tuple a')\in A\setminus B$.
  Note that for every $i$ we have $R(a_i,a_i')$,
  and by compatibility $(f(\tuple a),f(\tuple a'))\in R$,
  but this contradicts the condition that both $f(\tuple a)$ and $f(\tuple a')$ are outside $B$.
    
  Next, we assume (d) and aim to prove (b).
  Again, no assumptions on the algebra are necessary.
  We set $\{0,1\}$ to be the universe of $\alg C$ 
  and for every basic operation of $\alg A$, say $f^{\alg A}$, put $f^{\alg C} = g$ 
  where $g$ is any operation such that $f\sssB g$. 
  Clearly $\Clo(\alg C)\subseteq \Clo(\{0,1\};\maj)$ and therefore $\alg C$ has no $2$-absorbing subuniverses.
  Moreover, we claim that the relation $R=A\times \{0\}\cup B\times \{1\}$ is a subuniverse of $\alg A\times\alg C$.
  Take $(a_1,y_1),\dotsc,(a_n,y_n)\in R$, any $n$-ary symbol $f$ 
  and let $\tuple x$ be the characteristic tuple of $\tuple a$.
  If $f^{\alg A}(\tuple a)\in B$ we are done, so suppose that $f^{\alg A}(\tuple a)\in A\setminus B$ and therefore $f^{\alg C}(\tuple x)=0$.
  Then $\tuple y\leq \tuple x$
  and so $f^{\alg C}(\tuple y) =0$,
  which completes the proof.
  
  For (b) implies (a) see Proposition~\ref{prop:center_absors_in_T}
  and for (a) implies (d) see Theorem~\ref{thm:tomajclone}.
  Finally (a) always implies (e).
  For the converse implication we will use Proposition~\ref{prop:getting3abs}.
  By assumption, $B = \{b\} \abs \alg A$. Take any $a\in A\setminus B$ and consider $\Sg_{\alg A^2}((a,b),(b,a))$.
  If $(a,a)$ is an element of this algebra, then there is a term operation $t$ satisfying $t(a,b) = a = t(b,a)$ and, by Proposition~\ref{ff:TMsubs}, the set $\{a,b\}$ is a subuniverse.
  Since $t$ acts on $\{a,b\}$ as a semilattice with absorbing element $a$, the subalgebra of $\alg A$ with universe $\{a,b\}$ (which is minimal Taylor by Proposition~\ref{prop:MinimalVariety})  is term-equivalent to the semilattice with absorbing element $a$, but then  $\{b\}$ can not be an absorbing subuniverse.
  This contradiction shows that the assumption of Proposition~\ref{prop:getting3abs} is satisfied
  and that $\{b\}$ $3$-absorbs $\alg A$ as required.
\end{proof}

Theorem~\ref{thm:center_abs} now easily follows, we just need to deal with the Taylor part in item (c).

\thmcenterabs*

\begin{proof} \hypertarget{ProofThm510}{}
The proof of the previous theorem shows that all of the items other than ``$B$ is a Taylor center of $\alg A$'' are equivalent to the witness of centrality in item (d). If $\Clo(\alg C) = \Clo(\{0,1\}, \maj)$, then $B$ is a Taylor center. In the other case, where $\Clo(\alg C)$ is the clone of projections, all of the term operations are $\sssB$-related to a projection, so $B$ is projective and then strongly projective by Theorem~\ref{thm:bin_abs}. 
We define an algebra $\alg D$ in the signature of $\alg A$ term-equivalent to the rock-paper-scissors algebra  by interpreting an $n$-ary symbol $f$ 
 in $\alg D$ using an arbitrarily chosen term involving all the essential variables of $f^{\alg A}$, e.g., $f^{\alg D}(d_1, \dots, d_n) = \winner(d_1,\winner(d_2, \dots )\dots)$ if all the coordinates of $f^{\alg A}$ are essential. 
The witness for $B$ being a Taylor center can then be taken  $R = (A \times \{\mbox{rock}\}) \cup (B \times D)$. Indeed, $\alg D$ has no nontrivial 2-absorbing subuniverse and $R$ is a subuniverse of $\alg A \times \alg D$: if $f$ is an $n$-ary symbol and $\tuple{a},\tuple{d}$ tuples such that $(a_i,d_i) \in R$ for each $i$, then either $a_i \in B$ for some essential coordinate $i$ of $f^{\alg A}$ and then $f^{\alg A}(\tuple{a}) \in B$ by strong projectivity,
or $a_i \not \in B$ for all the essential coordinates $i$ of $f^{\alg A}$, but then $d_i = \rock$ for all essential coordinates, therefore $f^{\alg D}(\tuple{d}) = \rock$ by idempotency (note that $f^{\alg A}$ and $f^{\alg D}$ have the same essential coordinates). 
In both cases $(f^{\alg A}(\tuple a),f^{\alg D}(\tuple{d})) \in R$. 
\end{proof}

Proposition~\ref{prop:3_abs_intersect} is a simple consequence of Theorem~\ref{thm:tomajclone}.

\propThreeabsintersect*

\begin{proof} \hypertarget{ProofProp512}{}
  Items (1) and (2) are proved using a cyclic operation $t$ of some odd arity $p$ such that $t\sssB\maj_p$ and $t\sssC\maj_p$.
  Such an operation is provided by Theorem~\ref{thm:tomajclone}.
  
  Note that $t$ generates the whole clone of $\alg A$. 
  For (2) and the first part of (1)  we apply $t$ to a tuple consisting of elements from $B\cup C$. The result will be from $B$ if the majority of the arguments are from $B$ and similarly for $C$.

  For the second part of (1) take $f(x,y,z)=t(x,\dotsc,x,y,\dotsc,y,z,\dotsc,z)$ where the number of $x$'s, $y$'s and $z$'s is roughly equal (more precisely, the number of $x$'s and $y$'s is at least $p/2$ and the same is true for the other two pairs). 
  By construction $f$ is a witness both for $\ B\abs_3\alg A$ and for $C\abs_3\alg A$, and thus for $B\cap C$ as well.

To show (3), recall that absorption is transitive (see Subsection~\ref{subsec:app_prelim_absorption}), so it is enough to verify that, for any $a \in A \setminus C$, the pair $(a,a)$ is not in $\Sg(C \times \{a\} \cup \{a\} \times C)$ and apply Proposition~\ref{prop:getting3abs}.
If $a \not\in B$, this follows from the fact that the relation $B(x) \vee B(y)$ is a subuniverse of $\algA^2$ (see Theorem~\ref{thm:center_abs}) and if $a \in B \setminus C$, this follows from the fact that the relation $C(x) \vee C(y)$ on $B$ is a subuniverse of $\algB^2$.
\end{proof}

\subsection{Edges}

\ThmMinimalEdgesQuotients*

\begin{proof} \hypertarget{ProofThm513}{}
  In case (a) there is a binary term operation acting like the semilattice operation on $\{a/\theta,b/\theta\}$ with top element $b/\theta$. 
  By Proposition~\ref{ff:TMsubs} together with Proposition~\ref{prop:MinimalVariety}, 
  the set $\{a/\theta,b/\theta\}$ is a subuniverse of $\alg E/\theta$ and thus is equal to it.
  By the classification of Post~\cite{post41:clones} and minimality of $\alg E/\theta$ we conclude (a).
  Case (b) is identical, except the operation is a ternary majority.
  In (c), Theorem~\ref{thm:FTA} implies that we have the Mal'cev operation $x-y+z$ for some abelian group $\alg G$, so $\alg E/\theta$ is term equivalent to the affine Mal'cev algebra of $\alg G$. 
   Since $E/\theta = \Sg_{\alg E/\theta}(a/\theta,b/\theta) = a/\theta + \{t \cdot (b/\theta-a/\theta): t \in \mathbb{Z}\}$ (where $+$ and $\cdot$ are computed in $\alg G$), the set $\{t \cdot (b/\theta-a/\theta): t \in \mathbb{Z}\}$ covers $G$, therefore the group homomorphism $\mathbb{Z} \to \alg G$ defined by $t \mapsto t (b/\theta-a/\theta)$ is surjective. We conclude using the first isomorphism theorem that $\alg G$ is isomorphic to $\mathbb{Z}/m$ for some $m$.

  For the last statement,
  suppose we have two different congruences $\theta$ and $\gamma$~%
  (put $\delta =\gamma\cap\theta$) witnessing the same semilattice~(resp. majority) edge.
  If we have a binary~(resp. ternary) operation acting as a semilattice operation $t$ with matching top elements~(or a majority operation)  on both $\{a/\theta,b/\theta\}$ and $\{a/\gamma,b/\gamma\}$, then $t$ acts as a semilattice~(majority) on $\{a/\delta, b/\delta\}$.
  But then, by Proposition~\ref{ff:TMsubs}, $E=a/\delta \cup b/\delta$
  which contradicts the fact that $\theta$ and $\gamma$ were different.

  It remains to obtain an operation acting as a semilattice on both semilattice edges. 
  In the semilattice case both $b/\theta$ and $b/\gamma$ are 2-absorbing subuniverses of $\alg E$ and, by Theorem~\ref{thm:bin_abs}, e.g. $t(x,y, \dots y)$ for a cyclic $t$ is such an operation.
  In the majority case the reasoning is identical except that all the sets $a/\theta, b/\theta, a/\gamma, b/\gamma$ are 3-absorbing
  and all these absorptions can be witnessed by a single term~%
  (e.g. the term obtained by identifying variables in a cyclic term as in the proof of Proposition~\ref{prop:3_abs_intersect}).
\end{proof}

\ProUniqueType*

\begin{proof} \hypertarget{ProofProp514}{}
  Let $\alg E=\Sg_{\alg A}(a,b)$ and let $\theta$ be a maximal congruence on $\alg E$ witnessing the minimal edge.
  Since the only simple affine Mal'cev algebras are (up to isomorphism) the affine Mal'cev algebras of $\zee{p}$ for a prime $p$, we have by Theorem~\ref{thm:minimaledgesquotients} that $\alg E/\theta$ is term-equivalent to a two-element semilattice, a two-element majority algebra, or an affine Mal'cev algebra of a group isomorphic to $\zee{p}$. 
  None of these algebras have a nontrivial subalgebra and so, by minimality, there is no congruence  incomparable with $\theta$.
  Indeed, if a congruence $\alpha$ is incomparable with the maximal congruence $\theta$, then at least one $\alpha$-class $B$  intersects two of the $\theta$-classes. Since $B/\theta \leq \alg E/\theta$ and  $\alg{E}/\theta$ does not have nontrivial subalgebras, then $B$ intersects all of the classes; in particular the classes containing $a$ and $b$. By minimality of the edge $(a,b)$, we get $B=E$, a contradiction.
  Similarly, since $\theta \subseteq \mu_E$, also $\mu_E = \theta$. Indeed, otherwise by definition of $\mu_E$ there is a proper subuniverse $B$ of $\alg E$ intersecting two of the $\theta$-classes, which leads to a contradiction as above.
\end{proof}

Proposition~\ref{prop:thin_semilattice} is an immediate consequence of the following proposition.

\begin{proposition}{\label{prop:inside_abs}}
    Let $\alg A$ be minimal Taylor algebra:
    \begin{enumerate}
        \item If $\alg B\abs_2\alg A$ and $a\in A\setminus B$
        there exists $b\in B$ such that $(a,b)$ is a semilattice edge and $\Sg_{\alg A}(a,b)=\{a,b\}$.
        \item If $\alg B\abs_3\alg A$ and $a\in A\setminus B$ there exists $b\in B$ such that:
        \begin{enumerate}
        \item $(a,b)$ is a semilattice edge and $\Sg_{\alg A}(a,b)=\{a,b\}$, or
        \item $(a,b)$ is a majority edge and $\Sg_{\alg A}(a,b)\cap B$ is one of the two equivalence classes of the congruence witnessing it.
        \end{enumerate}
    \end{enumerate}
\end{proposition}
\begin{proof}
    We will prove both items simultaneously. First we choose $b\in B$ such that $\Sg_{\alg A}(a,b)$ is minimal under inclusion in the set $\{\Sg_{\alg A}(a,b)\}_{b\in B}$.
    We let $\alg C = \Sg_{\alg A}(a,b)$ and $D = C\cap B$, and note that by our choice of $b$,
    for every $b'\in D$ we have $C = \Sg_{\alg A}(a,b')$.
    
    We put $R =\Sg_{\alg C^2}((a,b),(b,a))$. The first step is to show that if $D^2\cap R\neq\emptyset$ then $(b,b)\in R$ and $(a,b)$ is semilattice edge such that $\Sg_{\alg C}(a,b) = \{a,b\}$.
    Indeed, if $D^2\cap R \neq \emptyset$
    then $(D+R)\cap D\neq\emptyset$ and $a\in D+R$
    and consequently $D+R=C$.
    Further, if $D+R = C$ then there exists $b'\in D$ such that $(b',b)\in R$, but then $b-R$ contains both $a$ and $b'$ and thus $b-R = C$;
    in particular $(b,b)\in R$.
    We have shown that in this case there is a binary term operation acting on $\{a,b\}$ as a join-semilattice operation with top $b$. 
    By Proposition~\ref{ff:TMsubs}, $\{a,b\}$ is a subuniverse of $\alg A$,
    by Proposition~\ref{prop:MinimalVariety} the subalgebra with this subuniverse is a minimal Taylor algebra, which is clearly term equivalent to a two-element semilattice.
    
    Going back to the claims in the proposition, if $B\abs_2\alg A$, then a pair in $D^2\cap R$ is produced by a single application of the operation witnessing the 2-absorption to $(a,b)$ and $(b,a)$.
    This finishes the 2-absorption case.
    
    We are left with the case of $B\abs_3\alg A$ and  $R\cap D^2=\emptyset$.
    Note that $D \abs_3 \alg C$ and, by Theorem~\ref{thm:center_abs},
    $R\subseteq B(x)\vee B(y)$. Therefore $D+R = C\setminus D$ is a 3-absorbing subuniverse of $\alg C$.
    By item (2) in Proposition~\ref{prop:3_abs_intersect}, we conclude that the partition of $C$ into $D$ and $C\setminus D$ defines a congruence of $\alg C$ and the quotient modulo this congruence is term equivalent to a two-element majority algebra.
    This finishes the proof.
\end{proof}
  
\PropThinSemilattice*

\begin{proof} \hypertarget{ProofProp515}{}
  Let $(a,b)$ be as in the statement.
  Put $\alg A' = \Sg_{\alg A}(a,b)$ and $\theta$ a congruence witnessing the edge.
  The subalgebra $b/\theta$ 2-absorbs $\alg A'$,
  and Proposition~\ref{prop:inside_abs} provides $b'$ in $b/\theta$ such that $\Sg_{\alg A'}(a,b')=\{a,b'\}$.
  Since $(a,b)$ is a minimal edge, we get $b=b'$.
\end{proof}

\subsection{Absorption and edges}

The first major goal of this subsection is to prove the connection between absorption and stability under edges stated in Theorems~\ref{thm:bin_abs_stable} and \ref{thm:abs_stable}. 
The following observation will be often implicitly used in proofs: if $B$ is stable under semilattice (majority, abelian) edges in an algebra $\alg A$ and $\alg C \leq \alg A$, then $B \cap C$ is stable under semilattice (majority, abeialn) edges in the algebra $\alg C$. 

The first lemma  proves one direction in the mentioned theorems. 

\begin{lemma}\label{lem:absorption_implies_stable}
Let $\alg A$ be a minimal Taylor algebra and $B \subseteq A$.
If $B$ absorbs $\alg A$, then $B$ is stable under all the abelian and semilattice edges. Moreover, if $B$ 2-absorbs $\alg A$, then $B$ is stable under all the edges.
\end{lemma}

\begin{proof}
   Assume that $B$ absorbs $\alg A$ and $(b,a)$ is an edge with witnessing congruence $\theta$ of $\alg E := \Sg_{\alg A}(b,a)$ such that $b/\theta$ intersects $B$. By Theorem~\ref{thm:AbsIsSubuniv}, $B$ is a subuniverse of $\alg A$.
   As $B \abs \alg A$ we also have $B \cap E \abs \alg E$ and then $b/\theta \in (B \cap E)/\theta \abs \alg E/\theta$.

By Theorem~\ref{thm:minimaledgesquotients}, if $(b,a)$ is a semilattice edge, then $\alg E/\theta$ is term equivalent to the two element semilattice with absorbing element $a/\theta$.
By the description of term operations in the two element semilattice (see Subsection~\ref{subsec:boolean}), the only nontrivial absorbing subuniverse of $\alg E/\theta$ is $\{a/\theta\}$, therefore $(B \cap E)/\theta = E/\theta$, so each $\theta$-class indeed intersects $B$. 
If $(b,a)$ is an abelian edge, then $\alg E/\theta$ is term equivalent to an affine Mal'cev algebra. These algebras do not have any nontrivial absorbing subuniverses (since, e.g., the term operations are of the form $\sum a_ix_i$ where $\sum a_i = 1$ and these operations do not witness any nontrivial absorption) and we get $(B \cap E)/\theta = E/\theta$ in this case as well.
   If, additionally, $B$ 2-absorbs $\alg A$, then
   $(B \cap E)/\theta \abs_2 \alg E/\theta$ and we get the same conclusion by the description of term operations in the two-element majority algebra.
\end{proof}

Our next aim is to prove the other direction in Theorem~\ref{thm:bin_abs_stable}, that subsets stable under all the edges are 2-absorbing. As we will often work with minimal semilattice edges and sequences of such edges, the following terminology will be useful.

\begin{definition}
Let $\alg A$ be a minimal Taylor algebra. 
\begin{itemize}
    \item An \emph{s-edge} is a  minimal semilattice edge.
    \item An \emph{s-walk} is a sequence $a_1, \dots, a_k \in A$ such that $(a_i,a_{i+1})$ is an s-edge for every $i \in \{1, \dots, k-1\}$. (We allow the trivial walk with $k=1$.)
    \item A set $B \subseteq A$ is \emph{s-closed} if $a \in B$ whenever $(b,a)$ is an s-edge with $b \in B$. 
\end{itemize} 
\end{definition}

\noindent
Recall (Proposition~\ref{prop:thin_semilattice}) that for an s-edge $(a,b)$ in a minimal Taylor algebra $\alg A$, the set $\{a,b\}$ is a subuniverse of $\alg A$ (so the witnessing congruence for this edge is the equality relation on $\{a,b\}$). It follows that every set stable under semilattice edges is s-closed. 

Before proving Theorem~\ref{thm:bin_abs_stable} we give a useful criterion for s-closed subsets to be 2-absorbing. 

\begin{proposition}\label{prop:s-closed-bin} Suppose that $B$ is an s-closed subset of a minimal Taylor algebra $\alg{A}$. Then the following are equivalent:
\begin{itemize}
\item[(a)] $B$ 2-absorbs $\alg{A}$,
\item[(b)] for all $a \in A \setminus B$ and all $b \in B$, the algebra $\Sg_{\alg{A}}(a,b)$  has a nontrivial 2-absorbing subuniverse,
\item[(c)] for all $a \in A$ and all $b \in B$, there is a directed s-walk from $a$ to an element in $B$ which is contained in $\Sg_{\alg A}(a,b)$.
\end{itemize}
\end{proposition}

\begin{proof} To see that (a) implies (b), note that $B \abs_{2} \alg{A}$ implies that $B \cap \Sg_{\alg A}(a,b) \abs_{2} \Sg_{\alg A}(a,b)$. 

We prove that (b) implies (c) by induction on the size of $\alg A$ (assuming that the implication is true for all algebras of strictly smaller size). Let $\alg E = \Sg_{\algA}(a,b)$ and let $C$ be a nontrivial 2-absorbing subuniverse of $\alg E$. 
By Proposition~\ref{prop:inside_abs} there is an s-walk from $a$ to some $a' \in C$ contained in $E$: either $a \not\in C$ and there is an s-edge $(a,a')$ by that proposition, or $a \in C$ and we have the trivial s-walk $a$ from $a$ to $a$. If $a' \in B$, then we are done, so we assume $a' \not\in B$. 
Similarly, there is an s-walk from $b$ to some $b' \in C$ and, since $B$ is s-closed, we have $b' \in B$. 
Let $\alg A' = \Sg(a',b')$ and $B' = B \cap A'$. Note that $\alg A'$ is a proper subalgebra of $\alg A$ (since $A' \subseteq C \subsetneq E \subseteq A$), that $B'$ is an s-closed subset of $\alg A'$ (since $B$ is s-closed in $\alg A$ and $B' \subseteq B$), and that $\alg A'$ with $B'$ satisfies the condition of item (b). By induction hypothesis there is an s-walk from $a'$ to some element $b'' \in B'$ contained in $\Sg(a',b'') \subseteq E$. But we also have a walk from $a$ to $a'$ contained in $E$ and we are done by concatenating these two walks.

Now suppose that (c) holds. Let $\alg F$ be the subalgebra of $\alg A^{A^2}$ formed by all the binary term operations of $\alg A$ (recall Subsection~\ref{subsec:app_prelim_pp}) and let $\alg R \leq \alg A^X$ be the projection of $\alg F$ onto the set of coordinates $X = \{(a,b),(b,a): a \in A, b\in B\}$. 
Notice that any $f \in R \cap B^X$ witnesses that $B$ is a 2-absorbing subset of $\alg A$, therefore it is enough to show that $R$ intersects $B^X$. Assume that this is not the case, for the sake of proving a contradiction.

Let $Y \subseteq X$, $n \in \mathbb N$ be such that
\begin{itemize}
    \item $Y$ is maximal such that the projection of $R$ onto $Y$ intersects $B^Y$ (note that $Y \subsetneq X$ by the assumption), and 
    \item $n$ is the smallest number such that there exists a coordinate $\tuple{x} \in X \setminus Y$, there exists $f \in R$ with $f|_Y \in B^Y$, and there exists an s-walk $a_1, \dots, a_n$ in $\proj_{\tuple{x}} R$ with $a_1 = f(\tuple{x})$ and $a_n \in B$ (and fix such a pair $\tuple{x}$ and an s-walk). 
\end{itemize}
Note that the definition in the second item makes sense since $\proj_{\tuple{x}} R$ contains an element of $B$ (as witnessed by one of the two projection operations in $R$), so such an s-walk $a_1, \dots, a_n$ exists by (c). 
Now $\{a_2\}$ is a 2-absorbing subuniverse of $\Sg_{\alg A}(a_1,a_2) = \{a_1,a_2\}$ since $(a_1,a_2)$ is an s-edge. 
Therefore the subuniverse $R'$ of $\alg R$ obtained by fixing the coordinate $\tuple{x}$ of $R$ to $\{a_2\}$ 2-absorbs the subalgebra of $\alg R$ obtained by fixing the same coordinate to $\{a_1,a_2\}$. By Proposition~\ref{prop:inside_abs}, there exists an s-edge from $f$ to an element of $R'$, giving us an s-edge $(f,g)$ in $R$ for some $g \in R$ with $g(\tuple{x}) = a_2$. Since $(f,g)$  is an s-edge, the pair $(f(\tuple{y}),g(\tuple{y}))$ is, for any $\tuple{y} \in Y$, also an s-edge (witnessed by the same term). Recalling that $B$ is s-closed we obtain $g|_Y \in B^Y$. But $g(\tuple{x})$ starts the s-walk $a_2, \dots, a_n$, which is one step shorter than the original walk $a_1, \dots, a_n$, a contradiction to the minimality of $n$ (if $n>2$) or the maximality of $Y$ (if $n=2$). 
\end{proof}

The last lemma needed for the proof of Theorem~\ref{thm:abs_stable} provides some information about two-generated algebras which do not necessarily form edges. It is a weaker version of Theorem~\ref{thm:simple_non-edge_has_ternary_absorption}. 

\begin{lemma}\label{lem:baby_existence_of_ternary_absorption} If $\algA$ is a minimal Taylor algebra which is generated by two distinct elements $a,b \in A$, then either $\algA$ has a nontrivial abelian quotient, or $\algA$ has a nontrivial 3-absorbing subuniverse.
\end{lemma} 

\begin{proof}
Factoring by a maximal congruence we can assume that $\alg A$ is simple (as 3-absorbing subuniverses lift from quotients). If $\alg A$  has a ternary strongly functional subpower or a binary irredundant subdirect subpower, then we are done by Proposition~\ref{prop:ImpliesAbelian} or Corollary~\ref{cor:newAT} (since the subpower is then necessarily linked by the simplicity of $\alg A$). Also, if $(a,b)$ is an s-edge, we are done as well.

Otherwise, $\Sg_{\alg A^2}((a,b),(b,a))$ is not full (it does not contain $(b,b)$), so it must be a graph of a bijection -- an automorphism of $\alg A$. Moreover, by Theorem~\ref{ff:subdirect}, $\alg A$ has no proper irredundant subdirect subpowers, therefore $\Sg_{\alg A^3}((b,a,a),(a,b,a),(a,a,b))$ is full. In particular, it contains $(a,a,a)$, therefore there exists a term operation $t$ such $t(b,a,a)=t(a,b,a)=t(a,a,b)=a$. Since $\alg A$ has an  automorphism switching $a$ and $b$ we see that $t$ acts as a majority operation on $\{a,b\}$. Now $\{a,b\}$ is a subuniverse of $\alg A$ by  Proposition~\ref{ff:TMsubs}, therefore $A = \{a,b\}$ as $\alg A$ is generated by $a$,$b$. The majority operation witnesses that $\{a\}$ (and $\{b\}$) 3-absorbs $\alg A$. 
\end{proof}

\BinAbsIsStable*

\begin{proof} \hypertarget{ProofThm519}{}
 Lemma~\ref{lem:absorption_implies_stable} shows that (a) implies (b), and so we concentrate on the other direction. 
 
 By Proposition~\ref{prop:s-closed-bin}, it is enough to show that  for any $a \in A\setminus B$ and any $b \in B$, there is an s-edge from $a$ to an element of $\Sg(a,b) \cap B$. We may assume without loss of generality that $\Sg(a,b) = A$, and we will inductively assume that the theorem is true for algebras of size smaller than $|A|$.

If $\algA$ has a nontrivial affine quotient, then since $B$ is stable under abelian edges, $B$ must intersect the congruence class which contains $a$, so we can apply the inductive assumption to find an  s-edge from $a$ to $B$ which is contained in that congruence class. Otherwise, Lemma~\ref{lem:baby_existence_of_ternary_absorption} implies that there is some nontrivial  3-absorbing subalgebra $\alg{C}$ of $\algA$.

Next we show that $C$ intersects $B$. Suppose the contrary and apply Proposition~\ref{prop:inside_abs} to the element $b$ and $C \abs_3 \alg A$ -- we get an s-edge from $b$ to $C$, which is impossible since $B$ is s-closed, or a majority edge $(b,c)$ with $c \in C$ such that the witnessing congruence $\theta$ of $\alg E = \Sg_{\alg A}(b,c)$ has two equivalence classes, $C \cap E$ and $E \setminus C$, which is impossible since $B$  is stable under majority edges. (Note that we can get a stronger form of instability in the majority case as follows. 
We get $B \cap E \abs_2 E \setminus C$ by induction hypothesis, therefore $B \cap E \abs_3 \alg E$ by transitivity of 3-absorption stated in item (3) of Proposition~\ref{prop:3_abs_intersect}. Then the union of $B \cap E$ and $C \cap E$ is a subuniverse of~$\alg E$ by item (1) in Proposition~\ref{prop:3_abs_intersect}. Since $E$ is generated by $b \in B \cap E$ and $c \in C \cap E$, we obtain $E \setminus C = B \cap E$.)

Now we have $C \cap B \neq \emptyset$ and  we can apply the induction hypothesis to see that $C \cap B$ 2-absorbs $\alg{C}$. This implies that $C \cap B$ is a 3-absorbing subalgebra of $\algA$ by the transitivity of 3-absorption (item (3) in Proposition~\ref{prop:3_abs_intersect}). Thus we may assume without loss of generality that $C \subseteq B$.

By Proposition \ref{prop:inside_abs}, there is either a semilattice edge from $a$ to $c \in C$, or a majority edge from $a$ to $c \in C$ witnessed by the congruence on $\Sg(a,c)$ with classes $\Sg(a,c) \cap C$ and $\Sg(a,c) \setminus C$. In the first case we are done, so suppose we are in the second case. 
Since $\Sg(a,c) \cap C \subseteq B$ and $B$ is stable under majority edges, the congruence class $\Sg(a,c) \setminus C$ has a nonempty intersection with~$B$.  Then we may apply the induction hypothesis to $\Sg(a,c) \setminus C$ to see that there is a semilattice edge from $a$ to $(\Sg(a,c) \setminus C) \cap B$, which finishes the proof.
\end{proof}

Notice that the  proof of (b) implies (a) shows a stronger claim: it is enough to assume that $B$ is stable under abelian edges, it is s-closed, and that there is no (minimal) majority edge $(b,a)$, witnessed by a congruence $\theta$ on $\Sg(b,a)$, such that $b/\theta \subseteq B$ and $a/\theta \cap B = \emptyset$.

Our next project is to prove that (b) implies (a) in Theorem~\ref{thm:abs_stable}. 
The following lemma combines the relational description of strongly projective subuniverses in Proposition~\ref{prop:rel_descr_strong_blocking} and the sufficient condition for abelianess in Proposition~\ref{prop:ImpliesAbelian}.

\begin{lemma} \label{lem:symmetric_ternary_relation_and _strong_projectivity}
  Let $\alg A$ be a simple algebra and $R \leq_{sd} \alg A^3$ be a symmetric ternary relation whose projection to each pair of coordinates is full. For any $a \in A$ denote by $R_a$ the relation $R_a(x,y) \equiv R(a,x,y)$. Then either $\alg A$ is abelian or the set 
  $B = \{a \in A: R_a \mbox{ is linked } \}  $ is a nonempty strongly projective subuniverse of $\alg A$. 
\end{lemma}

\begin{proof}
Consider the ternary relation $S$ defined  by $S(x,y,z)$ if $y$ and $z$ are "left-linked" in $R_x$, that is, $(y,z)$ is in the relation $(R_x-R_x)+ \cdots + (R_x - R_x)$ with a sufficient number of summands.
Note that $S$ is pp-definable from $R$, so $S$ is a subuniverse of $\alg A^3$. Since the projection of $R$ onto each pair of coordinates is full, the relation $R_a$ is subdirect for every $a$. If $R_a$ is linked, then $S_a$ (defined analogously to $R_a$) is full. If $R_a$ is not linked, then it is the graph of a bijection $A \to A$ (since $\alg A$ is simple), so $S_a$ is the equality relation.
In summary, $S$ is equal to the relation $B(x) \vee (y=z)$. By Proposition~\ref{prop:rel_descr_strong_blocking}, $B$ is strongly projective.

If $B$ is empty, then every $R_a$ is the graph of a bijection. But then, by symmetry of $R$, fixing any coordinate to any $a$ and projecting out this coordinate gives the graph of a bijection. Therefore $R$ is strongly functional, so $\alg A$ is abelian by Proposition~\ref{prop:ImpliesAbelian}, finishing the proof.
\end{proof}

We are ready to prove the main tool for Theorem~\ref{thm:abs_stable}.
It will be useful to generalize the concept of a $B$-essential relation introduced in Subsection~\ref{subsec:app_prelim_absorption} to infinite powers: A subuniverse~$R$ of  $\alg A^X$ is \emph{$B$-essential} if $R$ does not intersect $B^{X}$
but every projection of $R$ onto all but one of the coordinates intersects the corresponding power of $B$. By Proposition~\ref{prop:abs_blockers}, $\{a\}$ absorbs $\alg A$ by a term operation of arity $n$ if and only if it does not have any $\{a\}$-essential subpower of arity $n$.

\begin{theorem} \label{thm:real_abs_stable}
Let $\alg A$ be a Taylor algebra and $a \in A$. Suppose that $\{a\}$ does not absorb $\alg A$ but does absorb every proper subalgebra of $\alg A$ that contains $\{a\}$. Then
\begin{itemize}
    \item $\alg A$ has a nontrivial abelian quotient  or
    \item there exists a nonempty projective subuniverse $B$ of $\alg A$ such that $a \not \in B$.
\end{itemize}
\end{theorem}

\begin{proof}
   First, observe that it is enough to prove the claim for simple $\alg A$. Indeed, if $\alpha$ is a maximal congruence of $\alg A$, then either $a/\alpha$ absorbs $A$ (in which case $\{a\}$ absorbs $\alg A$ by transitivity of absorption, a contradiction to the assumptions), or it does not (in which case we apply the simple case -- observe that projective subuniverses lift from quotients).
   Assume therefore that $\alg A$ is simple.
   
   We further assume that every relation $S \leq_{sd} \alg A^2$ whose left center contains $a$ is full. Indeed, if it is not,
   then by Corollary~\ref{cor:ATvariation} either $\alg A$ has a proper projective and 2-absorbing subuniverse $B$ or the left center of $R$ is an absorbing subuniverse of $\alg A$. By transitivity of absorption and the assumptions, no absorbing subuniverse of $\alg A$ can contain $a$. Therefore $B$ would  be a projective subuniverse not containing $a$ and the proof would be concluded.

   Next we observe that there exists a symmetric $\{a\}$-essential relation $R \leq \alg A^{\mathbb{N}}$ (here it will be convenient to use an infinitary relation). Indeed, for any arity $n$ there exists by Proposition~\ref{prop:abs_blockers} some $\{a\}$-essential relation, i.e., containing tuples $(b_1,a,a, \dots,a)$, $(a,b_2,a\dots, a), \dots$ and not containing $(a,a, \dots, a)$. Taking an element $b = b(n)$ that appears the most often among the $b_i$ and fixing (and projecting out) the other coordinates to $a$ we get an $\{a\}$-essential relation of arity at least $n/|A|$ containing $(b,a,a\dots)$, $(a,b,a, \dots)$, \dots. Finally, by taking $R = \Sg_{\alg A^{\mathbb N}}((b,a, a, \dots, ),(a,b,a, \dots), \dots)$, where $b$ is an element that appears infinitely many times among the $b(n)$, we get a symmetric $\{a\}$-essential relation.
   
   For a symmetric relation $R \leq \alg A^{\mathbb N}$, denote $B(R) = \{b \in A: (b,a,a, \dots) \in R\}$, and take a symmetric $\{a\}$-essential relation $R$ (i.e., $a \not \in B(R) \neq \emptyset$) such that $B(R)$ is maximal. Set $B= B(R)$. We will show that $R$ either gives us 
   a ternary subpower forcing strong projectivity or abelianess, or $R$ is a witness for $B$ being a projective subuniverse, i.e., after fixing any co-finite collection coordinates to $a$ and projecting them out, $R$ becomes $B(x_1) \vee B(x_2) \vee \dots \vee B(x_k)$ (see Proposition~\ref{prop:cubetermblocker}); observe here that the obtained relation is indeed a subpower of $\alg A$ despite the infinite arity of $R$. Since $a \not\in B$ we will be done.
   
   We claim that $\Sg_{\alg A}(a,b) = A$ for every $b \in B$. Indeed, otherwise the relation $R$ defined by $\Sg_{\alg A^{\mathbb N}}\{(b,a,a, \dots)$, $(a,b,a, \dots), \dots\}$ is an $\{a\}$-essential symmetric subuniverse of $(\Sg_{\alg A}(a,b))^{\mathbb{N}}$, so $\{a\}$ does not absorb $\Sg_{\alg A}(a,b)$ by Proposition~\ref{prop:abs_blockers} (consider, as above, the relations obtained by fixing co-finite collections of coordinates to $a$), contradicting the assumptions of the theorem.
   
   Next observe that the projection of $R$, call it $Q$, to any two distinct coordinates is full, even after fixing all but one of the remaining coordinates to $a$. Indeed, such a projection $Q$ contains  the pairs $(a,a),(a,b),(b,a)$ (for any $b \in B$). As $\Sg(a,b)=A$, the left center of  $Q$ contains $\{a\}$ so it is full by the assumption made in the second paragraph of the proof.
   
   Denote by $S$ the projection of $R$ onto two of the coordinates after fixing the rest to $a$, i.e., $S(x,y) \equiv R(x,y,a,a,a, \dots)$. Note that if we fix a coordinate of $R$ to a set $C$ and project onto the remaining coordinates, then we get a symmetric relation $R'$ with $B(R') = C + S$, which will be $\{a\}$-essential iff $a \not\in C+S \neq \emptyset$. 
   
   Consider the relation $R'$  obtained by fixing a coordinate of $R$ to $B+S$ and projecting to the remaining coordinates. Observe that $R'$ is symmetric and the set $B' := B(R') = B+S+S$ contains~$B$ (note that $S$ is symmetric) therefore, by the maximality of $B = B(R)$, either $B' = B$ or $a \in B'$. 
   
   If $B' (= B+S+S) = B$, then $S$ is not linked, so it is a graph of a bijection by simplicity of $\alg A$. We fix all but three arbitrarily selected coordinates of $R$ to $a$, project onto the three coordinates, and call $T$ the obtained subuniverse of $\alg A^3$. The relation $T$ has full projection to the first two coordinates (as we argued above) and the binary relation obtained by fixing a coordinate to $a$ is the graph of a bijection. 
   Lemma~\ref{lem:symmetric_ternary_relation_and _strong_projectivity} now implies that either $\alg A$ is abelian or $\alg A$ has a nonempty strongly projective subuniverse that does not contain $a$. Since strong projectivity implies projectivity, we are done.

   If $a \in B'$, then from $a + S = B$ we see that there exists some $(b,b') \in (B \times B) \cap S$. By similar reasoning to the proof of Proposition~\ref{prop:inside_abs}, $S$ contains $B(x_1) \vee B(x_2)$ (it contains $(a,b')$ as well as $(b,b')$ so also $A \times \{b'\}$ since $\Sg\{a,b\} = A$; then for any $b'' \in B$ it contains $(b'',b')$ and $(b'',a$) so $\{b''\} \times A$ as well). Furthermore, $S$ is equal to $B(x_1) \vee B(x_2)$ since otherwise there is $(c,c') \in S$, $c,c' \not\in B$, but then we can fix in $R$ a coordinate to $c$ (and project onto the remaining ones): the obtained $R'$ has $B(R') = c+S$, which contains $B$ and $c'$ but does not contain $a$, a contradiction to the maximality of $B=B(R)$. 
   
   Pick any $b \in B$ and consider $T(x,y) = R(x,y,b,a,a, \dots)$. It contains $(a,b), (b,a), (a,a)$, therefore it is the full relation $A^2$ (as the left center of $T$ contains $a$). It follows that  after fixing in $R$ all but three coordinates to $a$ (and projecting them out), the resulting ternary relation contains $B(x_1) \vee B(x_2) \vee B(x_3)$. Similarly to the argument above, it cannot contain any other triple, such as $(c,c',c'')$, because we would  fix two coordinates to $(c,c')$ (and project them out) and get a relation $R'$ with larger $B(R')$ (containing $B$ and $c''$). 
   
   If the algebra $\alg A$ is minimal Taylor we can stop the proof here, because we have already obtained the subpower $B(x_1) \vee B(x_2) \vee B(x_3)$ (see Theorem~\ref{thm:bin_abs}).
   For general Taylor algebras, we can by induction obtain the subpower $B(x_1) \vee B(x_2) \vee \dots \vee B(x_n)$ for every $n$, in a completely analogous way. 
\end{proof}

Theorem~\ref{thm:abs_stable} is a simple consequence of the results proved so far.

\AbsIsStable*

\begin{proof} \hypertarget{ProofThm521}{}
That absorbing sets are stable under semilattice and abelian edges and that (a) implies~(b) follow from Lemma~\ref{lem:absorption_implies_stable}. That (b) implies (a) follows from Theorem~\ref{thm:real_abs_stable}. Indeed, if $\{b\}$ does not absorb $\alg A$ then we take a minimal subalgebra $\alg B$ of $\alg A$ such that $b \in B$ and $\{b\}$ does not absorb $\alg B$, and apply the theorem. Either $\alg B$ has a nontrivial abelian quotient, in which case $\{b\}$ is not stable under abelian edges, or $\alg B$ has a nonempty projective subuniverse $C$ that does not contain $b$. In the latter case $C \abs_2 \alg A$ by Proposition~\ref{prop:projective-implies-binabs} and there exists a (minimal) semilattice edge $(b,c)$ with $c \in C$ by Proposition~\ref{prop:inside_abs}, therefore $\{b\}$ is not stable under semilattice edges.
\end{proof}

The ``unified operations'' theorem, Theorem~\ref{thm:unifiededges}, is a simple consequence of the results proved so far as well. The following proposition implies a refined version discussed in Subsection~\ref{subsec:abs_vs_edges}.

\begin{proposition} \label{prop:cyclicInEZAlgebras}
 Let $t$ be an $n$-ary cyclic term operation of a minimal Taylor algebra $\alg A$ and $B \abs_3 \alg A$. Then for any $m \geq n/2$ and any $\tuple{a} \in A^n$ such that $a_1, \dots, a_m \in B$ we have $t(\tuple{a}) \in B$. 
\end{proposition}

\begin{proof}
Let $\tuple{b}$ be the tuple obtained by cyclically shifting the tuple $\tuple{a}$ by $n-m$ positions. Since $m \geq n/2$, each pair $(a_i,b_i)$ is in the relation $B(x_1) \vee B(x_2)$; and since $t$ is cyclic, we have $t(\tuple{a}) = t(\tuple{b})$. But $t$ is compatible with $B(x_1) \vee B(x_2)$ by Theorem~\ref{thm:center_abs}, therefore $t(\tuple{a}) \in B$. 
\end{proof}

\UnifiedEdges*

\begin{proof} \hypertarget{ProofThm523}{}
  Choose positive integers $n,k,l$ such that $n \equiv 1 \pmod{|A|!}$, $k \equiv 1 \pmod{|A|!}$, $2k + l = n$, and $2k \geq l$. Let $m=k$. 
  Note that $\alg A$ has a cyclic operation $t$ of arity $n$: every prime divisor $p$ of $n$ is greater than $|A|$ and thus there exists a cyclic term operation $t_p$ of arity $p$ by Theorem~\ref{thm:cyclic}. The operation $t$ can then be obtained by a star composition of the $t_p$s. 
  
  Define $f$ by
  $$
  f(x,y,z)  =      t(\underbrace{x,x, \dots, x}_{k \times},
         \underbrace{y,y, \dots, y}_{l \times },
         \underbrace{z,z, \dots, z}_{m \times }).
$$
Because $\alg{E}/\theta$ and $\alg A$ are both minimal Taylor algebras, the 
 operation $f$ satisfies the second and the fourth item by using the cyclicity of $t$ and Proposition~\ref{prop:cyclicInEZAlgebras}, and it satisfies the first and the fifth item by Theorem~\ref{thm:bin_abs}. Note that for these claims we actually only need $k+l, k+m, l+m \geq n/2$.
 
   For the third item recall item (c) in Theorem~\ref{thm:minimaledgesquotients} and note that a cyclic operation in an affine Mal'cev algebra of $\zee{q}$ is equal to $\sum_{i=1}^n ax_i \pmod{q}$ where $na \equiv 1 \pmod{q}$.
  By simple arithmetic we conclude that if the conditions on $k,l,m$ hold, then the ternary operation $f$ is $x-y+z$. Indeed, since $q$ divides $|A|!$, we get $k = m \equiv 1 \equiv n \pmod{q}$, $a \equiv na \equiv 1 \pmod{q}$, and $l = n-2k \equiv -1 \pmod{q}$, so
  $f(x,y,z) \equiv akx + aly + azm \equiv x-y+z \pmod{q}$.
\end{proof}

Now our aim is to prove a refined version of Lemma~\ref{lem:baby_existence_of_ternary_absorption} with the final goal of Theorem~\ref{thm:single_operation_generates}. The following lemma will be used to produce a proper 3-absorbing subuniverse containing a generator. 

\begin{lemma} \label{lem:producing-three-abs}
Let $\alg A$ be a minimal Taylor algebra generated by $a,b \in A$. Suppose that $C \abs_3 \alg A$ is nontrivial and 2-absorbs a subalgebra of $\alg A$ containing $a$ and a subalgebra of $\alg A$ containing $b$.  
Then $\alg A$ has a proper 3-absorbing subuniverse containing $a$ or $b$. 
\end{lemma}
\begin{proof}
   For simplicity, we will say that a set 2-absorbs $a$ (or $b$) if it 2-absorbs a subalgebra containing $a$ (or $b$).
   Let $S = \Sg_{\alg A^2}((a,b),(b,a))$. As a first step we obtain a nontrivial 3-absorbing subuniverse $D$ that 2-absorbs $a$ and $b$ and such that $S \cap (C \times D)$ is nonempty: take $D=C$ if $C+S = A$, or $D=C+S$ otherwise. Fix $(c,d) \in S \cap (C \times D)$ and let $f$ be a witness for $(c,d) \in S$, that is, $f(a,b) = c$ and $f(b,a) = d$. 
   Next observe that if $a$  is in $C$ or $D$, then the goal is reached. Suppose henceforth that $C \cup D$ does not contain $a$. 
   
   By Theorem~\ref{thm:tomajclone}, $f \sssC f'$ and $f \sssD f'$ for some monotone selfdual binary operation, i.e., $f'$ is one of the two projections. Suppose $f'$ is the first projection, so $f(C,A) \subseteq C$ and $f(D,A) \subseteq D$. From these inclusions, from $f(a,b), f(b,a) \in C \cup D$, and from $f(\{a,b\},C) \subseteq C$, $f(\{a,b\},D) \subseteq D$ (by strong projectivity from Theorem~\ref{thm:bin_abs}) it follows that $C \cup D \cup \{b\}$ absorbs the \emph{set} $E = C \cup D \cup \{a,b\}$ by $f$. It is therefore enough to show that $E$ is a subuniverse of $\alg A$ -- then $E = \Sg(a,b)$ and the proof will be concluded. 

   Let $t$ be a cyclic term operation of arity $p$, take $$s(x_1, \dots, x_p) = f(f(f( \dots f(x_1,x_2),x_3), x_4) \dots x_p),$$ and let $h$ be the cyclic composition of $t$ and $s$. The result of applying $s$ to a tuple $\tuple{e} \in E^p$ is, by the properties above, in $C \cup D$ whenever $\tuple{e}$ is not the constant tuple of $a$'s or $b$'s. We then have $h(\tuple{e}) \in t(C \cup D, \dots, C \cup D)$, which is by item (1) of Proposition~\ref{prop:3_abs_intersect} a subset of $C \cup D$. Since $h$ is a cyclic operation and it preserves $E$, we see that $E$ is a subuniverse of $\alg A$ by Proposition~\ref{ff:TMsubs}, which concludes the proof.
\end{proof}

\begin{theorem} \label{thm:simple_non-edge_has_ternary_absorption} If $\algA$ is a minimal Taylor algebra which is generated by two distinct elements $a,b \in A$, then either $\algA$ has a nontrivial abelian quotient, or at least one of $a$, $b$ is contained in a proper  3-absorbing subuniverse of $\algA$.
\end{theorem}

\begin{proof} 
Just like in the proof of Lemma~\ref{lem:baby_existence_of_ternary_absorption} we can assume that $\algA$ is simple and not abelian, and assume for contradiction that neither of the generators $a,b$ is contained in a proper absorbing 3-subuniverse. We also assume that $\algA$ has no proper 2-absorbing subuniverse, as otherwise Lemma~\ref{lem:producing-three-abs} gives a contradiction immediately.

First we will prove that every reflexive relation $S \leq_{sd} \alg A^2$ is either the equality or the full relation. Suppose not. Then, by simplicity of $\alg A$, $S$ is linked. By replacing $S$ by $S-S$, perhaps multiple times if necessary, we can further assume that $S-S = A^2$ while $S$ is still not full. Then there is some $c \in A$ such that $(a,c), (b,c) \in S$, and since $A = \Sg(a,b)$, we see that $S$ has a nontrivial right center $C$ and is still reflexive. 

We claim that $C$ 2-absorbs the subalgebra $\alg D$ of $\alg A$ with universe $a+S$ (which contains $a$ since $S$ is reflexive). By Proposition~\ref{prop:s-closed-bin}.
we just need to check that for any $d \in D \setminus C$ and any $c \in C$, $\Sg(d,c)$ has a proper 2-absorbing subuniverse. To see this, note that the binary relation $T = S \cap (A \times \Sg(d,c))$ is a subdirect subuniverse of $\algA \times \Sg(d,c)$ (since $c$ is in the right center of $S$) and $a$ is contained in the left center of $T$ (since $d,c \in a + S$). Since $d \not\in C$, we see that $T$ is a proper subset of  $A \times \Sg(d,c)$, so the left center of $T$ is a proper subalgebra of $\algA$. Thus if $\Sg(d,c)$ has no proper 2-absorbing subuniverse, then the left center of $T$ is a center of $\alg A$, and $a$ is contained in a proper 3-absorbing subuniverse  of $\algA$ by Proposition~\ref{prop:center_absorbs}, a contradiction. 
The same argument shows that $C$ 2-absorbs a subalgebra containing~$b$. Application of Lemma~\ref{lem:producing-three-abs} then gives us a contradiction, so we have proved that every reflexive subdirect subuniverse of~$\alg A^2$ is the equality or the full relation.

Observe next that $S = \Sg_{\alg A^2}((a,b),(b,a))$ is a graph of a bijection. Indeed, it cannot be full as otherwise $(b,b) \in S$ implies that $(a,b)$ is an s-edge and then $\{b\}$ 2-absorbs $A = \{a,b\}$. If it is proper and linked, then we can use the same argument as above with the little tweak that $D$ may contain $b$ (not $a$) since $S$ is not reflexive (but contains $(a,b)$). 

Finally, consider the subdirect symmetric 
relation
$R = \Sg_{\alg A^3}\{(a,a,b),(a,b,a),(b,a,a)\}.$
Since the projection $\proj_{12}(R)$ of $R$ onto the first two coordinates contains $(a,a), (a,b)$ and $\Sg(a,b)=A$, we see that $\proj_{12}(R)$ has a left center that contains $a$. The center cannot be proper since otherwise we get a nontrivial 2-absorbing subuniverse or the center of $\proj_{12}(R)$ is a Taylor center of $\alg A$ and we apply~Proposition~\ref{prop:center_absorbs} to get a contradiction. Therefore $\proj_{12}(R) = A^2$. If $(a,a,a) \in R$, then the same argument as in Lemma~\ref{lem:baby_existence_of_ternary_absorption} shows that $(a,b)$ is a majority edge, which gives us a proper 3-absorbing subuniverse containing $a$ (and also one containing $b$), a contradiction.
The binary relation $T(x,y) \equiv R(a,x,y)$ is thus subdirect (as the projection of $R$ onto any two coordinates is full) but not full (as $(a,a) \not \in T$). Observe that $T+S$ is reflexive but cannot be full (as otherwise $T= (T+S)+S$ is full as well), therefore it is the equality relation and hence $T$ is a graph of a bijection.
It follows that the strongly projective subuniverse $B$ from Lemma~\ref{lem:symmetric_ternary_relation_and _strong_projectivity} is neither empty (as $\alg A$ is not abelian) nor equal to $A$ (as $a \not \in B$). Theorem~\ref{thm:bin_abs} shows that $B$ 2-absorbs $\alg A$, which is a contradiction.
\end{proof}

While we do not have an example of a minimal Taylor algebra in which some pair (or its reverse) is not an edge, we do not have strong support for the nonexistence of such an example either. However, we conjecture the following. 

\begin{conjecture}\label{recursive-structure} If $\algA$ is a minimal Taylor algebra which is generated by two elements $a,b \in A$ such that neither $(a,b)$ nor $(b,a)$ is an edge, then there are proper 3-absorbing subuniverses $C, D \abs_3 \algA$ such that $a \in C$ and $b\in D$.
\end{conjecture}

Theorem~\ref{thm:single_operation_generates} follows from the following more general theorem. 

\begin{theorem} \label{thm:real_single_operation_generates}
Let $\algA$ be a minimal Taylor algebra and $f_1, ..., f_n$ any collection of term operations of $\algA$ such that for each edge $(a,b)$ with witnessing maximal congruence $\theta$ of $\Sg_{\alg A}(a,b)$, at least one~$f_i$ acts nontrivially on $\Sg_{\alg A}(a,b)/\theta$ (that is, does not act as a projection). 
Then $f_1, ..., f_n$ generate the clone of $\algA$.
\end{theorem}

\begin{proof} 
Let $\alg G$ be the reduct of $\algA$ with basic operations $f_1, ..., f_n$.
By Proposition~\ref{prop:taylor-without-powers} it is enough to show that for any $a,b \in G$, there is no two-element quotient $\Sg_{\alg G}(a,b)/\alpha$ such that each $f_i$ acts as a projection. Suppose for contradiction that there was such a pair $a,b$ and congruence $\alpha$ on $\Sg_{\alg G}(a,b)$, and choose $a,b$ such that $\Sg_{\alg G}(a,b)$ is inclusion minimal. Note that $a/\alpha$, $b/\alpha$, and $\Sg_{\alg G}(a,b)$ might not be subuniverses of $\algA$.

By Theorem~\ref{thm:simple_non-edge_has_ternary_absorption}, we see that either $\Sg_\algA(a,b)$ has a nontrivial  abelian quotient and then $\Sg_\algA\{a,b\}/\theta$ can be chosen of prime order by taking a maximal $\theta$ (by Theorem~\ref{thm:minimaledgesquotients}), or one of $a,b$ is contained in a proper 3-absorbing subuniverse of $\Sg_\algA(a,b)$.

If there is a congruence $\theta$ such that $\Sg_\algA(a,b)/\theta$ is affine of prime order, then by assumption some $f_i$ acts nontrivially on $\Sg_\algA(a,b)/\theta$, so there is some ternary term operation $p \in \Clo_3(\alg G)$ such that the restriction of $p$ to $\Sg_\algA(a,b)/\theta$ is Mal'cev (by the fact that the affine Mal'cev algebra of $\zee{p}$ has minimal clone which follows e.g. from the description of term operations of this algebra). Since $f_i$ acts as a projection on $\Sg_{\alg G}(a,b)/\alpha$, $p$ also acts as a projection on $\Sg_{\alg G}(a,b)/\alpha$. Suppose without loss of generality that $p$ does not act like third projection on $\Sg_{\alg G}(a,b)/\alpha$. Then we have
\[
p(a,a,b) \in a/\alpha \cap b/\theta.
\]
But then $\Sg_{\alg G}(p(a,a,b),b)/\alpha = \Sg_{\alg G}(a,b)/\alpha$ and
\[
\Sg_{\alg G}(p(a,a,b),b) \subseteq \Sg_{\alg G}(a,b) \cap b/\theta,
\]
contradicting the minimality of $\Sg_{\alg G}(a,b)$ (as the latter set does not contain $a$ because $\theta$ separates $a$ and $b$).

Now suppose that $a$ is contained in a proper  3-absorbing subalgebra $\alg{C}$ of $\Sg_\algA(a,b)$. We show that this absorption is witnessed by an operation in $\Clo(\alg G)$. 

If $C$ is a 2-absorbing subuniverse of $\Sg_\algA(a,b)$, then there is some s-edge going into $C$ (by Proposition~\ref{prop:inside_abs}), and then some $f_i$ acts nontrivially on this edge, so there is a binary  term operation $s$ of $\alg G$ which acts  as a semilattice operation on this edge. Since both coordinates of $s$ are essential and $C$ is a strongly projective subuniverse of $\Sg_\algA (a,b)$ by Theorem~\ref{thm:bin_abs}, this operation witnesses the 2-absorption $C \abs_2 \Sg_\algA (a,b)$.

If $C$ is not a 2-absorbing subuniverse of $\Sg_\algA(a,b)$, then $C$ is not stable under majority edges by Lemma~\ref{lem:absorption_implies_stable} and  Theorem~\ref{thm:bin_abs_stable}. Let $(c,d)$ with witnessing congruence $\beta$ be a majority edge such that $c \in C$ and $d/\beta \cap C = 
\emptyset$. Some $f_i$ acts nontrivially on $\Sg_{\algA}(c,d)/\beta$, so there is a ternary term operation $m$ of $\alg G$ that acts as the majority operation on $\Sg_{\algA}(c,d)/\beta$. We claim that $m$ witnesses $C \abs_3 \Sg_\algA(a,b)$.
Indeed, $m \sssC g$ for some $g \in \Clo(\{0,1\},\maj)$ by Theorem~\ref{thm:center_abs}, but $g$ cannot be a projection: e.g., if $g$ is the first projection, then $m(c,d,d) \in C$ (because $m$ is $\sssC$-related to the first projection) and $m(c,d,d) \in d/\beta$ (because $m$ acts as the majority operation on the majority edge), a contradiction. Therefore $g=\maj$ and $m \sssC \maj$ implies the claim.

In both cases, we see that there is some ternary operation $t \in \Clo(\alg G)$ which witnesses the 3-absorption $C \abs_3 \Sg_\algA(a,b)$, and we may suppose without loss of generality that $t$ acts like the first projection on $\Sg_{\alg G}(a,b)/\alpha$. Then we have
\[
t(b,a,a) \in b/\alpha \cap C.
\]
But then $\Sg_{\alg G}(a,t(b,a,a))/\alpha = \Sg_{\alg G}(a,b)/\alpha$ and
\[
\Sg_{\alg G}(a,t(b,a,a)) \subseteq \Sg_{\alg G}(a,b) \cap C,
\]
contradicting the minimality of $\Sg_{\alg G}(a,b)$ again, and concluding the proof.
\end{proof}

\SingleOperationGenerates*

\begin{proof} \hypertarget{ProofThm524}{}
   This is a consequence of Theorem~\ref{thm:real_single_operation_generates}.
\end{proof}

\subsection{Examples from Section~\ref{sec:mtalgebras}} \label{subsec:app_examples}

\ExAbsorptionDoesntImplyThreeAbsorption*

\begin{proof} \hypertarget{VerEx511}{}
The algebra is minimal Taylor because $m$ generates a minimal clone (see \cite{CsakanyMinimalClones}).
The set $C = \{0,1\}$ is an absorbing subuniverse of $\alg A$ as witnessed by the 4-ary operation $$m(m(m(x_1,x_2,x_3),x_2,x_4),x_3,x_4).$$ The set $C$ is not a center of $\alg A$ since for any potentially witnessing relation $R \leq_{sd} \alg A \times \alg B$
the subuniverse $D = 2+R \leq \alg B$ satisfies $m(D,B,B) \subseteq D$ (as $m(2,1,0)=2$) and $m(B,D,D) \subseteq D$ (as $m(0,2,2)=2$), so $m(x,y,y)$ witnesses that $D$ is a 2-absorbing subuniverse of $\alg B$.
\end{proof}

\ExMajorityAintSimple*

\begin{proof} \hypertarget{VerEx516}{}
That $\algA$ is Taylor follows from the fact that $g$ is symmetric. To show that $\algA$ is \emph{minimal} Taylor, we will show that every Taylor reduct $\algA'$ of $\algA$ has a ternary cyclic term $f$, and that every ternary cyclic term $f$ of $\algA$ generates the same clone as $g$.

Consider any Taylor reduct $\algA'$ of $\algA$, and note that the congruence $\alpha$ is also a congruence on $\algA'$. Each congruence class of $\alpha$ is a Taylor reduct of an abelian algebra of size $2$, so $\algA'$ has a ternary term $p$ which acts as the minority operation on the congruence classes of $\alpha$, and $\algA'/\alpha$ is a Taylor reduct of a majority algebra, so $\algA'$ has a ternary term $m$ which acts as the majority operation on $\algA'/\alpha$. Cyclically composing these terms, we get a ternary cyclic term $f(x,y,z) = p(m(x,y,z),m(y,z,x),m(z,x,y))$. Thus every Taylor reduct of $\algA$ has a ternary cyclic term $f$.

To constrain the set of possible ternary cyclic terms $f$ of $\algA$, we will use the fact that every relation of $\algA$ which is preserved by $g$ must also be preserved by $f$. To this end, we need to find interesting relations on $\algA$.

Note that the cyclic permutation $(0\;1\;2\;3)$ is an automorphism of $\algA$, so the same is true for any reduct of $\algA$ (equivalently, the graph of the permutation $(0\;1\;2\;3)$ is a binary relation on $\algA$ which is preserved by $g$). 
One can easily verify that the binary relation $\alg{S} = \Sg_{\algA^2}\{(0,1),(1,0)\}$ is equal to $\{(x,y) \mid x = y \pm 1 \pmod{4}\}$, and that there is a congruence $\theta$ on $\alg{S}$ with congruence classes $\{(x,y) \mid x = y+1 \pmod{4}\}$ and $\{(x,y) \mid x = y-1 \pmod{4}\}$. Additionally, the quotient $\alg{S}/\theta$ is isomorphic to the $2$-element minority algebra.

By the above, the operation $f$ must be compatible with the automorphism $(0\;1\;2\;3)$, the binary relation $\alg{S}$, and the congruences $\alpha$ and $\theta$. Additionally, $f$ must act as the minority operation on $\{0,2\}$ and on $\alg{S}/\theta$, and $f$ must act as the majority operation on $\algA'/\alpha$.

The above constraints on $f$ allow us to determine every value of $f$ in terms of the value of $f(0,0,1)$. To see this, note that since $f$ acts as the minority operation on $\alg{S}/\theta$, we can determine the value of $f(a\pm 1, b\pm 1, c\pm 1)$ from the value of $f(a,b,c)$. For instance, we have $f(1,1,0) = f(0,0,1) - 1 \pmod{4}$ and $f(0,2,1) = f(1,1,0) - 1 \pmod{4} = f(0,0,1) - 2 \pmod{4}$. This, together with the fact that $f$ is cyclic, allows us to determine the value of $f(a,b,c)$ for every triple $a,b,c$ not all of the same parity from the value of $f(0,0,1)$. If $a,b,c$ all have the same parity, then $f$ acts as the minority operation on $\{a,b,c\}$. Since $f$ acts as majority on $\algA'/\alpha$, we have $f(0,0,1) \in \{0,2\}$, so there are only two possibilities for $f$. If $f(0,0,1) = 0$, then $f$ is equal to $g$. Otherwise, if $f(0,0,1) = 2$, then we have $g(x,y,z) = f(f(x,x,f(x,y,z)),f(y,y,f(x,y,z)),f(z,z,f(x,y,z)))$, so $f$ and $g$ generate the same clone. 
\end{proof}

\ExAbelianAintEasy*

\begin{proof} \hypertarget{VerEx517}{}
The algebra $\alg{S} = \Sg_{\algA^2}\{(a,b),(b,a)\}$ has a congruence $\psi$ such that $\alg{S}/\psi$ is isomorphic to $(\mathbb{Z}/4, x-y+z)$.
Explicitly, the congruence classes of $\psi$ are $\{(a,b),(c,d)\}$, $\{(b,a),(d,c)\}$, $\{(a,d),(c,b)\}$,$\{(b,c),(d,a)\}$, and in this order they correspond to the elements $0,1,2,3$ of $\mathbb{Z}/4$. Abusing notation, we will identify the congruence classes of $\psi$ with the elements of $\mathbb{Z}/4$ in the remainder of the proof.

To prove that this example is minimal Taylor, first note that every pair of elements forms an abelian edge, so the same must be true in any Taylor reduct. Thus by Theorem \ref{thm:sm-free}, any Taylor reduct $\algA'$ of $\algA$ must have a Mal'cev term $q(x,y,z)$. The restriction of $q$ to any two-element affine subalgebra or quotient must be the minority operation. This fixes the restriction of $q$ to the sets $\{a,c\}$ and $\{b,d\}$, as well as the restriction of $q$ to the quotient $\algA/\theta$, and the fact that $\mathbb{Z}/4$ has only one Mal'cev term forces $q$ to act on $\alg{S}/\psi = \mathbb{Z}/4$ as $x-y+z$.

Similarly to the previous example, the above constraints on the term $q$ imply that $q$ is completely determined by the restriction of $q$ to the set $\{a,b\}$. For instance, we have $(q(a,b,d)$, $q(b,a,a))/\psi = q(0,1,3) = 2 = (a,d)/\psi$, so from $q(b,a,a) = b$ we can conclude that $q(a,b,d) = c$. Since $q$ is Mal'cev, the only undetermined values of $q$ with all inputs from $\{a,b\}$ are the values of $q(a,b,a)$ and $q(b,a,b)$, which determine each other. Thus $q$ is completely determined by the value of $q(a,b,a)$, which is either $b$ or $d$. If $q(a,b,a) = b$, then $q$ is equal to $p$. Otherwise, if $q(a,b,a) = d$, then we have $p(x,y,z) = q(x,q(x,q(x,z,y),q(y,z,y)),z)$, so $p$ and $q$ generate the same clone. 
\end{proof}

\ExStabilityCantBeSimplified*

\begin{proof} \hypertarget{VerEx520}{}
First we will show that every Taylor reduct of $\algA$ is term-equivalent to $\algA$. Since $\algA$ has no majority edges and no $\mathbb{Z}/2$ edges, and since every pair of elements forms an edge, the same must be true in any Taylor reduct, so any Taylor reduct must have a binary commutative term~$f$ by Theorem \ref{thm:m2-free}.

Since $\mathbb{Z}/3$ has only one idempotent binary commutative term operation, the restriction of $f$ to $\{0,1,2\}$ is given by $f(x,y) = 2x+2y \pmod{3}$. Similarly, the restriction of $f$ to $\{0,*\}$ is given by $f(0,*) = f(*,0) = *$. Finally, because the values of $f$ on $\algA/\theta$ are known, we have $f(1,*)/\theta = f(1,0)/\theta = 2/\theta$, so $f(1,*) = 2$, and similarly $f(2,*) = 1$. Thus $f$ is the same as the operation $\cdot$ displayed above.

To see that $\{*\}$ is not absorbing, note that by Theorem \ref{thm:center_abs} and Theorem \ref{thm:m-free}, if $\{*\}$ was absorbing, then $\{*\}$ would be 2-absorbing, and then by Theorem \ref{thm:bin_abs}, the binary operation $\cdot$ would witness the absorption. However, we have $1\cdot * = 2 \not\in\{*\}$, so $\{*\}$ is not 2-absorbing.
\end{proof}

\section{Proofs for Section~\ref{sec:omitting_types}: Omitting types}\label{sec:proofsomittingtypes}

This section  contains proofs of the theorems stated in Section~\ref{sec:omitting_types}.

\ThmAFree*

\begin{proof} \hypertarget{ProofThm61}{}
Combining the results of \cite{Barto14:local,Bulatov16:restricted} and \cite{Bulatov04:graph}, $\algA$ has bounded width if and only if no subalgebra of $\alg A$ has a nontrivial abelian quotient. The proof of Theorem 2.8 in~\cite{Kozik15:maltsev} shows that this is also equivalent to having wnu term operations of every arity $n \geq 3$. Thus, (ii), (iii), and (iv) are equivalent.
To show that (ii) implies (i), it suffices to observe that if $\algA$ contains an abelian edge $(a,b)$ and congruence $\theta$ witnesses that, then the algebra $\Sg_\algA(a,b)/\theta$ is abelian. 
Finally, to show that (i) implies (ii), assume that there is a subalgebra $\algB$ of $\algA$ and a congruence $\theta$ on $\algB$ such that $\algB/\theta$ is abelian. Then for any $a,b$ from different equivalence classes of $\theta$, the pair $(a,b)$ is an abelian edge.
\end{proof}

\ThmSFree*

\begin{proof} \hypertarget{ProofThm62}{}
First, to show that (i) implies (ii) let $B \abs_2 \algA' \leq \alg A$. If $B$ is a nonempty proper subset of~$A'$, then Proposition~\ref{prop:inside_abs} gives us a minimal semilattice edge in $\alg A'$ and thus in $\alg A$, a contradiction.

By Theorem 2.12. of~\cite{Berman10:varieties} the existence of a 3-edge term operation is equivalent to a so-called 3-cube operation and then (ii) implies (iii) follows from Theorem 4.5 in~\cite{kazda21:cube-terms} (see also \cite{kearnes17:cube-term-blockers}) which provides a 3-cube operation whenever $\Clo(\algA)$ is generated by a single ternary operation (which is the case by Theorem~\ref{thm:single_operation_generates}) and
no subalgebra of $\algA$ has a nontrivial projective subuniverse (which is the case by Proposition~\ref{prop:projective-implies-binabs}).

From~\cite{Berman10:varieties} it follows that (iii) implies (iv) -- it is proved there that the existence of an edge (or cube) term operation is equivalent to having few subpowers.

Finally, we show that (iv) implies (i). If $\algA$ contains a semilattice edge, then by Theorem~\ref{thm:minimaledgesquotients}(a) it also has a subalgebra term equivalent to a 2-element semilattice. It is known from \cite{Berman10:varieties,Idziak10:few} that a semilattice does not have few subpowers.
\end{proof}

\ThmMFree*

\begin{proof} \hypertarget{ProofThm63}{}
To see that (i) implies (ii) recall that 3-absorbing subuniverses and centers are the same by Theorem~\ref{thm:center_abs}, that 3-absorbing subuniverses are stable under abelian and semilattice edges by Theorem~\ref{thm:abs_stable}, and that 
2-absorbing subuniverses are exactly those stable under all the edges by Theorem~\ref{thm:bin_abs_stable}.
 By item (4) in Proposition~\ref{prop:bin_abs_intersect}, (ii) implies (ii'). The fact that every majority edge defines two disjoint 3-absorbing subuniverses shows that (ii') implies (i).
\end{proof}

Before proving Theorem~\ref{thm:m2-free}, we first verify one of the implications.

\begin{proposition} \label{prop:commutative}
If $\alg A$ is a minimal Taylor algebra without majority edges and $\zee{2}$-edges, then $\alg A$ has a commutative binary term operation.
\end{proposition}

\begin{proof}
By Lemma 4.4 in~\cite{Barto12:cyclic} it is enough to show that  for any $a,b\in\algA$ there is a term operation $f$ such that $f(a,b)=f(b,a)$, or in other words that the binary symmetric subpower $\Sg_{\alg{A^2}}((a,b),(b,a))$ of $\alg A$ contains a tuple of the form $(c,c)$. Observe that, conversely, any symmetric nonempty binary subpower of an algebra with a commutative binary term operation intersects the diagonal. 

Let $\alg A$ be a minimal counterexample to the proposition with respect to $|A|$ and choose $a,b \in A$ such that $S = \Sg\{(a,b),(b,a)\}$ does not intersect the diagonal.
Clearly $\Sg(a,b) = A$ by minimality.
If $\algA$ has a proper congruence $\theta$, then there is some $c \in \algA$ such that $S \cap (c/\theta)^2 \ne \emptyset$ (since $\alg A/\theta$ is not a counterexample), but then $S \cap (c/\theta)^2 \leq c/\theta^2$ is a binary nonempty symmetric subpower of $c/\theta$ that avoids the diagonal, a contradiction to the minimality of $\alg A$.

Suppose now that $\algA$ has a proper  3-absorbing subuniverse $B$. If $B \cap (B+S) \ne \emptyset$, then there exist $d,e$ such that $(d,e) \in S \cap B^2$ (take any $e \in B \cap (B+S)$ and then $d \in B$ such that $(d,e) \in S$), so since $d,e$ generate a proper subalgebra of $\algA$ we see by minimality that $\Sg((d,e),(e,d)) \le \alg{S}$ contains a diagonal element, a contradiction.
On the other hand, if $B \cap (B+S) = \emptyset$, then by item (2) in Proposition~\ref{prop:3_abs_intersect}  $E=B \cup (B+S)$ must be a subuniverse of $\algA$ with congruence $\theta$ corresponding to the partition $B, B+S$, such that $\alg E/\theta$ is a two element majority algebra. Any pair in different equivalence classes is then a majority edge, a contradiction.

It remains to deal with the case that $\alg A$ is simple and $\alg A$ has no nontrivial 3-absorbing subuniverse. It then follows
from Proposition~\ref{lem:baby_existence_of_ternary_absorption} that $\alg A$ is abelian, so $\alg A$ is term equivalent to an affine Mal'cev algebra of a group isomorphic to $\zee{p}$. By the absence of $\zee{2}$-edges, $p$ is odd. 
Moreover,  
Corollary~\ref{cor:newAT} implies that $S$ is the graph of a bijection $A \to A$ -- the graph of an automorphism of $\alg A$. Since $S$ is generated by $(a,b)$ and $(b,a)$, the automorphism has order $2$. But $|A|=p$ is odd, so $\alpha(c)=c$ for some $c$, therefore $S$ contains $(c,c)$, a contradiction. 
\end{proof}

\ThmMTwoFree*

\begin{proof} \hypertarget{ProofThm64}{}
By Proposition~\ref{prop:commutative}, (i) implies (iii), and
(iii) implies (iii') by minimality of the Taylor algebra.

That (iii') implies (i) follows from the fact that the two element majority algebra and the two element affine Mal'cev algebra both have no nontrivial binary operations.
\end{proof}

\ThmSMFree*

\begin{proof} \hypertarget{ProofThm65}{}
To prove that (i) implies (ii) assume for contradiction that $B$ is a nontrivial absorbing subuniverse of $\alg A' \leq \alg A$. By Corollary~\ref{cor:connected} there exist $b \in B$ and $a \in A'\setminus B$ such that $(a,b)$ or $(b,a)$ is a minimal edge. Since $\alg A'$ is $\types{sm}$-free, we get that $(b,a)$ (and $(a,b)$) is necessarily a minimal abelian edge; let  $\theta$ be a witnessing congruence of $\Sg(b,a)$. By Theorem~\ref{thm:abs_stable}, the set $B$ is stable under abelian edges, therefore each $\theta$-class intersects $B$, in particular there exists $a'$ in $B \cap a/\theta$. But then $\Sg(b,a') \subseteq B$ is strictly contained in $\Sg(b,a)$, contradicting the minimality of the edge.

Condition (ii) is the property \emph{HAF} from \cite{Barto15:malcev}. By Theorem~1.4 from the same paper $\algA$ has a Mal'cev term operation, so (iii) holds.

Finally, if $\algA$ has a semilattice or a majority edge $(a,b)$ witnessed by a congruence $\theta$, then the algebra $\Sg_\algA(a,b)/\theta$ (and thus $\alg A$) cannot have a Mal'cev term operation by Theorem~\ref{thm:minimaledgesquotients}. Therefore (iii) implies (i).
\end{proof}

\ThmASFree*

\begin{proof} \hypertarget{ProofThm66}{}
Let us show that (i) implies (iii'). Since $\algA$ has no abelian and semilattice edges, by Theorem~\ref{thm:abs_stable} every 1-element subset of $\algA$ is an absorbing subuniverse. By Theorem~\ref{thm:center_abs} every such subset is also 3-absorbing. Finally, by Theorem~\ref{thm:unifiededges} there exists a ternary operation $f$ witnessing all 3-absorptions. Since every singleton is a 3-absorbing subuniverse, $f$ is a majority operation.

Clearly (iii') implies (iii).

Since both a two-element semilattice and 
an affine Mal'cev algebra do not have a
near unanimity term operation, 
by Theorem~\ref{thm:minimaledgesquotients} (iii) implies (i).
\end{proof}

\section{Conclusion}
\label{sec:conclusion}

We have introduced the concept of minimal Taylor algebras
and used it to significantly unify,  simplify, and extend the three main algebraic approaches to the CSP -- via absorption, via four types of subalgebras, and via edges. We believe that the theory started in this paper will help in attacking further open problems in computational complexity of CSP-related problems and Universal Algebra. There are, however, many directions which call for further exploration.

First, several technical questions naturally arise from the presented results:
 Do every two elements of a minimal Taylor algebra form an edge?
 How can we characterize sets stable under affine and semilattice edges in a global way? Is it possible to characterize (3-)absorption in terms of edges? Does stability under other edge-types correspond to a global property? Is every minimal bounded width algebra a minimal Taylor algebra?
 Are the equivalent characterizations in Theorem~\ref{thm:m-free} equivalent to ``every subalgebra has a unique minimal absorbing (rather than 3-absorbing) subuniverse''?   

Second, an interesting question arises in connection to the enumeration project discussed in Subsection~\ref{subsec:follow-up}. Are all minimal Taylor algebras finitely related? Here an algebra is \emph{finitely related} if its relational clone of invariant relations is generated by finitely many relations. The question is already open for 3-element algebras, a positive answer would e.g. give us a concrete list of hardest tractable CSPs in terms of relations.

Third, both CSP dichotomy proofs~\cite{Bulatov17:dichotomy,Zhuk20:dichotomy} require and develop more advanced Commutator Theory~\cite{freese:CMcommutator,kearnes:shape} concepts and results, while in this paper we have merely used some fundamental facts about the basic concept, the abelian algebra. Is it possible to develop our theory in this direction as well, potentially providing sufficient tools for the dichotomy result? Also, is there a natural concept that would replace thin edges in Bulatov's approach?

Fourth, which of the facts presented in the paper have their counterpart for nonminimal Taylor algebras or even general finite idempotent algebras? Here we would like to mention Ross Willard's work (unpublished) that provides a generalization for some of the advanced facts in Zhuk's approach. 

Finally, there is yet another, older, and highly developed theory of finite algebras, the tame congruence theory started in~\cite{hobby88:tct}.
What are the connections to the theory initiated in this paper?

\section*{Acknowledgement}

The authors thank the referees for the many very good comments that helped to improve the proofs.

\printbibliography

\end{document}